\documentclass[11pt,english]{article}
\usepackage[table]{xcolor}
\usepackage{lineno}
\usepackage{amsmath,amssymb,amsthm}
\usepackage{dsfont}
\usepackage{multicol}
\usepackage[margin=1in]{geometry}
\usepackage{graphicx,color}
\usepackage{enumitem}
\usepackage{fullpage}
\usepackage[noblocks]{authblk}
\usepackage{tcolorbox}
\usepackage{babel}
\usepackage{wrapfig}
\usepackage{mdframed}
\usepackage{mathtools}
\usepackage{floatrow}
\usepackage{multirow}
\usepackage{tablefootnote}
\usepackage[leftcaption]{sidecap}
\usepackage{hhline}
\usepackage[ruled,linesnumbered,vlined]{algorithm2e}
\usepackage{tablefootnote}
\usepackage{thm-restate}
\usepackage{bbding}
\usepackage{pifont}
\usepackage{nicefrac}

\definecolor{Darkblue}{rgb}{0,0,0.4}
\definecolor{Brown}{cmyk}{0,0.61,1.,0.60}
\definecolor{Purple}{cmyk}{0.45,0.86,0,0}
\definecolor{Darkgreen}{rgb}{0.133,0.543,0.133}
\usepackage[colorlinks,linkcolor=Darkblue,filecolor=blue,citecolor=blue,urlcolor=Darkblue,pagebackref]{hyperref}
\usepackage[nameinlink]{cleveref}
\usepackage[colorinlistoftodos,prependcaption,textsize=tiny]{todonotes}

\newcommand{\atodoAfter}[1]{}
\newcommand{\aidea}[1]{}

\newcommand{\submit}[1]{}
\definecolor{BrickRed}{rgb}{.72,0,0}
\def\EMPH#1{\emph{\textcolor{BrickRed} {#1}}}
\newcommand{\brick}[1]{{\color{BrickRed}#1}}

\newif\ifdraft
\draftfalse

\newcommand{\namedref}[2]{\hyperref[#2]{#1~\ref*{#2}}}
\newcommand{\propref}[1]{\hyperref[#1]{property~(\ref*{#1})}}

\newlength{\Oldarrayrulewidth}

\newcounter{resultrow}
\crefname{resultrow}{Result}{Results}
\Crefname{resultrow}{Result}{Results}
\newtheorem{theorem}{Theorem}
\newtheorem{lemma}{Lemma}
\newtheorem{definition}{Definition}
\newtheorem{claim}{Claim}
\newtheorem{observation}{Observation}
\newtheorem{corollary}{Corollary}

\newtheorem{remark}{Remark}
\newtheorem{inductiveHypothesis}{Inductive Hypothesis}

\newcommand{\poly}{\mathrm{poly}}
\newcommand{\polylog}{\mathrm{polylog}}

\newcommand{\R}{\mathbb{R}}
\newcommand{\E}{\mathbb{E}}

\newcommand{\N}{\mathbb{N}}

\newcommand{\est}{{\rm est}}

\newcommand{\supp}{\mathrm{supp}}

\newcommand{\diam}{\mathrm{diam}}

\newcommand{\Parent}{\mathrm{Parent}}

\newcommand{\first}{{\mathrm{first}}}
\newcommand{\last}{{\mathrm{last}}}

\def\cA{\ensuremath{\mathcal{A}}}

\def\cC{\ensuremath{\mathcal{C}}}
\def\cD{\ensuremath{\mathcal{D}}}

\def\cF{\ensuremath{\mathcal{F}}}

\def\cH{\ensuremath{\mathcal{H}}}

\def\cP{\ensuremath{\mathcal{P}}}
\def\cQ{\ensuremath{\mathcal{Q}}}
\def\cR{\ensuremath{\mathcal{R}}}

\def\cT{\ensuremath{\mathcal{T}}}

\newcommand{\Texp}{\mathsf{Texp}}
\newcommand{\Exp}{\mathsf{Exp}}

\definecolor{forestgreen}{rgb}{0.13, 0.55, 0.13}

\SetCommentSty{mycommfont}

\def\eps{\varepsilon}

\DeclareMathAlphabet{\mathpzc}{OT1}{pzc}{m}{it}
\newcommand{\etal}{{\em et al. \xspace}}

\newlength{\dhatheight}

\newcommand {\ignore} [1] {}
\newcommand{\initOneLiners}{\setlength{\itemsep}{0.2pt}\setlength{\parsep }{0.2pt}\setlength{\topsep }{0.2pt}}

\title{Stochastic Embedding of Digraphs	into DAGs}
\author{Arnold Filtser \thanks{Email: \texttt{arnold.filtser@biu.ac.il}. This research was supported by the ISRAEL SCIENCE FOUNDATION (grant No. 1042/22).}\\Bar-Ilan University}
\date{}
\begin{document}
	\maketitle
	\begin{abstract}
		Given a weighted digraph $G=(V,E,w)$, a stochastic embedding into DAGs is a distribution $\mathcal{D}$ over pairs of DAGs $(D_1,D_2)$ such that for every $u,v$: (1) the reachability is preserved: $u\rightsquigarrow_G v$ (i.e., $v$ is reachable from $u$ in $G$) implies that $u\rightsquigarrow_{D_1} v$ or $u\rightsquigarrow_{D_2} v$ (but not both), and (2) distances are dominated: $d_G(u,v)\le\min\{d_{D_1}(u,v),d_{D_2}(u,v)\}$.
		The stochastic embedding $\mathcal{D}$ has expected distortion $t$ if for every $u,v\in V$,
		$$\mathbb{E}_{(D_1,D_2)\sim\mathcal{D}}\left[d_{D_1}(u,v)
		\cdot\mathds{1}[u\rightsquigarrow_{D_1}v]+d_{D_2}(u,v)
		\cdot\mathds{1}[u\rightsquigarrow_{D_2}v]\right]\le t\cdot d_G(u,v)~.$$
		Finally, the sparsity of $\mathcal{D}$ is the maximum number of edges in any of the DAGs in its support.
		Given an $n$ vertex digraph with $m$ edges, we construct a stochastic embedding into DAGs with expected distortion $\tilde{O}(\log n)$ and $\tilde{O}(m)$ sparsity, improving a previous result by Assadi, Hoppenworth, and Wein [STOC 25], which achieved expected distortion $\tilde{O}(\log^3 n)$.
		Further, we can sample DAGs from this distribution in $\tilde{O}(m)$ time.
	\end{abstract}
	\setcounter{secnumdepth}{5}
	\setcounter{tocdepth}{3}
	\begin{multicols}{2}
		\tableofcontents
	\end{multicols}
	\thispagestyle{empty}
	\newpage
	\pagenumbering{arabic}
	\section{Introduction}
	Metric embedding is a powerful algorithmic paradigm with numerous applications.
	The basic idea is to take a complicated metric space and embed its points into a simpler metric space, while approximately preserving the original structure.
	Then, an algorithmic problem can be solved in the simpler embedded space and then pulled back to the original space.
	Here we focus on low-distortion metric embeddings, where the goal is to approximately preserve pairwise distances between points.\footnote{Classically, we say that a metric embedding $f:X\rightarrow Y$, between metric spaces $(X,d_X)$ and $(Y,d_Y)$, has distortion $t$ if $\forall x,y\in X$, $d_X(x,y)\le d_Y(f(x),f(y))\le t\cdot d_X(x,y)$.} 
	There are a variety of celebrated results of this form; a partial list includes:
	embedding general metric spaces into Euclidean space \cite{Bou85,ABN11,BFT24},
	embedding high-dimensional Euclidean space into lower-dimensional Euclidean space (a.k.a. dimension reduction) \cite{JL84,AC09,LN17,NN19},
	embedding the shortest path metric of topologically restricted (undirected) graphs (e.g., planar, minor-free) into $\ell_1$ \cite{Rao99,GNRS04,CJLV08,AFGN22,KLR19,Fil25faces},
	embedding low-treewidth graphs (with additive distortion) \cite{FKS19,FL22tw},
	and embedding general graphs into their sparse/light subgraph (a.k.a. spanners) \cite{ADDJS93,CW18,FS20} (see also the book and survey \cite{NS07,ABSHJKS20}).

	There is a natural tension between the ``simplicity'' of the target metric space and the distortion that one can achieve. 
	A prominent example is the shortest path metric of trees. While tree metrics are very simple (indeed, many NP-complete problems are tractable on trees), embedding a simple metric space such as the shortest path metric of the $n$-cycle $C_n$ requires distortion $\Omega(n)$ \cite{RR98,Gupta01}. 
	Fortunately, there is a way to circumvent this barrier: stochastic metric embeddings (introduced by \cite{AKPW95}). 
	Here, instead of embedding points from $X$ to $Y$ with a worst-case guarantee, there is a distribution over metric embeddings into metric spaces in some metric family $\cF$ (e.g., tree metrics), such that every pairwise distance is preserved only in expectation.\footnote{Formally, a distribution $\cD$ over metric embeddings $f:X\rightarrow Y$, for $(Y,d_Y)\in\cF$, has expected distortion $t$ if $\forall x,y\in X$ it is dominating (i.e., $\forall f\in \supp(\cD)$, $d_X(x,y)\le d_Y(f(x),f(y))$), and $\E_{f\sim\cD}\left[d_Y(f(x),f(y))\right]\le t\cdot d_X(x,y)$. 
	} 
	Fakcharoenphol, Rao, and Talwar \cite{FRT04} (following \cite{Bar96,Bartal98}, see also \cite{Bartal04}) 
	showed that every $n$-point metric space stochastically embeds into a distribution over trees with expected distortion $O(\log n)$. 
	This result was very influential, enjoyed tremendous success, and has numerous applications.\footnote{A partial list of applications of stochastic embeddings include 
		clustering \cite{BGT12},
		network design \cite{KKMPT12,GKR00}, 
		solvers of linear systems \cite{CKMPPRX14},
		oblivious routing \cite{Racke08}, and much more.}
	It is possible to sample such an embedding in an online fashion \cite{IMSZ10,BFU20,BFT24}, dynamically maintain it \cite{FGH21}, and compute it in a distributed manner \cite{BEGL24}. 
	One can also embed an (undirected) graph into a distribution of its spanning trees \cite{EEST08,ABN08-focs,AN19}. 
	Additional follow-ups include stochastic embeddings of minor-free graphs into low-treewidth graphs with additive distortion \cite{CFKL20,FL21,FL22tw}, with multiplicative distortion \cite{CLPP23,CCCLPP25}, and stochastic embeddings of hop-constrained distances \cite{HHZ21,Fil21}.

	\EMPH{Digraphs} (directed graphs) are notoriously harder than undirected graphs and considerably less understood. 
	Indeed, their shortest path distance is not a metric, as it lacks the symmetry condition. 
	They do not admit sparse spanners.\footnote{Consider a complete bipartite digraph $\overrightarrow{K_{n,n}}$ where there are two subsets of vertices $L,R$ of size $n$ each, and the directed edges are $L\times R$. The only reachability-preserving subgraph is $\overrightarrow{K_{n,n}}$ itself with $n^2$ edges.\label{foot:Knn}} 
	Embeddings of digraphs into directed $\ell_1$ (see \cite{CMM06}) have been studied, and it is known that general digraphs require distortion $\tilde{\Omega}(n^{1/7})$ to be embedded into directed $\ell_1$ \cite{CK09} (compared with $O(\log n)$ for embeddings of undirected graphs into $\ell_1$ \cite{Bou85}). 
	On the positive side, planar digraphs embed into directed $\ell_1$ with distortion $O(\log^3 n)$ \cite{KS21} (see also \cite{SSS19}). 
	Reachability-preserving minors were also studied in \cite{GHP20} for both general and planar digraphs.
	
	Assadi, Hoppenworth, and Wein \cite{AHW25} recently introduced stochastic embeddings of digraphs. 
	As the directed analog of trees, \cite{AHW25} chose to use \EMPH{DAGs} (directed acyclic graphs). 
	This is a natural choice, as many problems are much easier on DAGs than on general digraphs. 
	For a digraph $G$, we use \EMPH{$u\rightsquigarrow_{G}v$} to denote that there is a path from $u$ to $v$ in $G$, or in other words, that $v$ is reachable from $u$ in $G$. 
	First, observe that it is impossible to get finite expected distortion by sampling a single embedding into a DAG. 
	Indeed, if $u$ and $v$ belong to the same \EMPH{SCC} (strongly connected component), in a DAG $D$, it is impossible that $u\rightsquigarrow_{D}v$ and $v\rightsquigarrow_{D}u$; thus, the distortion of one of the pairs will be $\infty$. 
	It therefore makes sense to embed a digraph into a distribution over pairs of DAGs $(D_1,D_2)$ such that for every pair $u,v$ with $u\rightsquigarrow_{G}v$, it holds that either $u\rightsquigarrow_{D_1}v$ or $u\rightsquigarrow_{D_2}v$. 
	Next, observe that we cannot hope that the DAGs in the support of the embedding will be sparse. Indeed, consider the complete bipartite digraph $\overrightarrow{K_{n,n}}$ (see \cref{foot:Knn}); one has to use $\Omega(n^2)$ edges.\footnote{\cite{AHW25} showed a much stronger lower bound. In fact, there is a digraph such that if one uses DAGs with $n^{2-o(1)}$ edges, $n^{1-o(1)}$ DAGs are required to preserve the reachability of all pairs. See \Cref{row:4} in \Cref{tab:DAGCover}.\label{foot:LBHesse}} 
	The third observation is that it is completely hopeless to embed into subgraphs. 
	Indeed, consider the directed $n$-cycle $C_n$, where there are directed edges $v_1\rightarrow v_2 \rightarrow v_3\rightarrow\dots\rightarrow v_n\rightarrow v_1$. 
	Then any DAG subgraph can satisfy reachability only for one of the pairs $(v_1,v_n),(v_2,v_1),\dots,(v_n,v_{n-1})$.
	
	\begin{definition}[Stochastic Embedding into DAGs]
		Given a weighted digraph $G=(V,E,w)$, a stochastic embedding into DAGs is a distribution $\cD$ over pairs $(D_1,D_2)$ such that $D_1,D_2$ are DAGs over $V$, and for every $u,v\in V$:
		\begin{itemize}
			\item Dominating distances: $d_G(u,v)\le\min\{d_{D_1}(u,v),d_{D_2}(u,v)\}$
			\item Reachability preservation: If $u\rightsquigarrow_{G}v$, then either $u\rightsquigarrow_{D_1}v$ or $u\rightsquigarrow_{D_2}v$ (but not both).
		\end{itemize}
		We say that $\cD$ has expected distortion $t$ if
		$$\forall u,v\in V\quad\E_{(D_1,D_2)\sim\cD}\left[d_{D_1}(u,v)
		\cdot\mathds{1}[u\rightsquigarrow_{D_1}v]+d_{D_2}(u,v)
		\cdot\mathds{1}[u\rightsquigarrow_{D_2}v]\right]\le t\cdot d_G(u,v)~.$$
		The sparsity of $\cD$ is $\max_{D\in\supp(\cD)}\{|E(D)|\}$, the maximum number of edges in any of the DAGs in the support of $\cD$.
	\end{definition}
	
	Assadi, Hoppenworth, and Wein \cite{AHW25} showed that every digraph with polynomial \EMPH{aspect ratio}\footnote{The aspect ratio of a weighted digraph $G=(V,E,w)$ is $\Phi=\frac{\max\left\{d(u,v)\mid u\rightsquigarrow_{G}v\right\}}{\min\left\{d(u,v)\mid u\rightsquigarrow_{G}v\right\}}$, the ratio between the maximum and minimum (finite) distances in $G$.\label{foot:aspect}} 
	stochastically embeds into DAGs with expected distortion $\tilde{O}(\log^3 n)$ and sparsity $\tilde{O}(m+n)$.\footnote{In the full arXiv version, \cite{AHW25} stated expected distortion $O(\log^3 n\cdot\log\log n)$ and sparsity $O((m+n)\cdot\log^2 n)$ (in the STOC proceedings the expected distortion was $O(\log^4 n)$).} 
	Further, w.h.p., \cite{AHW25} can sample from this stochastic embedding in $\tilde{O}(m)$ time. 
	The main result of this paper is a significant improvement in the expected distortion:
	
	\begin{theorem}[Main]\label{thm:mainExpectedDistortion}
		Given an $n$-vertex digraph $G=(V,E,w)$ with $m$ edges and aspect ratio$^{\ref{foot:aspect}}$ $\Phi$, there is a stochastic embedding into DAGs with expected distortion $O(\log n\cdot\log\log n)$ and sparsity $O((n+m)\log^2(n\Phi))$. 
		Further, for a digraph with polynomial aspect ratio, there is an algorithm running in $\tilde{O}(m)$ time that w.h.p. samples $(D_1,D_2)$ from the stochastic embedding above.
	\end{theorem}
	
	Note that the expected distortion in \Cref{thm:mainExpectedDistortion} is independent of the aspect ratio $\Phi$. One can actually apply the algorithm of \cite{AHW25} for arbitrary aspect ratio $\Phi$, obtaining expected distortion $\polylog(n\Phi)$ (and the same sparsity as in \Cref{thm:mainExpectedDistortion}). 
	Consider, for example, the case of quasi-polynomial ($2^{\polylog(n)}$), or even slightly exponential ($2^{n^\eps}$) aspect ratio. Then \Cref{thm:mainExpectedDistortion} can still be applied with expected distortion $\tilde{O}(\log n)$ and reasonable sparsity.

	\subsection{DAG Cover}
	Roughly speaking, in a stochastic embedding into trees we construct a collection of embeddings into trees such that the distortion of every pair is small on average. In contrast, in a tree cover we embed into a collection of trees such that the minimum distortion of every pair is small (while all the trees in the cover are dominating). 
	Tree covers were introduced by Arya \etal \cite{ADMSS95} in the context of Euclidean spaces and have been extensively studied ever since. 
	Every $n$-point metric admits a tree cover with distortion $2k-1$ and $O(n^{\nicefrac1k}\log n)$ trees \cite{TZ05}. In particular, distortion $O(\log n)$ with $O(\log n)$ trees. 
	Tree covers have also been studied in the context of doubling metrics \cite{CGMZ16,BFN22,FGN24,CCLMST24socg,CCLST25}, for planar and minor-free graphs \cite{BFN22,CCLMST23Planar,CCLMST24}, into spanning trees \cite{GKR04,ACEFN20}. There are also Ramsey trees (where every vertex has a tree in the cover approximating its shortest path tree) \cite{BLMN05,BLMN05b,MN07,NT12,ACEFN20,FL21}.
	
	Assadi \etal \cite{AHW25} introduced the notion of DAG cover:
	
	\begin{definition}[DAG Cover]
		Given a weighted digraph $G=(V,E,w)$, a DAG cover is a collection of DAGs $D_1,\dots,D_g$ over $V$ such that for every $u,v\in V$ and DAG $D_i$, $d_G(u,v)\le d_{D_i}(u,v)$. 
		The size of the DAG cover is $g$, the number of DAGs. 
		The DAG cover has distortion $t$ if for every $u,v\in V$ such that $u\rightsquigarrow_{G}v$ it holds that $\min_i d_{D_i}(u,v)\le t\cdot d_G(u,v)$. 
		The sparsity of a DAG cover is the total number of additional edges not in $G$: $\left|\cup_i E(D_i)\setminus E(G)\right|$.
	\end{definition}
	
	Assadi \etal \cite{AHW25} systematically studied DAG covers. See \Cref{tab:DAGCover} for a summary and explanation of their results. 
	By taking a union of the DAGs created by $O(\log n)$ independent samples from the stochastic embedding\footnote{Assadi \etal \cite{AHW25} used the fact that by Markov, when sampling $(D_1,D_2)$ from the stochastic embedding it holds that $\Pr_{(D_{1},D_{2})\sim\cD}\left[\min\left\{ d_{D_{1}}(u,v),d_{D_{2}}(u,v)\right\} >\tilde{\Omega}(\log^{3}n)\cdot d_{G}(u,v)\right]\le\frac{1}{2}$. Thus by taking the union of $O(\log n)$ samples, w.h.p. all pairs are satisfied.}, 
	Assadi \etal \cite{AHW25} obtained an $O(\log n)$-size DAG cover with distortion $\tilde{O}(\log^3 n)$ and sparsity $O((m+n)\cdot\log^3 n)$ (assuming polynomial aspect ratio, see \Cref{row:8} in \Cref{tab:DAGCover}). 
	By repeating the same argument on top of \Cref{thm:mainExpectedDistortion}, we conclude:
	
	\begin{corollary}[DAG Cover]\label{cor:DAGcover}
		Given an $n$-vertex digraph $G=(V,E,w)$ with $m$ edges and aspect ratio$^{\ref{foot:aspect}}$ $\Phi$, there is a DAG cover with $O(\log n)$ DAGs, distortion $O(\log n\cdot\log\log n)$, and sparsity $O((n+m)\log n\log^2(n\Phi))$. 
		Further, for a digraph $G$ with polynomial aspect ratio, there is an $\tilde{O}(m)$-time algorithm that w.h.p. constructs this DAG cover. 
	\end{corollary}

	\begin{table}[t]
		\centering
		\begin{tabular}{|l|l|l|l|l|l|}
			\hline
			\textbf{Result} & \textbf{Distortion} & \textbf{\# of DAGs} & \textbf{Sparsity} & \textbf{Ref} & \textbf{Notes} \\
			\hline
			\refstepcounter{resultrow}\label{row:1}(\theresultrow) & $1$ & $n$ & - & Trivial & Subgraph, UB \\
			\hline
			\refstepcounter{resultrow}\label{row:2}(\theresultrow) & Reachability & At least $n$ & - & Trivial & Subgraph, LB \\
			\hline
			\refstepcounter{resultrow}\label{row:3}(\theresultrow) & $1$ & $2$ & $O(n^2)$ & Trivial & UB \\
			\hline
			\refstepcounter{resultrow}\label{row:4}(\theresultrow) & Reachability & $n^{1-o(1)}$ & At most $n^{2-o(1)}$ & Thm. 1.2 & LB \\
			\hline
			\refstepcounter{resultrow}\label{row:5}(\theresultrow) & $1$ & $n^{o(1)}$ & $O\left(n^{2}\cdot\frac{(\log\log n)^{4}}{\log^{2}n}\right)$ & Thm. 1.3 & UB \\
			\hline
			\refstepcounter{resultrow}\label{row:6}(\theresultrow) & Reachability & $2$ & $O(m)$ & Obs. 1.4 & UB \\
			\hline
			\refstepcounter{resultrow}\label{row:7}(\theresultrow) & $1$ & $\Omega(n^{1/6})$ & $O(m^{1+\eps})$ & Thm. 1.6 & LB \\
			\hline
			\refstepcounter{resultrow}\label{row:8}(\theresultrow) & $\tilde{O}(\log^3 n)$ & $O(\log n)$ & $O((m+n)\log^3 n)$ & Thm. 1.5 & UB \\
			\hline
			\rowcolor{gray!20}
			\refstepcounter{resultrow}\label{row:9}(\theresultrow) & $\tilde{O}(\log n)$ & $O(\log n)$ & $O((m+n)\log^3 n)$ & \Cref{cor:DAGcover}  & UB \\
			\hline
		\end{tabular}
		\caption{\small Summary of known results on DAG cover. 
			Results (\ref{row:1})--(\ref{row:8}) are due to \cite{AHW25} (and the references are to theorems therein). 
			\Cref{row:1} is for DAG cover using subgraphs; it is an UB (upper bound) achieved by taking all the shortest path trees from all vertices. 
			\Cref{row:2} is also for a DAG cover using subgraphs where the requirement is only for reachability (finite distortion). The LB (lower bound) is achieved on the directed $n$-cycle $C_n$ with edges $v_1\rightarrow v_2 \rightarrow v_3\rightarrow\dots\rightarrow v_n\rightarrow v_1$, where it is argued that satisfying each of the pairs $(v_{i+1},v_i)$ requires a different DAG. 
			\Cref{row:3} is an easy UB achieving distortion $1$ using only $2$ DAGs, but each with $n^2$ edges. 
			The first DAG $D_1$ is constructed by taking an arbitrary order of the vertices $v_1,v_2,\dots,v_n$ and adding all the edges $(v_i,v_j)$ such that $i<j$ and $v_i\rightsquigarrow_{G}v_j$ with weight $d_G(v_i,v_j)$. 
			The second DAG $D_2$ is constructed in the same way w.r.t. the reversed order $v_n,v_{n-1},\dots,v_1$. 
			\Cref{row:4} is nontrivial, based on the digraph constructed in \cite{Hesse03} and mentioned in \cref{foot:LBHesse}. 
			It implies that with truly sub-quadratic sparsity, the number of DAGs in the cover has to be almost linear. 
			\Cref{row:5} demonstrates that one can slightly go beyond sub-quadratic sparsity while using a sub-polynomial number of DAGs. 
			\Cref{row:6} shows that for reachability, two DAGs of linear sparsity (in the number of edges) are sufficient. 
			In \Cref{row:7}, $\eps>0$ is an absolute constant. The result demonstrates that for exact distance preservation, using nearly linear sparsity will require a polynomial number of DAGs. 
			\Cref{row:8} is obtained by taking all the DAGs created in $O(\log n)$ independent samples from the stochastic embedding of \cite{AHW25}. 
			\Cref{row:9} is the same, but using \Cref{thm:mainExpectedDistortion} instead.
		}\label{tab:DAGCover}
	\end{table}

	\section{Technical Ideas} 
	In this section we present a simplified overview of the techniques and ideas in this paper, as well as in previous work.
	Brief and oversimplified summary: 
	Bartal \cite{Bar96} constructed stochastic tree embeddings with expected distortion $O(\log^2 n)$ using LDDs, and later \cite{Bartal98} improved this to expected distortion $\tilde{O}(\log n)$ using the so-called Seymour trick. 
	Bernstein \etal \cite{BNW25} introduced and constructed directed LDDs with loss parameter $O(\log^2 n)$. 
	Later, Bringmann \etal \cite{BFHL25} improved the loss parameter to $\tilde{O}(\log n)$ using the Seymour trick. 
	Assadi \etal \cite{AHW25} constructed a DAG cover by hierarchically using the LDDs of \cite{BFHL25}. 
	Embedding digraphs into DAGs is much more delicate than embedding graphs into trees, and there are considerable challenges. 
	\cite{AHW25} constructed a stochastic embedding into DAGs with expected distortion $\tilde{O}(\log^3 n)$. The analysis is similar in spirit to \cite{Bar96}: they lose a $\tilde{O}(\log n)$ factor across $O(\log n)$ different scales, and another $\log n$ factor due to technical reasons.
	
	In this paper we introduce a simpler framework to analyze the expected distortion. Using this framework one can show that the stochastic embedding of \cite{AHW25} actually has expected distortion $\tilde{O}(\log^2 n)$. 
	We then show that one can apply the Seymour trick on top of the entire hierarchy of directed LDDs (and not only per scale). 
	Finally, this more sensitive approach requires a delicate analysis of the ball carving process w.r.t. a path (previously, it was analyzed only w.r.t. an edge). 
	Eventually, we obtain the desired $\tilde{O}(\log n)$ expected distortion.

	\subsection{The Origins: Undirected Graphs}
	\paragraph*{Low Diameter Decomposition (LDD) for Undirected Graphs.}
	Given an undirected graph $G=(V,E,w)$, a \EMPH{$(\beta,\Delta)$-LDD} (low diameter decomposition) is a distribution over partitions $\cC$ of $V$ such that every cluster $C$ in every partition $\cC$ has diameter\footnote{The diameter of a cluster $C$ is $\max_{u,v}d_G(u,v)$ — that is, the maximum pairwise distance in the original graph. This is often called weak diameter (strong diameter is $\max_{u,v}d_{G[C]}(u,v)$). We will use the same diameter definition for both graphs and digraphs.\label{foot:diameter}} 
	at most $\Delta$, and for every pair $u,v\in V$, the probability that $u,v$ belong to different clusters is at most $\beta\cdot \frac{d_G(u,v)}{\Delta}$. 
	$\beta$ is called the \EMPH{loss parameter} (sometimes called Lipschitz parameter).
	LDDs are extensively studied, and every $n$-vertex graph admits LDDs with loss parameter $\beta=O(\log n)$ (regardless of $\Delta$), which is also tight \cite{Bar96}. 
	A classic way to construct LDDs is by ball carving. 
	Denote by $A=V$ the set of active vertices. Initially, all vertices are active. 
	Iteratively, while $A\ne\emptyset$ we construct clusters as follows: 
	Take an arbitrary center $x\in A$, sample a radius $R\sim\Exp(\lambda)$ using an exponential distribution with parameter $\lambda=\Theta(\frac{\log n}{\Delta})$, and carve a cluster $C=B_{G[A]}(x,R)$ containing all active vertices at distance at most $R$ from $x$ in the induced graph $G[A]$. We then update $A\leftarrow A\setminus C$ and continue until all vertices are clustered. 
	One can show that, as we carved at most $n$ balls, w.h.p. all the sampled radii were at most $\frac{\Delta}{2}$, and thus all clusters indeed have diameter at most $\Delta$.\footnote{In fact, \cite{Bar96}, as well as this paper, use a truncated exponential distribution: sampling radii conditioned on being at most $\frac{\Delta}{2}$. The analysis is based on the fact that the difference between the two distributions is small.} 
	On the other hand, using the memoryless property of the exponential distribution, once one of $u,v$ joined a cluster, the probability that the other will not join is bounded by $1-e^{-\lambda\cdot d_{G}(u,v)}=O(\log n)\cdot\frac{d_{G}(u,v)}{\Delta}$, and thus we obtain an LDD. 
	
	\paragraph*{Stochastic Tree Embeddings with Expected Distortion $O(\log^2 n)$ \cite{Bar96}.}
	LDDs are a basic routine in many algorithms. 
	In particular, they are used in the construction of stochastic tree embeddings. 
	The framework is as follows: 
	let $\Phi=\poly(n)$ be the maximum distance in $G$ (the minimum distance is $1$). 
	Construct a laminar collection of partitions $\cC_0,\cC_1,\dots,\cC_i,\dots$ where every cluster in $\cC_i$ has diameter at most $\Phi\cdot 2^{-i}$, and $\cC_i$ refines $\cC_{i-1}$. 
	This is done by taking each cluster $C\in \cC_{i-1}$ and applying on it an $(O(\log n),\Phi\cdot 2^{-i})$-LDD to obtain a partition of $C$, and taking all the clusters in all the partitions to obtain $\cC_i$. 
	For a pair of vertices $u,v\in V$, let $\Delta_{u,v}$ be the minimum diameter of a cluster containing both $u,v$. 
	Alternatively, let $i$ be the maximum index such that $u$ and $v$ belong to the same cluster in $\cC_i$, and set $\Delta_{u,v}=\Phi\cdot 2^{-i}$. 
	One can construct a tree $T$ on top of this partition such that for every vertex pair $u,v\in V$, $d_T(u,v)=O(\Delta_{u,v})$. It follows that $\E[d_T(u,v)]=O(1)\cdot\E[\Delta_{u,v}]$. A bound on the expected distortion follows: 
	\begin{align}
		\mathbb{E}[d_{T}(u,v)] & \le\sum_{i=0}^{\log\Phi}O(1)\cdot\Phi\cdot2^{-i}\cdot\Pr\left[u\text{ and }v\text{ are separated in }\cC_{i}\right]\nonumber \\
		& \le O(1)\cdot\sum_{i=0}^{\log\Phi}\Phi\cdot2^{-i}\cdot\beta\cdot\frac{d_{X}(x,y)}{\Phi\cdot2^{-i}}\nonumber\\
		& =\beta\cdot O(\log \Phi)\cdot d_{X}(x,y)=O(\log^{2}n)\cdot d_{X}(x,y)~.\label{eq:Bartal96}
	\end{align}
	
	\paragraph*{Stochastic Tree Embeddings with Expected Distortion $\tilde{O}(\log n)$ \cite{Bartal98}.\footnote{In fact, \cite{Bartal98} is based on a deterministic algorithm and a minimax argument. Nevertheless, the result in \cite{Bartal98} can be constructed using the arguments described here.}}
	The best possible loss parameter $\beta$ for LDDs is $O(\log n)$ \cite{Bar96}. 
	If we pay $O(\log n)$ in each of the $O(\log n)$ scales, as in \cref{eq:Bartal96}, there is no way to avoid $O(\log^2 n)$ expected distortion. 
	The improvement idea is to show that it is possible to pay only $\tilde{O}(\log n)$ over all scales, following the so-called Seymour trick \cite{Seymour95}. 
	The basic intuition is as follows: 
	Suppose that, ``miraculously'', we are guaranteed that no matter how we choose the radii, every carved ball will contain about $\mu\ll m$ edges. In particular, the number of carved balls will only be $\approx\frac{m}{\mu}$. 
	This means that we can sample the radii $R$ much more aggressively, i.e., using parameter $\lambda_\mu=\frac{1}{\Delta}\cdot O(\log \frac{m}{\mu})$, instead of $\lambda=\frac{1}{\Delta}\cdot O(\log m)$. 
	As there are only $\approx\frac{m}{\mu}$ carved balls, it is still very likely that all the sampled radii are below $\frac{\Delta}{2}$, while the probability that a pair $s,t$ belong to different clusters is much smaller: $O(\log  \frac{m}{\mu})\cdot\frac{d_G(s,t)}{\Delta}$. 
	Note that if the pair $s,t$ is separated, the distance between them in the tree will be $d_T(s,t)=O(\Delta)$. 
	Assume by induction on the number of edges $|E|$ that the expected distortion is $O(\log |E|)$. 
	If both $s,t$ join the same cluster, then using the inductive hypothesis their expected distortion will be $O(\log \mu)$. We conclude:
	\begin{align}
		\E\left[d_{T}(s,t)\right] & =\Pr\left[\text{cut}\right]\cdot O(\Delta)+\Pr\left[\text{no cut}\right]\cdot\E\left[d_{T}(s,t)\mid s,t\text{ belong to a cluster with }\approx\mu\text{ edges}\right]\nonumber \\
		& \le O(\log\frac{m}{\mu})\cdot\frac{d_{G}(s,t)}{\Delta}\cdot O(\Delta)+O(\log\mu)\cdot d_{G}(s,t)\nonumber\\
		&=\left(O(\log\frac{m}{\mu})+O(\log\mu)\right)\cdot d_{G}(s,t)=O(\log m)\cdot d_{G}(s,t)~,\label{eq:Bartal98}
	\end{align}
	and thus the induction holds. 
	Removing the ``miraculous'' assumption (that all carved balls have size about $\mu$) costs another $O(\log\log n)$ factor in the analysis. 
	Roughly, this is done in $\log\log m$ phases, where in the $\ell$'th phase we choose the centers and radii so that we are guaranteed that the number of edges in the carved balls lies in the interval $\left[\frac{m}{2^{2^{\ell+1}}},\frac{m}{2^{2^{\ell}}}\right]$ (note that $2^{2^{\ell+1}}=(2^{2^{\ell}})^{2}$).

	\subsection{Directed LDDs and Laminar Topological Order}
	In undirected graphs, an alternative way to define LDDs is the following: a $(\beta,\Delta)$-LDD is a distribution over subsets of edges $S\subseteq E$, such that every connected component in $G\setminus S$ has diameter at most $\Delta$, and for every edge $e\in E$, $\Pr[e\in S]\le \beta\cdot\frac{w(e)}{\Delta}$. 
	Given such a subset $S$, it induces a partition into clusters $\cC$ — the connected components. For a pair of vertices $u,v$, note that if no edge along their shortest path was added to $S$, then $u$ and $v$ belong to the same cluster. By the triangle inequality and the union bound, it follows that the probability that $u,v$ belong to different clusters is at most $\beta\cdot \frac{d_G(u,v)}{\Delta}$.
	
	The directed counterpart of LDDs was defined only recently by Bernstein, Nanongkai, and Wulff-Nilsen \cite{BNW25}, and used for their breakthrough $\tilde{O}(m)$-time SSSP algorithm in graphs with negative edge weights. 
	Bernstein \etal \cite{BNW25} generalized the alternative LDD definition. 
	Given a digraph $G=(V,E,w)$, an \EMPH{$(\beta,\Delta)$-directed LDD} is a distribution over subsets of edges $S\subseteq E$, such that every SCC of $G\setminus S$ has diameter$^{\ref{foot:diameter}}$ at most $\Delta$, and for every edge $e\in E$, $\Pr[e\in S]\le \beta\cdot\frac{w(e)}{\Delta}$.
	
	\begin{wrapfigure}{r}{0.25\textwidth}
		\begin{center}
			\vspace{-30pt}
			\includegraphics[width=0.83\textwidth]{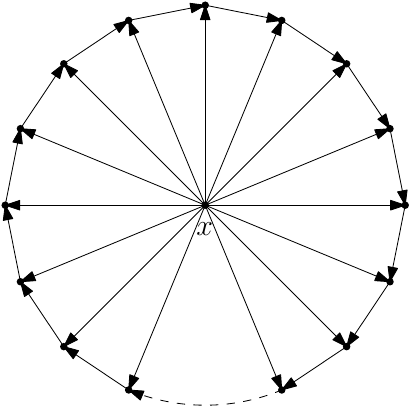}
			\vspace{-5pt}
		\end{center}
		\vspace{-30pt}
	\end{wrapfigure}
	\paragraph*{Directed LDD with $O(\log^2 n)$ loss parameter \cite{BNW25}.}
	Bernstein \etal \cite{BNW25} showed that every $n$-vertex digraph admits a directed LDD with loss parameter $\beta=O(\log^2 n)$. 
	The basic approach is similar to \cite{Bar96} — carve balls with random radii sampled using an exponential distribution with parameter $\lambda=\Theta(\frac{\log n}{\Delta})$. 
	Here one can carve either the \EMPH{out ball} $\brick{B_G^+(x,R)}=\{v\mid d_G(x,v)\le R\}$ and delete (i.e., add to $S$) all the outgoing edges from $B_G^+(x,R)$, or the \EMPH{in ball} $\brick{B_G^-(x,R)}=\{v\mid d_G(v,x)\le R\}$ and delete all the ingoing edges into $B_G^-(x,R)$. 
	Note that once we carve the ball $B_G^*(x,R)$ (for $*\in\{\pm\}$), every SCC is either internal to $B_G^*(x,R)$ or external to $B_G^*(x,R)$. 
	However, there is a significant issue — there is no bound on the diameter of SCCs contained in the ball $B_G^*(x,R)$. Indeed, consider the example on the right where $B_G^+(x,1)$ contains the directed $n$-cycle as a SCC with huge diameter. 
	It follows that, unlike in undirected graphs, the treatment of the vertices inside the ball $B_G^*(x,R)$ is not yet finished. 
	Accordingly, \cite{BNW25} partition the digraph $G$ using balls $B_G^*(x,R)$. Then, inside each ball, every SCC of diameter larger than $\Delta$ is partitioned again and again until no large-diameter SCCs remain. 
	\cite{BNW25} choose the centers $x_j$ in a clever way to ensure that the depth of the recursion is bounded by $O(\log n)$. 
	The analysis is roughly in the spirit of \cite{Bar96}: in a single level of the recursion, the edge $e$ is cut (added to $S$) with probability $\le O(\log n)\cdot\frac{w(e)}{\Delta}$. As the depth of the recursion is $O(\log n)$, implying that there are at most $O(\log n)$ attempts to cut $e$. 
	By the union bound, $e$ joins $S$ with probability $\le O(\log^2 n)\cdot\frac{w(e)}{\Delta}$.
	
	\paragraph*{Directed LDDs with loss $\tilde{O}(\log n)$ \cite{BFHL25}.}
	In a follow-up work, Bringmann, Fischer, Haeupler, and Rustam \cite{BFHL25} constructed directed LDDs\footnote{The main LDD construction in \cite{BFHL25} is deterministic and uses a multiplicative weights argument. Here we refer to their efficient algorithm.} with loss parameter $\beta=\tilde{O}(\log n)$. 
	This is done in a similar spirit to our description of \cite{Bartal98}, using the Seymour trick \cite{Seymour95}. 
	The basic intuition is as follows: 
	suppose that miraculously we are guaranteed that no matter how we choose the radii, every carved ball contains $\approx\mu$ edges. In particular, the number of carved balls, or ``cutting events,'' is only $\approx\frac{m}{\mu}$. 
	Thus we can sample the radii more aggressively, using parameter $\lambda_\mu=\frac{1}{\Delta}\cdot O(\log \frac{m}{\mu})$ (instead of $\lambda=\frac{1}{\Delta}\cdot O(\log m)$). 
	As there are only $\approx\frac{m}{\mu}$ carved balls, it is very likely that all the sampled radii are below $\frac{\Delta}{2}$, while the probability of an edge $e$ being cut in this round is only $O(\log  \frac{m}{\mu})\cdot\frac{w(e)}{\Delta}$.
	
	Suppose by induction that if $e$ currently belongs to a SCC with $|E|$ edges, then the probability of $e$ being cut in the remaining algorithm is $O(\log  |E|)\cdot\frac{w(e)}{\Delta}$. 
	Assuming the induction holds for SCCs with fewer than $m$ edges (and that all the carved balls have $\approx \mu$ edges), the probability that $e$ is cut either now or later (by induction) in the rest of the algorithm is at most $\left(O(\log\frac{m}{\mu})+O(\log\mu)\right)\cdot\frac{w(e)}{\Delta}=O(\log m)\cdot\frac{w(e)}{\Delta}$. 
	The extra $\log\log n$ factor comes from ensuring that all the carved balls have similar size (there are $\log\log m$ phases as in \cite{Bartal98}). 
	
	It is also worth mentioning that very recently, Li \cite{Li25} used a combination of the undirected LDD from \cite{CKR04} and the Seymour trick \cite{Seymour95} to obtain directed LDDs with the same loss parameter $\beta=O(\log n\cdot \log\log n)$.\footnote{The motivation of \cite{Li25} was to improve the runtime. Li's algorithm runs in $m\cdot\tilde{O}(\log^{2}n)$ time, while \cite{BFHL25} takes $m\cdot\tilde{O}(\log^{5}n)$ time.}
	
	\paragraph*{Laminar Directed LDDs $\rightarrow$ Laminar Topological Order.}
	The construction of the stochastic embedding into DAGs (following \cite{AHW25}) is based on a recursive application of LDDs. 
	Suppose for simplicity that the input digraph $G$ is strongly connected, and let $\Delta$ be its diameter. 
	Begin by running a $(\beta,\frac{\Delta}{2})$-directed LDD for $\beta=O(\log n\cdot\log\log n)$. 
	As a result, we get a set of edges $S$, where every SCC in $G\setminus S$ has diameter at most $\frac{\Delta}{2}$. 
	Continue recursively: on every SCC $C$ of $G\setminus S$ with diameter $\Delta_C$, run a $(\beta,\frac{\Delta_C}{2})$-directed LDD to obtain a subset $S_C$ of edges such that every SCC of $G[C]\setminus S_C$ has diameter at most $\frac{\Delta_C}{2}$. 
	We stop once we get a SCC with diameter $<1$, i.e., a singleton. 
	This recursion has logarithmic depth. 
	Consider the recursion tree \EMPH{$\cT$}, containing all the SCCs \EMPH{$\cP$} on which we recursively run the algorithm. Note that given two SCCs $C,C'\in \cP$, they could be either disjoint or one can contain the other. For example, if $C\subseteq C'$, then $C'$ is an ancestor of $C$ in $\cT$. 
	For every pair $s,t\in V$, let \EMPH{$C_{s,t}$} be the minimal cluster in $\cP$ (lowest in $\cT$) containing both $s,t$, and let $\brick{\Delta_{s,t}}=\Delta_C$ be the diameter of this cluster. 
	Let $\brick{\widetilde{S}}=\cup_{C\in\cP}S_C$ be the set of deleted edges throughout the algorithm. 
	Let $D=G\setminus \widetilde{S}$ be the remaining graph at the end of the algorithm. As all the remaining SCCs are singletons, $D$ is a DAG. 
	Let $<_D=(v_1,v_2,\dots,v_n)$ be a topological order over the vertices w.r.t. the DAG $D$. 
	It is possible to guarantee that every SCC $C\in\cP$ appears in $<_D$ contiguously. That is, there are $i<j$ such that $C=\{v_i,v_{i+1},\dots,v_j\}$. 
	We call the pair $(<_D,\cP)$ a \EMPH{laminar topological order} (see \Cref{def:Laminar}).

	\subsection{Stochastic Embedding into DAGs}
	Fix $s,t\in V$, and let $\pi$ be the shortest $s$-$t$ path in $G$. 
	Is it possible to bound the expected distances in $D$? Not at all. 
	Indeed, we cannot upper bound the probability that some edge of $\pi$ is deleted. That is, at some stage of the algorithm, a subpath $\pi'$ of $\pi$ might belong to a SCC $C$ of diameter comparable to the length of $\pi'$. In this case, the event that some edge along $\pi'$ is cut might even be likely. 
	If some edge of $\pi$ is deleted, in $D$ there might be no path from $s$ to $t$ (even if $s<_D t$).
	
	\paragraph*{$2$-hop spanners.}
	Following \cite{AHW25}, we construct two DAGs $D_1,D_2$ based on $<_D$. 
	In $D_1$ we add only edges respecting the topological order $<_D$ (i.e., an edge $(u,v)$ is added only if $u<_D v$), while in $D_2$ we add only edges respecting the reversed order $(v_n,v_{n-1},\dots,v_1)$. 
	A $2$-hop spanner $H$ (see \Cref{lem:2hopSpanner}) over $(v_1,v_2,\dots,v_n)$ is a collection of $O(n\log n)$ edges $(v_a,v_b)$ for $a<b$, satisfying that for every $i<j$, there exists $i\le q\le j$ such that $(v_i,v_q),(v_q,v_j)\in H$. 
	We add to $D_1$ all $(v_i,v_j)$ edges from $H$, as long as $v_i\rightsquigarrow_{G}v_j$ (using $d_G(v_i,v_j)$ as the weight). 
	Note that for every $C\in\cP$, and $v_i,v_j\in C$ with $i<j$, as the vertices of $C$ are contiguous along $<_D$, it follows that the vertex $v_q$ guaranteed by the spanner also belongs to $C$, and in particular 
	$d_{D_1}(v_i,v_j)\le d_{D_1}(v_i,v_q)+d_{D_1}(v_q,v_j)\le 2\Delta_C$. 
	Similarly for $D_2$: for every edge $(v_i,v_j)\in H$ with $i<j$, if $v_j\rightsquigarrow_{G}v_i$, we add the edge $(v_j,v_i)$ to $D_2$. 
	By the same argument, for every $v_i,v_j\in C$ with $i<j$, $d_{D_1}(v_j,v_i)\le 2\Delta_C$. 
	We conclude that for every pair of vertices,
	\begin{equation}
		d_{D_{1}}(s,t)\cdot\mathds{1}[s<_{D}t]+d_{D_{2}}(s,t)\cdot\mathds{1}[s>_{D}t]\le2\cdot\Delta_{s,t}~.\label{eq:SeprationBoundOnDiatance}
	\end{equation}
	
	\begin{wrapfigure}{r}{0.25\textwidth}
		\begin{center}
			\vspace{-30pt}
			\includegraphics[width=0.83\textwidth,page=1]{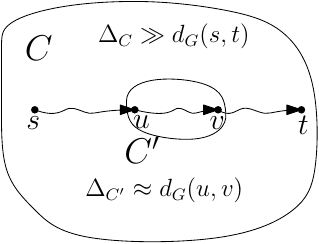}
			\vspace{-5pt}
		\end{center}
		\vspace{-15pt}
	\end{wrapfigure}
	\paragraph*{Adding Additional Edges to $D_1$.}
	Note that the guarantee in \cref{eq:SeprationBoundOnDiatance} is similar to the one used in \cite{Bar96} — the distance is bounded by the diameter of the smallest cluster containing both $s$ and $t$ (specifically $d_T(u,v)\le \Delta_{u,v}$). 
	In \cite{Bar96}, such an inequality, together with the loss parameter $O(\log n)$, implied an $O(\log^2 n)$ bound on the expected distortion. 
	Does \cref{eq:SeprationBoundOnDiatance} imply any bound on the expected distortion here? Not at all. 
	Suppose that $s$ and $t$ belong to the same SCC $C$ with $\Delta_C\gg d_G(s,t)$, and let $\pi$ be the shortest $s$-$t$ path contained in $C$. 
	The algorithm deletes a set of edges $S_C$. 
	As $\Delta_C\gg d_G(s,t)$, it is highly unlikely that some edge along $\pi$ is cut. 
	However, it is still possible (and in some cases even likely) that $s$ and $t$ end up in different SCCs of $G[C]\setminus S_C$ (see the illustration on the right). 
	Here the guarantee of \cref{eq:SeprationBoundOnDiatance} is clearly insufficient, and there is no bound on the distortion. 
	Next, we add all the edges of $D=G\setminus\widetilde{S}$ to $D_1$. 
	Is this enough? Still no. 
	Following the previous case, suppose that some subpath $\pi'=\pi[u,v]$ of $\pi$ is contained in some SCC $C'$ of diameter proportional to $d_G(u,v)$. 
	It might be a likely event that some edge of $\pi'$ is cut (added to $S_{C'}$), and thus $s\not\rightsquigarrow_{D}t$.
	
	\begin{wrapfigure}{r}{0.25\textwidth}
		\begin{center}
			\vspace{-20pt}
			\includegraphics[width=0.83\textwidth,page=2]{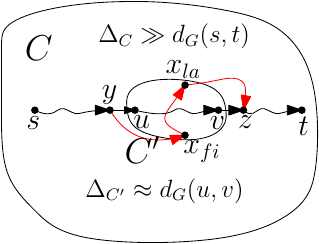}
			\vspace{-5pt}
		\end{center}
		\vspace{-15pt}
	\end{wrapfigure}
	The solution (due to \cite{AHW25}) is to add some additional edges. Specifically, suppose there are two disjoint clusters $C_u,C_v$, and an edge $(u,v)\in C_u\times C_v$ such that $(u,v)\in D$. 
	Let $C_u^\last$ (resp. $C_v^\first$) be the last (resp. first) vertex w.r.t. $<_D$ in $C_u$ (resp. $C_v$). We add to $D_1$ the edge $(C_u^\last,C_v^\first)$. 
	Note that as $\{u\},\{v\}\in\cP$, the edges $(u,C_v^\first)$ and $(C_u^\last,v)$ are also added (see illustration on page \pageref{page:4AddedEdges}). 
	These edges help overcome the issue with $\pi'$ above. 
	Indeed, even if there is no direct path from $u$ to $v$, we can ``jump over'' this obstacle. 
	In the figure, $y$ is the vertex preceding $u$ in $\pi$, and $z$ is the vertex following $v$. $x_{\rm{fi}}= C'^\first$ (resp. $x_{\rm{la}}= C'^\last$) is the first (resp. last) vertex in $C'$ w.r.t. $<_D$. 
	We added to $D_1$ the edges $(y,x_{\rm{fi}}),(x_{\rm{la}},z)$ with weights $d_G(y,x_{\rm{fi}})\le w((y,u))+\Delta_{C'}$, $d_G(x_{\rm{la}},z)\le w((v,z))+\Delta_{C'}$. 
	In addition, due to the $2$-hop spanner, $d_{D_1}(x_{\rm{fi}},x_{\rm{la}})\le 2\Delta_{C'}$. 
	Overall, $d_{D_1}(y,z)\le d_G(y,z)+4\Delta_{C'}$. 
	Thus we can circumvent the edge cut along $\pi'$ while paying only $4\Delta_{C'}$. 
	This is crucial, as now the ``penalty'' ($4\Delta_{C'}$) is proportional to the ``offense'' (cutting an edge from $\pi'$ in the cluster $C'$). 
	This concludes the construction of $D_1,D_2$ (note that there are no additional edges in $D_2$ beyond the $2$-hop spanner).

	\paragraph*{\cite{AHW25}'s Analysis.}
	Other than the different LDDs we are using (to be discussed later), the sampling and construction of the DAGs $(D_1,D_2)$ is largely the same as in \cite{AHW25}. 
	However, the analysis is quite different. 
	\cite{AHW25} break the analysis into two parts. First, they bound $\E\left[d_{D_{2}}(s,t)\cdot\mathds{1}[s>_{D}t]\right]$. This analysis largely follows \cite{Bar96}, and indeed they show that 
	$\E\left[d_{D_{2}}(s,t)\cdot\mathds{1}[s>_{D}t]\right]=\tilde{O}(\log^2 n)\cdot d_G(s,t)$. 
	\cite{AHW25} analysis of $\E\left[d_{D_{1}}(s,t)\cdot\mathds{1}[s<_{D}t]\right]$ is considerably more involved. Morally, it also attempts to follow \cite{Bar96}, but due to technical reasons they accumulate an additional $\log n$ factor, and prove an upper bound of $\tilde{O}(\log^3 n)\cdot d_G(s,t)$. 
	We will not delve into the details of \cite{AHW25}'s analysis, as it would be a diversion (from an already long story). 
	Instead, we will prove that using our framework, one can show that \cite{AHW25}'s stochastic embedding into DAGs actually has expected distortion $\tilde{O}(\log^2 n)\cdot d_G(s,t)$.
	
	\paragraph*{Framework for the analysis: $\alpha(s,t)$.}
	Set $\brick{\alpha(s,t)}:=d_{D_{1}}(s,t)\cdot\mathds{1}[s\le_{D}t]+3\Delta_{st}\cdot\mathds{1}[s\ge_{D}t]$. 
	Note that due to the $2$-hop spanner, in order to bound the expected distortion, it is enough to bound $\E[\alpha(s,t)]$. 
	However, $\alpha$ has a very desirable property — a type of triangle inequality. 
	Suppose that during the execution of the algorithm, the vertices of $\pi$ were partitioned into two SCCs $C_s$ and $C_t$, such that $C_s$ contains the prefix $\pi[s,u]$ of $\pi$, $C_t$ contains the suffix $\pi[v,t]$ of $\pi$, and the edge $(u,v)\in C_s\times C_t$ does not belong to the cut set $\widetilde{S}$. 
	Here necessarily $s<_D t$, and thus $\alpha(s,t)=d_{D_{1}}(s,t)$. 
	The interesting part is that the following inequality holds (see \Cref{lem:alphaAppliedRecursively}):
	\begin{equation}
		\alpha(s,t)\le\alpha(s,u)+d_{G}(u,v)+\alpha(v,t)~.\label{eq:triangleInequalityIntro}
	\end{equation}
	The proof of this inequality largely follows by case analysis in a spirit similar to the discussion we had for the need to add additional edges to $D_1$. See \Cref{fig:cases} for a pictorial illustration.
	
	Using \cref{eq:triangleInequalityIntro}, we can easily prove that \cite{AHW25}'s stochastic embedding has expected distortion $\tilde{O}(\log^2 n)$. 
	Indeed, let $\beta=\tilde{O}(\log n)$ be the loss parameter from \cite{BFHL25}. Assume by induction that for every pair of vertices $s,t$ belonging to the same SCC $Q$ of diameter $\Delta$ it holds that $\E[\alpha(s,t)]\le \log\Delta\cdot 6\beta\cdot d_{G[Q]}(s,t)$. 
	Let $\pi$ be the shortest $s$-$t$ path in $G[Q]$. 
	We now execute \cite{BFHL25} with diameter $\frac\Delta2$, to obtain a cut set $S$. 
	If some edge of $\pi$ belongs to $S$, then by the $2$-hop spanner $\alpha(s,t)\le 3\Delta$. However, the probability of this event is only $\beta\cdot\frac{d_{G[Q]}(s,t)}{\Delta/2}$. 
	Otherwise, no edge of $\pi$ is cut. However, it is still possible that the vertices of $\pi$ are partitioned into different SCCs (if all $\pi$ vertices belong to the same SCC, the proof is only easier). Suppose for simplicity that they are partitioned into exactly two SCCs $C_s,C_t$ as above. Note that $C_u$ and $C_v$ have diameter at most $\frac{\Delta}{2}$. 
	Using the inductive hypothesis, it holds that 
	\begin{align*}
		\E[\alpha(s,t)] & \le\E[\alpha(s,u)]+d_{G}(u,v)+\E[\alpha(v,t)]\\
		& \le\log\frac{\Delta}{2}\cdot6\beta\cdot d_{G[Q]}(s,u)+d_{G}(u,v)+\log\frac{\Delta}{2}\cdot6\beta\cdot d_{G[Q]}(v,t)\\
		& \le\log\frac{\Delta}{2}\cdot6\beta\cdot d_{G[Q]}(s,t)~.
	\end{align*}
	If the path $\pi$ is broken into more than $2$ SCCs, the inequality above still holds, by applying \cref{eq:triangleInequalityIntro} between every two consecutive SCCs (see \Cref{fig:Generalcase} and \cref{eq:alphaSTforPath}). We conclude:
	\begin{align}
		\E[\alpha(s,t)] & \le\Pr\left[\pi\text{ is cut}\right]\cdot3\Delta+\Pr\left[\overline{\pi\text{ is cut}}\right]\cdot\log\frac{\Delta}{2}\cdot6\beta\cdot d_{G[Q]}(s,t)\label{eq:log2DistortionAnalysisIntro}\\
		& \le\beta\cdot\frac{d_{G[Q]}(s,t)}{\Delta/2}\cdot3\Delta+\log\frac{\Delta}{2}\cdot6\beta\cdot d_{G[Q]}(s,t)\nonumber \\
		& =\left(1+\log\frac{\Delta}{2}\right)\cdot6\beta\cdot d_{G[Q]}(s,t)=\log\Delta\cdot6\beta\cdot d_{G[Q]}(s,t)~.\nonumber 
	\end{align}
	Thus we obtain expected distortion $\tilde{O}(\log^2 n)$, as promised.

	\subsection{Getting Expected Distortion $\tilde{O}(\log n)$}
	Bringmann \etal \cite{BFHL25} constructed directed LDDs using the Seymour trick. 
	We now take a closer look at the argument, from a potential function perspective.
	Consider an edge $e=(u,v)$. We argue that the probability that $e$ is cut is bounded by $\tilde{O}(\log m)\cdot \frac{w(e)}{\Delta}$, which matches its initial potential.
	Suppose that we grow an out ball $B=B_G^+(x,R)$ such that we are guaranteed that the number of edges in the ball will be $\approx\mu$ (regardless of the sampled $R$). 
	Accordingly, we sample $R$ using an exponential distribution with parameter $\lambda_\mu=\frac1\Delta\cdot \tilde{O}(\log\frac{m}{\mu})$.
	We think of $R$ as growing gradually (it is helpful to think of a geometric distribution: flip a coin, and as long as it comes up heads, $R$ continues to grow).
	Suppose that $u$ joined $B$, while $v$ is still outside $B$.
	Using memorylessness, the probability that $v$ will not join $B$, and the edge will be cut is at most $\tilde{O}(\log\frac{m}{\mu})\cdot \frac{w(e)}{\Delta}$.
	Morally, $e$ used up here $\tilde{O}(\log\frac{m}{\mu})\cdot \frac{w(e)}{\Delta}$ of its potential.
	If $e$ was fortunate enough and was not cut, it will now belong to a cluster of size $\approx\mu$, and appropriately has a remaining potential of $O(\log\mu)\cdot \frac{w(e)}{\Delta}$.
	If $e$ now belongs to a cluster with diameter greater than $\Delta$, we will continue to carve balls from this cluster, and argue inductively that the overall probability of being cut is 
	$\tilde{O}(\log\frac{m}{\mu})\cdot\frac{w(e)}{\Delta}+\tilde{O}(\log\mu)\cdot\frac{w(e)}{\Delta}=\tilde{O}(\log m)\cdot\frac{w(e)}{\Delta}$.
	However, it might be the case that $e$ does not belong to any SCC, or that the SCC containing $e$ has diameter at most $\Delta$. In this case, the LDD algorithm will stop processing $e$, and it will be left with $\tilde{O}(\log\mu)\cdot\frac{w(e)}{\Delta}$ unused potential.
	
	\paragraph*{Laminar Topological Order with Expected Distortion $\tilde{O}(\log n)$ over edges.}
	Next we will create a laminar topological order with a more local view.
	That is, given a SCC of diameter $\Delta$, we will follow \cite{BFHL25} and carve balls using a distribution proportional to the local density ($\lambda_\mu$), while attempting to construct clusters with diameter $\le\frac\Delta2$. 
	As a result, we will delete some edges, and get SCCs $C_1,\dots,C_k$.
	On each SCC $C_i$ of diameter $\Delta_{C_i}$, we will now continue running the algorithm w.r.t. $\Delta_{C_i}$, regardless of whether $\Delta_{C_i}\le \frac\Delta2$. 
	Continuing in this manner, we will eventually get a DAG, and thus a laminar topological order.
	
	Fix an edge $e=(u,v)$, and let $\Delta_e$ be the diameter of the cluster where it was added to the cut set $\widetilde{S}$, or $\Delta_e=w(e)$ if $e\notin\widetilde{S}$.
	We argue that $\E[\Delta_e]=\tilde{O}(\log m)\cdot w(e)$.
	That is, the expected diameter of the cluster where $e$ is cut (if any) is at most $\tilde{O}(\log m)\cdot w(e)$.
	Suppose that we carve a ball $B=B_G^+(x,R)$ such that we are guaranteed that the number of edges in the ball will be $\approx\mu$ (regardless of the sampled $R$). Accordingly, we sample $R$ using parameter $\lambda_\mu=\frac1\Delta\cdot \tilde{O}(\log\frac{m}{\mu})$.
	We do it gradually as before. Suppose that $u\in B$. The probability that $e\notin B$ (that is, the edge $e$ is cut) is at most 
	$\tilde{O}(\log\frac{m}{\mu})\cdot\frac{w(e)}{\Delta}$.
	If the edge $e$ is not cut, then either $e$ will not belong to any SCC (in which case $\Delta_e=w(e)$), or $e$ belongs to a SCC with $\approx\mu$ edges, in which case, by the inductive hypothesis $\E[\Delta_e]=\tilde{O}(\log \mu)\cdot w(e)$. Note that the inductive hypothesis does not assume anything about the diameter of the SCC containing $e$. We conclude
	\begin{align*}
		\E[\Delta_{e}] & =\Pr\left[e\text{ is cut}\right]\cdot\Delta+\Pr\left[e\text{ is not cut}\right]\cdot\E[\Delta_{e}\mid e\text{ is not cut}]\\
		& \le\tilde{O}(\log\frac{m}{\mu})\cdot\frac{w(e)}{\Delta}\cdot\Delta+\tilde{O}(\log\mu)\cdot w(e)=\tilde{O}(\log m)\cdot w(e)~.
	\end{align*}
	Thus, instead of paying loss $\tilde{O}(\log n)$ in every scale, we are now paying $\tilde{O}(\log n)$ over all scales!
	Can we utilize this new gain to improve the expected distortion? Is it possible to plug it into \cref{eq:log2DistortionAnalysisIntro} to prove that $\E[\alpha(s,t)]\le\tilde{O}(\log m)\cdot d_{G[Q]}(s,t)$?
	A priori, no. 
	The issue is that the triangle inequality is not sufficient to generalize the analysis from an edge to a path.
	If miraculously it so happens that all the balls that we carve are guaranteed to have density $\approx \mu$, then yes. We can say that the probability to cut any edge along $\pi$ is $\tilde{O}(\log\frac{m}{\mu})\cdot\frac{d_G(s,t)}{\Delta}$, and the analysis (using the triangle inequality over $\alpha$) will go through.
	Unfortunately, this is not the case.
	That is, different portions of the shortest $s$-$t$ path $\pi$ might be carved using different densities.
	Indeed, consider the illustration below.
	First $x_1$ carves a ball containing $\mu_1$ edges. This ball contains the suffix $\pi[z_1,t]$. No edge of $\pi$ was cut, however the entire prefix $\pi[s,y_1]$ remains active and ready to be carved.
	After some time, the density parameter was reduced to $\mu_2\ll\mu_1$, and now we carve a ball around $x_2$.
	This ball contains the subpath $\pi[z_2,y_1]$, and so on.
	
	\begin{center}
		\includegraphics[scale=0.8]{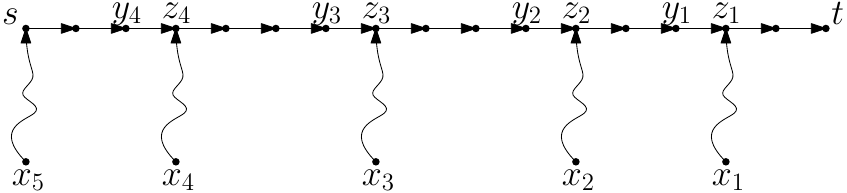}
	\end{center}
	
	\paragraph*{Analyzing w.r.t. a Path.}
	As it turns out, the way to proceed with the analysis is to consider the effects of the carving on a path as a hole (and not only w.r.t. individual edges), which is considerably more delicate.
	Assume by induction that given a SCC $Q$ with $m$ edges, such that $s,t\in Q$, it holds that $\E[\alpha(s,t)]=\tilde{O}(\log m)\cdot d_{G[Q]}(s,t)$. 
	Further, we will assume that the induction holds also in the middle of the carving process. 
	That is, if we already carved some balls, and the set of remaining vertices is $Q'\subseteq Q$, then $\E[\alpha(s,t)]=\tilde{O}(\log m)\cdot d_{G[Q']}(s,t)$ (note that we use the original cardinality $m$).
	Let $\pi$ be the shortest $s$-$t$ path in $G[Q]$.
	Suppose that we carve a ball $B=B_G^+(x,R)$ such that we are guaranteed that the number of edges in the ball will be $\approx\mu$, and thus sample $R$ using parameter $\lambda_\mu=\frac1\Delta\cdot \tilde{O}(\log\frac{m}{\mu})$.
	An edge of $\pi$ will be cut iff $B_G^+(x,R)\cap \pi$ is not a suffix of $\pi$.
	Note that $B_G^+(x,R)$ might intersect $\pi$ in many different places, and the patterns $B_G^+(x,R)\cap \pi$ for different values of $R$ might be quite complicated (see \Cref{fig:PathBrake} for an illustration).
	Let $z\in \pi$ be the closest vertex to $x$. Let $y$ be the predecessor vertex of $z$ in $\pi$.
	For simplicity, let's suppose that all the vertices in the prefix $\pi[s,y]$ are not reachable from $x$ (like in the figure above). 
	Suppose that $R$ grew enough so that $z$ joins $B_G^+(x,R)$.
	By the triangle inequality, if $R$ then grows by an additional $d_G(z,t)$, $B_G^+(x,R)\cap \pi$ will be the suffix $\pi[z,t]$ and no edge of $\pi$ will be cut.
	It follows that the probability of a cut is at most $\tilde{O}(\log\frac{m}{\mu})\cdot\frac{d_G(z,t)}{\Delta}$.
	If there was no cut, then we can apply the induction hypothesis on
	the prefix $\pi[s,y]$ w.r.t. $Q'=Q\setminus B_G^+(x,R)$.
	For simplicity, let's suppose that the sampled ball $B_G^+(x,R)$ is strongly connected, and recall that the number of edges there is about $\mu$.
	We thus can apply the induction hypothesis here as well.
	Using \cref{eq:triangleInequalityIntro} (which holds for this case as well), we conclude:
	\begin{align*}
		\E[\alpha(s,t)\mid\overline{\pi\text{ is cut}}] & \le\E[\alpha(s,y)]+d_{G[Q]}(y,z)+\E[\alpha(z,t)]\\
		& \le\tilde{O}(\log m)\cdot d_{G[Q]}(s,y)+d_{G[Q]}(y,z)+\tilde{O}(\log\mu)\cdot d_{G[Q]}(z,t)\\
		& \le\tilde{O}(\log m)\cdot\left(d_{G[Q]}(s,y)+d_{G[Q]}(y,z)+d_{G[Q]}(z,t)\right)-\tilde{O}(\log\frac{m}{\mu})\cdot d_{G[Q]}(z,t)\\
		& =\tilde{O}(\log m)\cdot d_{G[Q]}(s,t)-\tilde{O}(\log\frac{m}{\mu})\cdot d_{G[Q]}(z,t)~.
	\end{align*}
	Overall, taking into account the possibility of a cut:
	\begin{align*}
		\E[\alpha(s,t)] & \le\Pr\left[\pi\text{ is cut}\right]\cdot3\Delta+\Pr\left[\overline{\pi\text{ is cut}}\right]\cdot\E[\alpha(s,t)\mid\overline{\pi\text{ is cut}}]\\
		& \le\tilde{O}(\log\frac{m}{\mu})\cdot\frac{d_{G}(z,t)}{\Delta}\cdot3\Delta+\tilde{O}(\log m)\cdot d_{G[Q]}(s,t)-\tilde{O}(\log\frac{m}{\mu})\cdot d_{G[Q]}(z,t)\\
		& =\tilde{O}(\log m)\cdot d_{G[Q]}(s,t)~.
	\end{align*}
	
	This concludes the simplified discussion of the ideas and techniques in the paper.

	\section{Preliminaries}
	\EMPH{$\tilde{O}$} notation hides polylogarithmic factors, that is $\tilde{O}(g)=O(g)\cdot\polylog(g)$.
	$\log x=\log_2x$ will be in base $2$ (unless explicitly denoted otherwise), and we will use $\ln$ for $\log_e$.
	For an event $X$, \EMPH{$\mathds{1}[X]$} is an indicator that takes value $1$ if the event $X$ occurs, and $0$ otherwise.
	\EMPH{W.h.p.} (with high probability) stands for probability $1-\frac{1}{n^c}$ for an arbitrarily large constant $c$. We will always be able to govern the constant $c$ and make it as large as needed. Note that if we have polynomially many events that occur w.h.p., then we can say that all of them occur together w.h.p.
	
	A \EMPH{digraph} $G=(V,E,w)$ consists of a set of vertices $V$, a directed set of edges $E\subseteq V\times V$ (with no self-loops), and a positive weight function $w:E\rightarrow\R_{>0}$. By scaling, we will assume w.l.o.g. that the minimum edge weight is $1$.
	We will also use \EMPH{$V(G),E(G)$} to denote the set of vertices and edges of a digraph $G$.
	Given a subset of edges $E$, 
	$\brick{G\setminus S}=(V,E\setminus S,w)$ is the digraph obtained by deleting the set of edges $S$ (while keeping the same weights on the remaining edges).
	For a subset $A\subseteq V$ of vertices, $\brick{G[A]}=(A,E_A=E\cap A\times A,w_{\upharpoonright E_{A}})$ denotes the subgraph \EMPH{induced} by $A$, i.e., the graph with vertex set $A$ where we keep all edges between vertices of $A$.
	Given a digraph $G$, \EMPH{$s\rightsquigarrow_{G}t$} denotes that there is a path from $s$ to $t$ in $G$. Similarly, we will use $s\not\rightsquigarrow_{G}t$ to denote that there is no path from $s$ to $t$ in $G$.
	We say that $u$ and $v$ are \EMPH{strongly connected} if $u\rightsquigarrow_{G}v$ and $v\rightsquigarrow_{G}u$. 
	Note that being strongly connected is an equivalence relation. The equivalence classes are called \EMPH{SCCs} (strongly connected components).
	We say that a cluster $A\subseteq V$ is strongly connected if for every pair $u,v\in A$, $u\rightsquigarrow_{G[A]} v$.
	Note that the SCCs are maximal subsets (w.r.t. inclusion order) that are strongly connected.
	
	Given a path $\pi=(v_0,v_1,\dots,v_k)$, $\brick{\pi[v_i,v_j]}=(v_i,v_{i+1},\dots,v_j)$ denotes the subpath from $v_i$ to $v_j$. Similarly, we use $\pi(v_i,v_j)=(v_{i+1},v_{i+2},\dots,v_{j-1})$ to denote the subpath without endpoints (we might also use $\pi[v_i,v_j)$ and $\pi(v_i,v_j]$ to include only one endpoint).
	\EMPH{$d_{G}$} denotes the shortest path (quasi)metric in $G$, i.e., $d_G(u,v)$ is the minimal weight of a path from $u$ to $v$. 
	If there is no such path ($u\not\rightsquigarrow_{G}v$) then we set $d_G(u,v)=\infty$.
	The \EMPH{diameter} of a cluster $C\subseteq V$ is $\brick{\diam(C,G)}=\max_{u,v\in C}d_G(u,v)$.\footnote{This is often called \emph{weak} diameter. A related notion is the \emph{strong} diameter of a cluster $C$, defined as $\max_{u,v \in C}d_{G[C]}(u,v)$. See also \cref{foot:diameter}.} 
	Note that the diameter of a cluster is finite iff all vertices of $C$ belong to the same SCC.
	The \EMPH{aspect ratio}$^{\ref{foot:aspect}}$ is 
	$\Phi=\frac{\max\left\{d(u,v)\mid u\rightsquigarrow_{G}v\right\}}{\min\left\{d(u,v)\mid u\rightsquigarrow_{G}v\right\}}$, the ratio between the maximum and minimum (finite) distances in $G$.
	Given a center $v\in V$ and a radius $r\ge0$, the \EMPH{outgoing ball} $\brick{B^{+}_{G}(v,r)}=\{u\in V\mid d(v,u)\le r\}$ is the set of vertices at distance at most $r$ from $v$, while the \EMPH{ingoing ball} $\brick{B^{-}_{G}(v,r)}=\{u\in V\mid d(u,v)\le r\}$ is the set of vertices whose distance to $v$ is at most $r$. For $*\in\{+,-\}$, \EMPH{$|B^{*}_{G}(v,r)|$} denotes the number of edges fully contained in $B^{*}_{G}(v,r)$.
	For a set of vertices $C\subseteq V$, 
	$\brick{\delta^+(C)}=\left\{(u,v)\in E\mid u\in C,v\notin C\right\}$ denotes all edges with the source endpoint in $C$ and the target endpoint outside $C$. Similarly, $\brick{\delta^-(C)}=\left\{(u,v)\in E\mid u\notin C,v\in C\right\}$.
	
	\subsection{DAGs and Laminar Topological Order}
	A \EMPH{DAG} (directed acyclic graph) is a digraph not containing any directed cycle. In particular, the set of SCCs is the set of singleton vertices.
	Given a DAG $D=(V,E,w)$, a \EMPH{topological order} is a total order $<_{\rm TO}$ over $V$ such that for every edge $(u,v)\in E$, $u<_{\rm TO}v$.
	Note that this total order is not necessarily unique.
	Given a digraph $G$, one can construct a topological order over the SCCs. That is, for two SCCs $C_1,C_2$, if there is a path from a vertex in $C_1$ to a vertex in $C_2$, then $C_1$ must appear in the topological order before $C_2$.
	Given a digraph $G$, which is not necessarily a DAG, it is often useful to still order its vertices in a topological order w.r.t. some DAG contained in $G$.
	
	A key notion in previous work on stochastic embedding into trees is the laminar partition: a collection of partitions of $V$, $\cP_0,\cP_1,\dots,\cP_i,\dots$, such that all the clusters in $\cP_i$ have diameter at most $2^i$, and $\cP_i$ refines $\cP_{i+1}$.
	The directed analog we introduce here is the notion of laminar topological order. This is essentially a laminar partition (implicitly) equipped with an order. Here all the clusters are strongly connected, and the final order over the singleton clusters is a topological order w.r.t. some DAG contained in $G$.
	
	\begin{definition}[Laminar Topological Order]\label{def:Laminar}
		Given a digraph $G=(V,E,w)$, a laminar topological order of $G$ is a pair $(<_D,\cP)$ where $<_D$ is a total order over $V$, and $\cP\subseteq P(V)$ is a set of clusters with the following properties:
		\begin{enumerate}
			\item\label{req:singleton} Singleton clusters: For every vertex $v\in V$, $\{v\}\in\cP$.
			\item\label{req:hierarchy} Hierarchy: For every two clusters $C_{1},C_{2}\in\cP$, either the clusters
			are disjoint ($C_{1}\cap C_{2}=\emptyset$), or one contains the other
			($C_{1}\subseteq C_{2}$ or $C_{2}\subseteq C_{1}$).
			\item\label{req:continuity} Continuity: Every cluster $C\in\cP$ is contiguous w.r.t. $<_D$. That is, for every $u,v\in C$, for every $x$ such that $u<_D x<_D v$, it holds that $x\in C$.
			\item\label{req:connectivty} Connectivity: For every cluster $C\in\cP$, $G[C]$ is strongly connected.
		\end{enumerate}
	\end{definition}
	
	Given a laminar topological order $(<_D,\cP)$, and a cluster $C\in\cP$, we denote $\Delta_C=\diam(C,G)=\max_{u,v\in C}d_G(u,v)$. 
	Additionally, we denote by $C^\first$ and $C^\last$ the first and last vertices in the cluster (w.r.t. the order $<_D$).
	Given a pair of vertices $u,v$, we denote by $\Parent_{\cP}(u,v)$ the minimal cluster (w.r.t. inclusion order) containing both $u,v$. 
	If there is no such cluster then $\Parent_{\cP}(u,v)=\emptyset$.
	When the laminar topological order is clear from the context, we use the notation
	$C_{u,v}=\Parent_{\cP}(u,v)$, and denote by $\Delta_{u,v}=\Delta_{C_{u,v}}=\Delta_{\Parent_{\cP}(u,v)}$ the diameter of the minimum cluster containing $u,v$. If $\Parent_{\cP}(u,v)=\emptyset$, set $\Delta_{u,v}=\infty$.
	
	\section{Algorithm}\label{sec:algortihm}
	In order to simplify the proof, we will begin by proving a weaker version of \Cref{thm:mainExpectedDistortion}, where we assume polynomial aspect ratio and ignore the running time.
	\begin{theorem}[Bounded Aspect Ratio]\label{thm:mainExpectedDistortionAspectRatio}
		Given an $n$-vertex digraph $G=(V,E,w)$ with $m$ edges and polynomial aspect ratio, there is a stochastic embedding into DAGs with expected distortion $O(\log n\cdot\log\log n)$ such that every DAG in the support has at most $O(n\log n+m\log^2n)$ edges.
	\end{theorem}
	
	\paragraph*{Organization.}
	In \Cref{sec:algortihm} we describe the algorithm sampling DAGs $(D_1,D_2)$ satisfying the properties of \Cref{thm:mainExpectedDistortionAspectRatio}, and prove many of its basic properties.
	\Cref{sec:ExpectedDistortion} is then devoted to analyzing the expected distortion.
	Later, in \Cref{sec:RemoveAspectRatio} we remove the aspect ratio assumption.
	Finally, in \Cref{sec:runtime} we explain how to efficiently implement the algorithm, concluding \Cref{thm:mainExpectedDistortion}.
	
	\subsection{Laminar Topological Order Algorithm}
	Our algorithm \texttt{Laminar-Topological-Order} (\Cref{alg:LaminarPartition}) is a simple and natural recursive algorithm. It receives as input a digraph $G=(V,E,w)$ and a cluster $U$ such that $G[U]$ is strongly connected. It returns a laminar topological order of $G[U]$.
	In paragraph \ref{para:GeneralDigraph} below, we describe how to treat a general digraph (not necessarily strongly connected).
	
	The stopping condition of the algorithm is the case where $U$ is a singleton cluster $\{v\}$.
	In this case the algorithm returns the trivial laminar topological order $\left((v),\{U\}\right)$. 
	Otherwise, the algorithm makes a call to the \texttt{Digraph-Partition} algorithm, giving it the diameter $\Delta$ of the cluster $U$ (w.r.t. distances in $G$) and the number of edges $m$ in $G[U]$. Note that \texttt{Digraph-Partition} does not receive the original digraph $G$.\footnote{We send to the \texttt{Digraph-Partition} algorithm the number of edges $m$ of $G[U]$, even though it can be computed locally. This is done for technical reasons. Specifically, during the inductive proof of the expected distortion, it is useful to think of $m$ and $\Delta$ as given parameters.} 
	We get back a subset \EMPH{$S$} of edges (which we call the \EMPH{cut set}). 
	We then remove $S$ from $G[U]$ to obtain the SCCs $\cC$.
	We order these SCCs topologically. 
	Specifically, for $C,C'\in\cC$, if in $G[U]\setminus S$ there is a path from a vertex in $C$ to a vertex in $C'$, then $C<C'$. This defines a partial order over the SCCs $\cC$, which we arbitrarily complete to a total order; denote the SCCs $C_1,C_2,\dots,C_k$ w.r.t. the obtained total order.
	In \Cref{cor:progress}, we guarantee that for every SCC $C_i$ either (1) $|E(G[C_i])|\le\frac m2$, or (2) $\diam(G,C_i)\le\frac\Delta2$. 
	Thus in every SCC we make significant progress. In particular, it follows that the depth of the recursion is bounded by $O(\log (m\Delta))=O(\log n)$.
	
	On every SCC $C_i$ we recursively run the algorithm \texttt{Laminar-Topological-Order} to obtain a laminar topological order $(<_{D_i},\cP_i)$ of $G[C_i]$.
	The algorithm returns the union of all these laminar topological orders.
	Formally: $\cP=\{U\}\cup\bigcup_{i=1}^k\cP_i$, and $<_D$ denotes the order obtained in a lexicographic manner. That is, for $u\in C_i,v\in C_j$, $u<_Dv$ if either $i<j$ or $i=j$ and $u<_{D_i}v$.
	
	\begin{algorithm}[p]
		\caption{$(<_D,\cP)=\texttt{Laminar-Topological-Order}(G=(V,E,w),U)$}\label{alg:LaminarPartition}
		\DontPrintSemicolon
		\SetKwInOut{Input}{input}\SetKwInOut{Output}{output}
		\Input{Digraph $G=(V,E,w)$, subset $U\subseteq V$ such that $G[U]$ is strongly connected}
		\Output{Laminar topological order $(<_D,\cP)$ of the digraph $G[U]$}
		
		\If{$|U|=1$}{
			Let $v$ be the single vertex in $U$\;
			\Return{$\left((v),\{U\}\right)$}
		}
		$\cP\leftarrow \{U\}$\;
		$\Delta\leftarrow\diam(G,U)$\label{line:diamCompute}\;
		$m\leftarrow|E(G[U])|$\;
		$S=\texttt{Digraph-Partition}(G[U],\Delta,m)$\;
		Let $C_1,\dots,C_k$ be all the SCCs of $G[U]\setminus S$ topologically ordered \label{line:SCCpartition}\tcp*{s.t. if there is a path from a vertex in $C_j$ to a vertex in $C_{j'}$, then $j<j'$}
		\For{$i=1$ to $k$\label{line:LaminarFor}}{
			$(<_{D_i},\cP_i)=\texttt{Laminar-Topological-Order}(G,C_i)$\;
			$\cP\leftarrow\cP\cup\cP_i$\;
		}
		Let $<_D$ be the total order obtained by concatenating $<_{D_1}\circ<_{D_2}\circ\dots\circ<_{D_k}$\tcp*{that is, for $u\in C_i,v\in C_j$, $u<_Dv$ if either $i<j$ or $i=j$ and $u<_{D_i}v$}
		\Return{$(<_D,\cP)$}
	\end{algorithm}
	
	\begin{algorithm}[p]
		\caption{$S=\texttt{Digraph-Partition}(G=(V,E,w),\Delta,m)$}\label{alg:DigraphPartition}
		\DontPrintSemicolon
		\SetKwInOut{Input}{input}\SetKwInOut{Output}{output}
		\Input{Digraph $G=(V,E,w)$, scale parameter $\Delta>0$, edge parameter $m\ge|E|$}
		\Output{Cut set $S\subseteq E$, remaining heavy vertex set $\cR$}
		\tcp{The algorithm uses parameters $L,r_\ell,\mu_\ell,\lambda_\ell$. These parameters are determined by $m$ and $\Delta$. See the verbal description of the algorithm for their values.}
		$A\leftarrow V$,~$S\leftarrow\emptyset$\;
		
		\For{$\ell=L,L-1,\dots,1$\label{line:ForL}}{
			$j\leftarrow1$\;
			\While{There is a vertex $x_j\in A$ and $*\in\{+,-\}$ such that \hspace{200pt}\phantom{.}\qquad
				$\frac{m}{\mu_{\ell-1}}\le|B_{G[A]}^{*}(x_{j},r_{\ell-1})|\le|B_{G[A]}^{*}(x_{j},r_{\ell})|\leq\frac{m}{\mu_{\ell}}$\qquad\phantom{.}\label{line:WhileCurve}}{
				Sample  $R \sim \Texp_{[r_{\ell-1},r_\ell ]}(\lambda_\ell)$\label{line:sampleR}\;
				Set $B_{\ell,j}\leftarrow B_{G[A]}^{*}(x_j, R)$\label{line:CurveBall}\;
				Update $S \gets S \cup \delta^*(B_{\ell,j})$\;
				Unmark $A\leftarrow A\setminus B_{\ell,j}$\label{line:unmark}\;		
				$j\leftarrow j+1$\;	
			}
		}
		Set $\cR\leftarrow A$\;
		Add to $S$ all the edges in $G[\cR]$ of weight greater than $r_0$\label{line:RemoveHeavyEdges}\;
		\Return{$S$}
		
	\end{algorithm}
	
	\paragraph{Treating general digraphs. \label{para:GeneralDigraph}}
	In general, we are given a digraph $G=(V,E,w)$ which is not necessarily strongly connected. In this case, we simply run \Cref{alg:LaminarPartition} with input $\left(G=(V,E,w),U\right)$; however, here we artificially set $S=\emptyset$, and begin \Cref{alg:LaminarPartition} at \cref{line:SCCpartition}.
	That is, we immediately begin by ordering $G$’s SCCs topologically, running \texttt{Laminar-Topological-Order} on each SCC recursively, and eventually producing a global laminar topological order. 
	
	The proof of the following observation is straightforward.
	\begin{observation}
		Given a general digraph $G$, \texttt{Laminar-Topological-Order} produces a laminar topological order.
	\end{observation}
	\begin{proof}
		The proof is by induction on the number of vertices $|U|$. The base case where $|U|=1$ is trivial.
		For the inductive step, assume first that $G[U]$ is strongly connected. We run the algorithm on $(G,U)$, obtain a cut set $S$, and SCCs $C_1,\dots,C_k$ of $G[U]\setminus S$ topologically ordered.
		For every SCC $C_i$, by induction, we obtain a laminar topological order $(<_{D_i},\cP_i)$.
		The global laminar topological order $(<_D,\cP)$ is set as $\cP=\{U\}\cup\bigcup_{i=1}^k\cP_i$, and $<_D$ is the concatenation $<_{D_1}\circ\cdots <_{D_k}$.
		We argue that $(<_D,\cP)$ fulfills the $4$ required properties, and thus it is a laminar topological order for $G$.
		\begin{itemize}[noitemsep, topsep=1pt]
			\item Singleton clusters: every vertex $v\in U$ belongs to some SCC $C_i$. By induction $\{v\}\in\cP_i$, and thus $\{v\}\in\cP$, as required.
			\item Hierarchy: we added a new cluster $U$, that contains every other cluster. For any two clusters $X,X'$ (other than $U$), if they belong to the same laminar topological order $(<_{D_i},\cP_i)$, the hierarchy property holds by induction. Otherwise, they belong to different SCCs, and thus are disjoint.
			\item Continuity: the cluster $U$ is contiguous (as it contains all the vertices). Next, consider any other cluster $X$, belonging to some laminar topological order $(<_{D_i},\cP_i)$. By the inductive hypothesis, $X$ is contiguous w.r.t. $<_{D_i}$. As the vertices of $C_i$ are contiguous in $<_D$, $X$ remains contiguous.
			\item Connectivity: the cluster $U$ is strongly connected (by assumption).
			Any other cluster $X$ is strongly connected by the induction hypothesis.
		\end{itemize}
		
		Finally, consider the case where $G[U]$ is not strongly connected.
		This can happen only at the beginning of the algorithm (in the first call), and thus $U=V$. 
		Here, following paragraph \ref{para:GeneralDigraph}, $C_1,\dots,C_k$ are the SCCs of $G$.
		The proof and construction are the same, with the only difference that $V$ is not added as a cluster to $\cP$.
		Accordingly, the entire proof goes through (as we do not need to argue that $V$ is strongly connected).
	\end{proof}
	
	By induction from \Cref{cor:progress} (proved later) it follows:
	\begin{observation}\label{obs:RecursionDepth}
		Every vertex $v\in V$ belongs to at most $O(\log (n\Delta))=O(\log n)$ different clusters in $\cP$. 
	\end{observation}
	
	\paragraph*{Global cut set and induced DAG.}
	\sloppy Denote by \EMPH{$\widetilde{S}$} all the edges that have been removed in all the recursive calls produced by the call $\texttt{Laminar-Topological-Order}(G,V)$.
	Specifically, denote by $\widetilde{S}_i$ all the edges that have been removed in all the recursive calls produced by the call ${\texttt{Laminar-Topological-Order}(G,C_i)}$, then $\widetilde{S}=S\cup\bigcup_{i=1}^k\widetilde{S}_i$ is the \EMPH{global cut set}. 
	
	\begin{claim}\label{clm:DagTotalOrder}
		Set $\brick{D}=G\setminus\widetilde{S}$. $D$ is a DAG, and the total order returned by the laminar topological order $(<_D,\cP)$ is a topological order of $D$. 
	\end{claim}
	\begin{proof}
		The proof is by induction on the number of vertices $|U|$. The base case where $|U|=1$ is trivial.
		For the inductive step, we run the algorithm on $(G,U)$, obtain a cut set $S$, and SCCs $C_1,\dots,C_k$ of $G[U]\setminus S$ topologically ordered.%
		\footnote{If $G$ is not strongly connected (possible only at the start of the algorithm), then following paragraph \ref{para:GeneralDigraph}, we let $S=\emptyset$, and let $C_1,\dots,C_k$ be its SCCs.\label{foot:RefToParagraphNotConnected}}
		Let $s,t\in U$ such that $s<_Dt$. It suffices to show that $t\not\rightsquigarrow_{D}s$. Let $i,j\in[k]$ such that $s\in C_i$ and $t\in C_j$.
		
		Consider first the case where $i\ne j$. As $s<_Dt$, it necessarily holds that $i<j$. As $C_i$ comes before $C_j$ in the topological order of $G[U]\setminus S$, there is no path in $G[U]\setminus S$ from a vertex in $C_j$ to a vertex in $C_i$. In particular, $t\not\rightsquigarrow_{G[U]\setminus S}s$. 
		As $D=G[U]\setminus \widetilde{S}$ is a subgraph of $G[U]\setminus S$, it follows that $t\not\rightsquigarrow_{D}s$.
		
		Next, consider the case where $i=j$. By the inductive hypothesis, it holds that $D_i=G[C_i]\setminus \widetilde{S}_i$ is a DAG, and the order returned by the laminar topological order $(<_{D_i},\cP_i)$ is a topological order of $D_i$. 
		As $s<_Dt$, it follows that $s<_{D_i}t$, and thus by induction $t\not\rightsquigarrow_{D_i}s$.
		Suppose for the sake of contradiction that there is a path $\pi$ from $t$ to $s$ in $D$. 
		As there are no additional edges between $C_i$ vertices in $D$ (compared with $D_i$), 
		$\pi$ necessarily contains a vertex $v\notin C_i$.
		Note that $\pi$ is also a path in $G[U]\setminus S$.
		Thus $G[U]\setminus S$ contains a path from the SCC $C_i$ to $v$, and from $v$ to the SCC $C_i$. It follows that $v\in C_i$, a contradiction.
		We conclude that $t\not\rightsquigarrow_{D}s$, as required.	
	\end{proof}
	
	As \Cref{alg:LaminarPartition} is recursive on the SCCs, and thus does not add to the cut set additional edges between the SCCs, we observe:
	\begin{observation}\label{obs:noCutThenSurvive}
		Consider an edge $(u,v)\in G[U]$ such that $(u,v)\notin S$, and $u,v$ belong to different SCCs of $G\setminus S$. Then $(u,v)\in D$.
	\end{observation}

	\subsection{Single Partition Algorithm}
	The \texttt{Digraph-Partition} algorithm is described in \Cref{alg:DigraphPartition}.
	This algorithm is the basis for the creation of our laminar topological order, which leads to the creation of the DAGs.
	The algorithm is similar to the cutting procedure from \cite{BFHL25} (which in turn was used to create directed low-diameter decompositions).
	The algorithm receives as input a weighted digraph $G=(V,E,w)$, scale $\Delta>0$, and edge parameter $m$. 
	In our implementation, $\Delta$ is the weak diameter of $G$ (in some graph not given to the algorithm), while $m$ is the number of edges $|E|$.
	Nevertheless, we will think of them as general parameters, with the guarantee $m\ge|E|$. This freedom will be useful for our inductive arguments.
	Likewise, the algorithm will only be executed on strongly connected digraphs. However, for the same considerations, we will not assume that during the analysis.
	
	We denote by $\brick{A}\subseteq V$ the set of \EMPH{active vertices}, and by $\brick{S}\subseteq E$ the set of \EMPH{cut edges}.
	Initially, all the vertices are active $A\leftarrow V$, and the cut set is empty $S\leftarrow \emptyset$. 
	We say that an \EMPH{edge is active} if both its endpoints are active.
	The algorithm repeatedly carves a ball $B_{\ell,j}$ of (a random) radius $R$. The vertices of the carved ball $B_{\ell,j}$ cease to be active, while all the (one of two options: ingoing or outgoing) edges from $B_{\ell,j}$ join the cut set $S$.
	There are $L$ different \EMPH{phases} of ball carving, where in each phase the algorithm uses different parameters to find an appropriate center (as a function of ball density) and a different parameter to sample the radii. In each phase, as long as the algorithm can find centers fulfilling the density requirement, it carves balls around these centers. We call each such carving an \EMPH{iteration} of the algorithm.
	Once no appropriate center is found, the algorithm moves to the next phase.
	After the $L$ phases are exhausted, the algorithm denotes the set of remaining active vertices \EMPH{$\cR$}. It adds to $S$ all the edges between $\cR$ vertices of weight at least $r_0=\Omega(\frac{\Delta}{\log m})$, and returns the obtained cut set $S$.
	It is guaranteed (see \Cref{clm:PartitionIterationRDelta}) that every SCC in $G[\cR]\setminus S$ has diameter at most $\frac{\Delta}{2}$ (while there is no guarantee on the diameter of other SCCs in $G\setminus S$).
	
	The heart of the algorithm is the random choice of the radii of the carved balls. 
	The algorithm uses the following parameters (depending only on $m$ and $\Delta$):
	\begin{itemize}
		\item $\brick{L} = \lceil\log\log (m+1)\rceil$,
		\item $\brick{\delta}=(\log(m+1))^{-2}$,
		\item $\brick{r_0} := \frac{\Delta}{2^{L+4}}$ and $\brick{r_\ell} := r_{\ell-1} + \frac{\Delta}{2^{L-\ell+4}} + \frac{\Delta}{8 L}$ (for $1 \leq \ell \leq L$),
		\item $\brick{\mu_\ell} := 2^{2^{L-\ell}}$ (for $0 \leq \ell \leq L$),
		\item $\brick{\lambda_\ell}=\frac{1}{r_{\ell}-r_{\ell-1}}\cdot\ln(\frac{2\mu_{\ell}}{\delta})$.
	\end{itemize}
	Note that $L,\delta,r_\ell,\mu_\ell,\lambda_\ell$ are functions of $m$ and $\Delta$. In most cases the values of $m$ and $\Delta$ will be clear from the context. Where this is not the case, we will address the issue explicitly. 
	
	\begin{observation}\label{obs:rellSize}
		$r_L=\frac\Delta4-\frac{\Delta}{2^{L+4}}<\frac\Delta4$.
	\end{observation}
	\begin{proof}
		\begin{align*}
			r_{L} & =r_{0}+\sum_{\ell=L}^{1}\left(\frac{\Delta}{2^{L-\ell+4}}+\frac{\Delta}{8L}\right)=r_{0}+\frac{\Delta}{8}+\frac{\Delta}{8}\cdot\sum_{\ell=1}^{L}\frac{1}{2^{\ell}}\\
			& =r_{0}+\frac{\Delta}{8}+\frac{\Delta}{8}\cdot\left(1-\frac{1}{2^{L}}\right)=\frac{\Delta}{4}+r_{0}-\frac{\Delta}{2^{L+3}}=\frac{\Delta}{4}-\frac{\Delta}{2^{L+4}}~.
		\end{align*}
	\end{proof}
	\begin{observation}\label{obs:lambdaUB}
		$\lambda_\ell\le\frac{64}{\Delta}\cdot L\cdot\log\mu_{\ell}$.
	\end{observation}
	\begin{proof}
		\begin{align*}
			\lambda_{\ell} & =\frac{1}{r_{\ell}-r_{\ell-1}}\cdot\ln\left(\frac{2\mu_{\ell}}{\delta}\right)=\frac{1}{\frac{\Delta}{2^{L-\ell+4}}+\frac{\Delta}{8L}}\cdot\left(\ln(2\mu_{\ell})+\ln\frac{1}{\delta}\right)\\
			& \overset{(*)}{\le}\frac{1}{\Delta\cdot\max\left\{ \frac{1}{2^{L-\ell+4}},\frac{1}{8L}\right\} }\cdot\left(2^{L-\ell+1}+2L\right)\\
			& \le\frac{2}{\Delta}\cdot\left(2^{L-\ell}+L\right)\cdot16\cdot\min\left\{ 2^{L-\ell},L\right\} \\
			& \overset{(**)}{\le}\frac{64}{\Delta}\cdot L\cdot2^{L-\ell}=\frac{64}{\Delta}\cdot L\cdot\log\mu_{\ell}~,
		\end{align*}
		inequality $^{(*)}$ follows as $\ln\frac{1}{\delta}=\ln\left((\log(m+1))^{2}\right)=2\cdot\ln(\log(m+1))\le2\cdot L$.
		Inequality $^{(**)}$ follows since $(a+b)\cdot\min\{a,b\}\le2\cdot\max\{a,b\}\cdot\min\{a,b\}=2ab$. 
	\end{proof}

	The algorithm has $L$ phases, where we use $\ell$ to denote the current phase. Initially $\ell=L$, and it gradually decreases, until in the last phase $\ell=1$.
	There are radii parameters $r_\ell>r_{\ell-1}>\dots>r_1>r_0$, where during the $\ell$'th phase we sample radii from $[r_{\ell-1},r_\ell]$.
	There are also density parameters $\mu_\ell>\mu_{\ell-1}>\dots>\mu_1>\mu_0$.
	During the $\ell$'th phase, we look for a center vertex $x$ such that $\frac{m}{\mu_{\ell-1}}\le|B_{G[A]}^{*}(x_{j},r_{\ell-1})|\le|B_{G[A]}^{*}(x_{j},r_{\ell})|\leq\frac{m}{\mu_{\ell}}$ (for $*\in\{+,-\}$).
	When such a center vertex is found, we sample $R\in[r_{\ell-1},r_\ell]$, and carve the ball $B_{\ell,j}\leftarrow B_{G[A]}^{*}(x_j, R)$. All the vertices in $B_{\ell,j}$ cease to be active, while all the edges in $\delta^*(B_{\ell,j})$ join $S$.
	Once there is no such center vertex $x_j$, the $\ell$'th phase ends, and we move to the $\ell+1$'th phase.
	
	The sampling of the radius $R$ is done using \EMPH{truncated exponential distribution}:
	an exponential distribution conditioned on the event that the outcome lies in a certain interval.
	This distribution is often used in the context of metric embeddings and low-diameter decompositions; see, e.g., \cite{Bar96,Bartal04,AGGNT19,Fil19padded,BFT24,FFIKLMZ24,CF25}.
	It is useful because, while it gives a strict bound on the sampled radii, it is still “almost” memoryless.
	Formally, the $[\theta_1,\theta_2]$-truncated exponential distribution with
	parameter $\lambda$, denoted \EMPH{$\Texp_{[\theta_1, \theta_2]}(\lambda)$}, is the exponential distribution with parameter $\lambda$ conditioned on the event that the outcome is in $[\theta_1, \theta_2]$. The density function is
	$f(y)= \frac{ \lambda\, e^{-\lambda y} }{e^{-\lambda \theta_1} - e^{-\lambda \theta_2}}$ for $y \in [\theta_1, \theta_2]$.
	During the $\ell$'th phase we sample $R$ from $\Texp_{[r_{\ell-1},r_\ell ]}(\lambda_\ell)$.
	One should think of the $\delta$ parameter in the setting of $\lambda_\ell$ as an error parameter. Intuitively, we analyze the process w.r.t. the exponential distribution instead of the truncated exponential distribution, and $\delta$ is used to bound the difference between the two distributions.
	
	We continue by describing the basic properties of \Cref{alg:DigraphPartition}. The proofs are similar to \cite{BFHL25}.
	Denote by $\brick{\cC_\ell}=\cup_{j\ge1}B_{\ell,j}$ the set of all the vertices belonging to balls that were carved during the $\ell$'th phase of the for loop (\cref{line:ForL} of \Cref{alg:DigraphPartition}).
	\begin{claim}\label{clm:PartitionIteration}
		Every SCC $C$ of $G\setminus S$ is fully contained in one of the sets $\cC_L,\cC_{L-1},\dots,\cC_2,\cC_1,\cR$. Further, if $C\subseteq \cC_\ell$, then $C$ is fully contained in some $B_{\ell,j}$, and $|E(C)|\le\frac{m}{\mu_\ell}\le\frac m2$.
	\end{claim}
	\begin{proof}
		For simplicity of notation, denote $\cC_0=\cR$.
		During the $\ell$'th phase of the for loop (\cref{line:ForL}), we carve balls $B_{\ell,1},B_{\ell,2},\dots,B_{\ell,j_\ell}$. 
		We argue by induction on the creation time that every SCC in $G\setminus S$ that contains some vertex in $B_{\ell,j}$ is fully contained in $B_{\ell,j}$. This implies the first proposition of the claim.
		The base of the induction is before we start carving balls and thus holds trivially. 
		Suppose that the induction holds just before we carve the ball $B_{\ell,j}\leftarrow B_{G[A]}^{*}(x_j, R)$.
		Suppose that $*=+$ (the case $*=-$ is symmetric).
		All the outgoing edges from $B_{\ell,j}$ are added to $S$. 
		Consider a SCC $C$ of $G\setminus S$ which contains some vertex 
		$v\in B_{\ell,j}$. 
		By the induction hypothesis, all the vertices in $C$ are active at the time of the creation of $B_{\ell,j}$ (otherwise $C$ would not contain $v$).
		All the outgoing edges from $B_{\ell,j}$ to other active vertices are added to $S$. Thus $C$ cannot contain any active vertices outside $B_{\ell,j}$. We conclude that $C\subseteq B_{\ell,j}$, as required.
		
		Next, consider a SCC $C\subseteq B_{\ell,j}= B_{G[A]}^{*}(x_j, R)$.
		As $R\le r_\ell$, it follows that $|E(C)|\le  |B_{G[A]}^{+}(x_j, R)|\le  |B_{G[A]}^{+}(x_j, r_\ell)|\le\frac{m}{\mu_\ell}$.	
	\end{proof}
	
	\begin{observation}\label{obs:NumberOfBalls}
		During the $\ell$'th phase of the for loop (\cref{line:ForL} of \Cref{alg:DigraphPartition}), at most $\mu_{\ell-1}$ balls are carved.
	\end{observation}
	\begin{proof}
		We say that an edge $(u,v)$ is active iff both endpoints $u,v$ are active.
		Consider the $\ell$ phase of the For loop (\cref{line:ForL}).
		At iteration $j$ of the While loop (\cref{line:WhileCurve}), we carve the ball $B_{\ell,j}\leftarrow B_{G[A]}^{*}(x_j, R)$ by sampling a radius $R\in[r_{\ell-1},r_\ell]$.
		By the choice of the center $x_j$ and the radius $R$, it holds that $|B_{G[A]}^{*}(x_j, R)|\ge  |B_{G[A]}^{*}(x_j, r_{\ell-1})|\ge\frac{m}{\mu_{\ell-1}}$.
		In particular, after this $j$'th iteration, at least $\frac{m}{\mu_{\ell-1}}$ previously active edges become inactive.
		As at the beginning of the $\ell$'th phase there are at most $m$ active edges, it follows that during the $\ell$'th phase at most $\mu_{\ell-1}$ balls are carved, as required.
	\end{proof}
	
	\begin{claim}\label{clm:PartitionIterationRDelta}
		For every SCC $C$ in $G\setminus S$ that is contained in $\cR$, it holds that $\diam(C,G)\le2r_L\le\frac{\Delta}{2}$.
	\end{claim}
	\begin{proof}
		Denote by 
		\[
		\cH=\left\{ v\in\cR\;\middle|\;|B_{G}^{+}(v,r_L)|>\frac{m}{2}\ \text{and}\ |B_{G}^{-}(v,r_L)|>\frac{m}{2}\right\} 
		\]
		the set of $\cR$ vertices whose both in- and out-balls of radius $r_L$ contain more than $\frac m2$ edges. 
		Note that for every $u,v\in \cH$, the balls $B_{G}^{+}(u,r_L)$ and $B_{G}^{-}(v,r_L)$ intersect, and thus $d_G(u,v)\le2r_L$. 
		We argue that for every vertex $v\in \cR\setminus\cH$, the SCC of $v$ in $G\setminus S$ is $\{v\}$.
		In other words, the SCCs of $G\setminus S$ that are fully
		contained in $\cR$ are either singletons or fully contained in $\cH$. In both cases, the diameter is at most $2r_L$, which by \Cref{obs:rellSize} is bounded by $\frac\Delta2$.
		
		Consider $v\in\cR\setminus\cH$, and suppose that $|B_{G}^{+}(v,r_L)|\leq\frac{m}{2}$ (the case $|B_{G}^{-}(v,r_L)|\leq\frac{m}{2}$ is symmetric).
		We will prove that in $G[\cR]\setminus S$, $v$ has no outgoing edges, and hence by \Cref{clm:PartitionIteration} the SCC of $v$ in $G\setminus S$ is $\{v\}$.
		Denote by $A_\ell$ the set of active vertices at the beginning of the $\ell$'th phase of the algorithm. That is, $A_L=V$. Denote also $A_0=\cR$.
		We argue by induction on $\ell$ that  $|B_{G[A_\ell]}^{+}(v,r_{\ell})|\leq\frac{m}{\mu_{\ell}}$.
		The base case ($\ell =L$): note that $\mu_L=2$, and by our assumption $|B_{G}^{+}(v,r_L)|\leq\frac{m}{2}=\frac{m}{\mu_{L}}$.
		Next, suppose the induction holds for $\ell+1$, i.e., $|B_{G[A_{\ell+1}]}^{+}(v,r_{\ell+1})|\leq\frac{m}{\mu_{\ell+1}}$.
		As $v\in\cR\subseteq A_\ell$, $v$ does not belong to any of the balls carved during the $(\ell+1)$'th phase.
		It follows that $|B_{G[A_{\ell}]}^{+}(v,r_{\ell})|<\frac{m}{\mu_{\ell}}$, as otherwise 
		$\frac{m}{\mu_{\ell}}\le|B_{G[A_{\ell}]}^{+}(v,r_{\ell})|\le|B_{G[A_{\ell}]}^{+}(v,r_{\ell+1})|\le|B_{G[A_{\ell+1}]}^{+}(v,r_{\ell+1})|\leq\frac{m}{\mu_{\ell+1}}$,	
		and thus $v$ would have been chosen as a ball center in the while loop.
		
		Applying our inductive hypothesis at $\ell=0$, we conclude that $|B_{G[\cR]}^{+}(v,r_{0})|=|B_{G[A_{0}]}^{+}(v,r_{0})|\leq\frac{m}{\mu_{0}}<1$, where the last inequality holds as $\mu_{0}=2^{2^{L}}\ge2^{2^{\log\log(m+1)}}=m+1$.
		That is, all the outgoing edges from $v$ in $G[\cR]$ have weight greater than $r_0$.
		As all these edges are added to the cut set $S$ in \cref{line:RemoveHeavyEdges}, we conclude that in $G\setminus S$, $v$ has no outgoing edges, as required.
	\end{proof}
	
	Using \Cref{clm:PartitionIteration} and \Cref{clm:PartitionIterationRDelta} we conclude:
	\begin{corollary}\label{cor:progress}
		For every SCC $C$ of $G\setminus S$ either (1) $|E(G[C])|\le\frac m2$, or (2) $\diam(G,C)\le\frac\Delta2$.
	\end{corollary}

	\subsection{DAG construction}
	After executing the algorithm $\texttt{Laminar-Topological-Order}(G=(V,E,w),V)$, we obtain a laminar topological order $(<_D,\cP)$. In addition, by \Cref{clm:DagTotalOrder} we have a cut set $\widetilde{S}\subseteq E$ such that $D=G\setminus\widetilde{S}$ is a DAG, and $<_D$ is a topological order w.r.t. $D$.
	Following \cite{AHW25}, we create two DAGs $D_1,D_2$.
	$D_1$ will be an amplified version of $D$, while $D_2$ will be a DAG constructed with the reversed order of $<_D$.
	
	An important component of our construction is a $2$-hop $1$-directed spanner for the path graph. 
	It will be used to ensure that for every cluster $C\in \cP$ and $u,v\in C$, $\min\{d_{D_1}(s,t),d_{D_2}(s,t)\}\le 2\Delta_{C}$ (see \Cref{clm:DistanceUsing2HopSpanner}).
	The following lemma is well known \cite{AS87,BTS94}. We state the result in the context of our paper, and sketch its proof.
	\begin{lemma}[\cite{AS87,BTS94}]\label{lem:2hopSpanner}
		Consider the directed path $(v_1,v_2,\dots,v_n)$. There is a directed set of $O(n\log n)$ edges $H\subseteq\{(v_i,v_j)\mid i<j\}$, such that for every $i<j$
		either $(v_i,v_j)\in H$,
		or there is $q\in(i,j)$ such that $(v_i,v_q),(v_q,v_j)\in H$.
	\end{lemma}
	\begin{proof}
		Fix $k=\lceil\frac n2\rceil$. For every $i<k$, add the edge $(i,k)$ to $H$, while for every $i>k$, add the edge $(k,i)$ to $H$.
		Continue this process recursively on the subpaths $(v_1,\dots,v_{k-1})$, $(v_{k+1},\dots,v_n)$.
		The stopping condition of the recursion is when we get a path consisting of a single vertex, and then we add no edges.
		
		Clearly, as the depth of the recursion is $\log n$, and each vertex is responsible for at most a single edge in each level of the recursion, it holds that $|H|=O(n\log n)$.
		The proof of the $2$-hop path property is done by induction. 
		Consider $i<j$. 
		If either $i=k$ or $j=k$, then $(i,j)\in H$ and we are done.
		Otherwise, if $i<k<j$, then $(i,k),(k,j)\in H$, as required.	
		Finally, it holds that either $i,j< k$, or $i,j> k$.
		In both cases, using the inductive hypothesis, either $(i,j)\in H$, or there is $q\in(i,j)$ such that $(i,q),(q,j)\in H$.
	\end{proof}
	
	\paragraph*{Construction of $D_1$.}
	Let $(v_1,v_2,\dots,v_n)$ be the vertices $V$ ordered w.r.t. the total order $<_D$.
	The DAG $D_1$ contains the following edges:
	\begin{itemize}
		\item Let $H_1$ be the set of edges obtained from \Cref{lem:2hopSpanner} w.r.t. $(v_1,v_2,\dots,v_n)$. We add to $D_1$ the following set of edges:
		$$E_{H_1}=\left\{(v_i,v_j)\in H_1\mid v_i\rightsquigarrow_{G}v_j \right\}~.$$
		That is, we add to $D_1$ every edge $(v_i,v_j)\in H_1$ such that there is a $v_i$ to $v_j$ path in $G$.
		
		\item \sloppy For every edge $e=(u,v)\in D=E\setminus\widetilde{S}$, and every two disjoint clusters $C_u,C_v\in\cP$ such that $u\in C_u$ and $v\in C_v$, we add to $D_1$ the edge $(C_u^\last,C_v^\first)$.
		Recall that for a cluster $C\in\cP$, $C^\first$ and $C^\last$ are the first and last vertices in $C$ w.r.t. $<_D$, respectively.
		Note that as by the singleton property (of \Cref{def:Laminar}) $\{u\},\{v\}\in\cP$, the edges $e=(u,v)$, $(u,C_v^\first)$, and $(C_u^\last,v)$ are also added to $D_1$. It follows that $D\subseteq D_1$.
		See illustration below:	
		\begin{center}
			\includegraphics[scale=0.8]{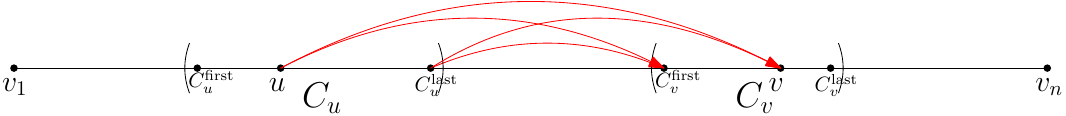}
		\end{center}
		\label{page:4AddedEdges}
	\end{itemize}
	The weight of every edge $(u,v)$ added to $D_1$ is simply $d_G(u,v)$.
	
	\paragraph*{Construction of $D_2$.} The construction of the DAG $D_2$ is similar to $D_1$, but w.r.t. the reverse order, and using only the first component. Formally, let $H_2$ be the set of edges obtained from \Cref{lem:2hopSpanner} w.r.t. $(v_n,v_{n-1},\dots,v_1)$. We add to $D_2$ the following set of edges:
	$$E_{H_2}=\left\{(v_j,v_i)\in H_2\mid v_j\rightsquigarrow_{G}v_i \right\}.$$
	The weight of every edge $(u,v)$ added to $D_2$ is simply $d_G(u,v)$.
	
	We continue by proving some basic properties regarding $D_1$ and $D_2$. First we argue that $D_1,D_2$ are indeed DAGs.
	
	\begin{claim}\label{clm:zD1D2Dags}
		$D_1$ and $D_2$ are DAGs.
	\end{claim}
	\begin{proof}
		We argue that for every edge $(v_i,v_j)\in D_1$ it holds that $i<j$. It follows that $D_1$ is indeed a DAG. The proof for $D_2$ is similar.
		Clearly, this proposition holds for all the edges in $H_1$, and thus also for the edges in $E_{H_1}\subseteq H_1$.
		Next, consider an edge $(u,v)\in D$. By \Cref{clm:DagTotalOrder}, $D$ is a DAG and $u<_Dv$.
		Consider two clusters $C_u,C_v\in\cP$ such that $C_u$ contains $u$ but not $v$, and $C_v$ contains $v$ but not $u$.
		By the hierarchy property (property \ref{req:hierarchy} of \Cref{def:Laminar}), $C_u$ and $C_v$ are disjoint. By the continuity property (property \ref{req:continuity}), and as $u<_Dv$, all the vertices in $C_u$ appear before all the vertices in $C_v$ w.r.t. $<_D$.
		It follows that the added edges $(u,C_v^\first),(C_u^\last,v),(C_u^\last,C_v^\first)$ respect our proposition.
		The claim follows.
	\end{proof}
	
	Next, we argue that the two DAGs $D_1,D_2$ are dominating.
	\begin{claim}\label{clm:zDominating}
		For every $s,t\in V$, $d_G(s,t)\le d_{D_1}(s,t),d_{D_2}(s,t)$.
	\end{claim} 
	\begin{proof}
		We begin by arguing that for every edge $(u,v)\in D_1\cup D_2$ it holds that $u\rightsquigarrow_{G}v$.
		This clearly holds for all the edges in $E_{H_1}\cup E_{H_2}$ (as we only added edges $(u,v)\in H_1\cup H_2$ such that $u\rightsquigarrow_{G}v$).
		All the other edges were added due to some edge $(u,v)\in D=E\setminus\widetilde{S}$ and two clusters $C_u,C_v\in \cP$.
		Note that by the connectivity property (property \ref{req:connectivty} of \Cref{def:Laminar}) both $C_u$ and $C_v$ are strongly connected. 
		As $(u,v)\in D\subseteq G$, for every $u'\in C_u$ and $v'\in C_v$, $u'\rightsquigarrow_G v'$. Thus our proposition holds.
		
		By the transitivity of the reachability relation in a digraph, it follows that for every $s,t\in V$, if either $s\rightsquigarrow_{D_1}t$, or $s\rightsquigarrow_{D_2}t$, then $s\rightsquigarrow_{G}t$. That is, if $t$ is reachable from $s$ in either $D_1$ or $D_2$, then $t$ is reachable from $s$ in $G$.
		Finally, the claim follows by the triangle inequality, as for every edge $(u,v)\in D_1\cup D_2$ we used the weight $d_G(u,v)$.
	\end{proof}
	
	We finish the section by bounding the size of the created DAGs:
	\begin{claim}\label{clm:SparsityOfDAGs}
		Let $m$ be the number of edges in $G$. Then $|E(D_1)|,|E(D_2)|\le O(n\log n+m\log^2(n\Delta))=O((n+m)\log^2n)$.
	\end{claim}
	\begin{proof}
		By \Cref{lem:2hopSpanner}, $|E_{H_1}|,|E_{H_2}|\le O(n\log n)$.
		By \Cref{obs:RecursionDepth} every vertex belongs to at most $O(\log n\Delta)$ clusters in $\cP$.
		It follows that for every edge $(u,v)\in D=E\setminus\widetilde{S}$ we add at most $O(\log^2(n\Delta))$ edges to $D_1$.
		The claim follows.	
	\end{proof}

\section{Expected Distortion Analysis}\label{sec:ExpectedDistortion}
For every pair $s,t\in V$, either $d_{D_1}(s,t)=\infty$ or $d_{D_2}(s,t)=\infty$ (or both).
Recall that $C_{s,t}=\Parent_{\cP}(s,t)$ is the minimal cluster (w.r.t. inclusion order) containing both $s,t$, and that $\Delta_{st}=\Delta_{C_{s,t}}$ is the diameter of $C_{s,t}$.
Next, we argue that, using only the edges from $E_{H_1},E_{H_2}$, we can bound one of $d_{D_1}(s,t),d_{D_2}(s,t)$ by $2\Delta_{st}$.
\begin{claim}\label{clm:DistanceUsing2HopSpanner}
	$\forall s,t\in V$, $\min\{d_{D_1}(s,t),d_{D_2}(s,t)\}\le 2\Delta_{st}$.
\end{claim}
\begin{proof}
	Consider a pair $s,t\in V$. If $s\not\rightsquigarrow_G t$, then there is no cluster $C\in\cP$ containing both $s,t$ (as such a cluster would not fulfill the connectivity property of \Cref{def:Laminar}). Thus $\Parent_{\cP}(s,t)=\emptyset$, and $\Delta_{st}=\infty$. The claim holds trivially. We will thus assume that $s\rightsquigarrow_G t$.
	
	Suppose first that $s<_Dt$.
	Denote $C_{st}=\Parent_{\cP}(s,t)$.
	By \Cref{lem:2hopSpanner}, either (1) $(s,t)\in H_1$, or (2) there exists a vertex $u$ such that $s<_Du<_Dt$ and $(s,u),(u,t)\in H_1$.
	In case (2), by the continuity property of \Cref{def:Laminar}, $u\in C_{st}$. Further, by the connectivity property of \Cref{def:Laminar} $s\rightsquigarrow_G u$ and $u\rightsquigarrow_G t$.
	It follows that $(s,u),(u,t)\in D_1$. Note that the weight of all these edges is at most $\Delta_{st}$, and hence $d_{D_1}(s,t)\le 2\Delta_{st}$.
	By a symmetric argument, in the case $t<_Ds$ it holds that $d_{D_2}(s,t)\le 2\Delta_{st}$.
\end{proof}

Fix a pair $s,t\in V$ such that $s\rightsquigarrow_G t$, and let $\pi$ be the shortest $s$–$t$ path in $G$.
In our analysis, we will bound the distance in $D_2$ using only the diameter of the minimal cluster $C_{st}$ containing both $s$ and $t$. However, in $D_1$ the situation is considerably more delicate.
For instance, it might be that we delete a set of edges $S$, disjoint from $\pi$, such that $s$ and $t$ belong to different SCCs $C_s,C_t$ of $G\setminus S$, respectively. 
Note that as $\pi\subseteq G\setminus S$, it still holds that $d_{G\setminus S}(s,t)=d_G(s,t)$. Further, as no edge of $\pi$ was deleted, it necessarily holds that $s<_Dt$ (\Cref{obs:noCutThenSurvive}).
We will bound the following term inductively:
\[
\brick{\alpha(s,t)}:=d_{D_{1}}(s,t)\cdot\mathds{1}[s\le_{D}t]+3\Delta_{st}\cdot\mathds{1}[s\ge_{D}t]~.
\]
We use $s\le_{D}t$ to denote either $s=t$, or $s<_{D}t$.
Note that $\alpha$ is reflexive: for every vertex $v$, $d_{D_1}(v,v)=0$ (by definition), and $\Delta_{vv}=0$ (by the singleton property of \Cref{def:Laminar}), hence $\alpha(v,v)=0$.
Observe that, in order to prove \Cref{thm:mainExpectedDistortionAspectRatio}, it is enough to bound $\E[\alpha(s,t)]$.
\begin{observation}\label{obs:alphaDominating}
	For every $s,t\in V$,
	\[
	d_{D_{1}}(s,t)\cdot\mathds{1}[s\le_{D}t]+d_{D_{2}}(s,t)\cdot\mathds{1}[s\ge_{D}t]\le\alpha(s,t).
	\]
	In addition, if $s$ and $t$ belong to the same SCC with diameter $\Delta$, then $\alpha(s,t)\le3\Delta$.
\end{observation}
\begin{proof}
	By \Cref{clm:DistanceUsing2HopSpanner}, $d_{D_{2}}(s,t)\cdot\mathds{1}[s>_{D}t]\le 2\Delta_{st}$, and hence
	\[
	d_{D_{1}}(s,t)\cdot\mathds{1}[s<_{D}t]+d_{D_{2}}(s,t)\cdot\mathds{1}[s>_{D}t]\le \alpha(s,t).
	\]
	Next, suppose that $s$ and $t$ belong to the same SCC with diameter $\Delta$. Then $\Delta_{st}\le\Delta$. By \Cref{clm:DistanceUsing2HopSpanner},
	\[
	\alpha(s,t)=d_{D_{1}}(s,t)\cdot\mathds{1}[s\le_{D}t]+3\Delta_{st}\cdot\mathds{1}[s\ge_{D}t]\le2\Delta\cdot\mathds{1}[s\le_{D}t]+3\Delta\cdot\mathds{1}[s\ge_{D}t]\le3\Delta.
	\]
\end{proof}

What makes $\alpha(s,t)$ useful is that it respects a certain type of triangle inequality.

\begin{lemma}\label{lem:alphaAppliedRecursively}
	Suppose that there are two different SCCs $C_s,C_t$ of $G[U]\setminus S$ such that there is an edge $(u,v)\in \pi$ with $(u,v)\notin S$, $\pi[s,u]\subseteq C_s$, and $\pi[v,t]\subseteq C_t$. Then 
	$$\alpha(s,t)\le\alpha(s,u)+d_G(u,v)+\alpha(v,t)~.$$
\end{lemma}
\begin{proof}
	We have $\pi=\pi[s,u]\circ(u,v)\circ\pi[v,t]$.	
	By \Cref{obs:noCutThenSurvive} $(u,v)\in D$. In particular, $u<_Dv$, and by the continuity property of \Cref{def:Laminar}, as $C_s,C_t\in\cP$, $s<_Dt$. 
	Hence $\alpha(s,t)=d_{D_1}(s,t)$.
	The algorithm proceeds independently on $C_s,C_t$, and orders the vertices in $C_s,C_t$ accordingly.
	We continue by case analysis (see \Cref{fig:cases} for an illustration of the cases, and the path used in each case). We next make observations for each case:
	\begin{enumerate}
		\item $s\le_{D}u$ and $v\le_{D}t$: As $(u,v)\in D_1$, by the triangle inequality,
		$$d_{D_1}(s,t)\le d_{D_{1}}(s,u)+d_{G}(u,v)+d_{D_{1}}(v,t)~.$$
		
		\item $s\le_{D}u$ and $v>_{D}t$: recall that  $C_{vt}=\Parent_{\cP}(v,t)$ denotes the minimal cluster in $\cP$ containing both $v,t$, and that $\Delta_{vt}=\diam(C_{vt},G)$. The algorithm added the edge $(u,C^{\first}_{vt})$, with weight
		$w(u,C^{\first}_{vt})=d_G(u,C^{\first}_{vt})\le d_G(u,v)+d_G(v,C^{\first}_{vt})\le d_G(u,v)+\Delta_{vt}$.
		By \Cref{clm:DistanceUsing2HopSpanner}, $d_{D_1}(C^{\first}_{vt},t)\le2\Delta_{vt}$.
		Thus,
		$$d_{D_1}(s,t)\le d_{D_{1}}(s,u)+w(u,C^{\first}_{vt})+d_{D_{1}}(C^{\first}_{vt},t)\le  d_{D_{1}}(s,u)+ d_G(u,v)+3\Delta_{vt}~.$$
		
		\item $s>_{D}u$ and $v\le_{D}t$: Let $C_{su}=\Parent_{\cP}(s,u)$ with diameter $\Delta_{su}$. 
		The algorithm added the edge $(C^{\last}_{su},v)$, with weight
		$w(C^{\last}_{su},v)=d_G(C^{\last}_{su},v)\le d_G(C^{\last}_{su},u)+d_G(u,v)\le \Delta_{su}+d_G(u,v)$.
		By \Cref{clm:DistanceUsing2HopSpanner}, $d_{D_1}(s,C^{\last}_{su})\le2\Delta_{su}$.
		Thus,
		$$d_{D_1}(s,t)\le d_{D_{1}}(s,C^{\last}_{su})+w(C^{\last}_{su},v)+d_{D_{1}}(v,t)\le  3\Delta_{su}+ d_G(u,v)+d_{D_{1}}(v,t)~.$$
		
		\item $s>_{D}u$ and $v>_{D}t$: Here we will use both the clusters $C_{su}$ and $C_{vt}$ of respective diameters $\Delta_{su}$ and $\Delta_{vt}$.
		The algorithm added the edge $(C^{\last}_{su},C^{\first}_{vt})$, of weight $w(C^{\last}_{su},C^{\first}_{vt})=d_G(C^{\last}_{su},C^{\first}_{vt})\le d_G(C^{\last}_{su},u)+d_G(u,v)+d_G(v,C^{\first}_{vt})\le \Delta_{su}+d_G(u,v)+\Delta_{vt}$.
		By \Cref{clm:DistanceUsing2HopSpanner}, $d_{D_1}(s,C^{\last}_{su})\le2\Delta_{su}$, and $d_{D_1}(C^{\first}_{vt},t)\le2\Delta_{vt}$.
		We conclude that 
		\[
		d_{D_{1}}(s,t)\le d_{D_{1}}(s,C_{su}^{\last})+w(C_{su}^{\last},C_{vt}^{\first})+d_{D_{1}}(C_{vt}^{\first},t)\le3\Delta_{su}+d_{G}(u,v)+3\Delta_{vt}~.
		\]

	\end{enumerate}
	\begin{figure}[t]
		\begin{center}
			\includegraphics[width=0.9\textwidth]{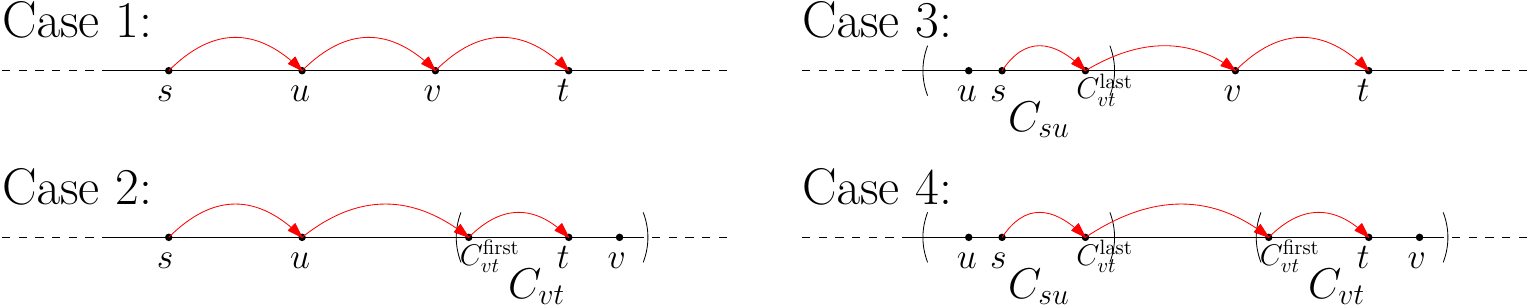}\caption{\emph{\small Illustration of the 4 cases in the proof of \Cref{lem:alphaAppliedRecursively}. In each case the vertices are ordered left to right according to $<_D$. }
			}\label{fig:cases}
		\end{center}
	\end{figure}
	Combining all four cases, we obtain
	\begin{align*}		
		\alpha(s,t)~=~d_{D_{1}}(s,t)& \quad\leq\quad\left(d_{D_{1}}(s,u)+d_{G}(u,v)+d_{D_{1}}(v,t)\right)\boldsymbol{\mathds{1}}[s\le_{D}u]\boldsymbol{\mathds{1}}[v\le_{D}t] \\
		& \phantom{\quad\leq\quad}+\left(d_{D_{1}}(s,u)+d_{G}(u,v)+3\Delta_{vt}\right)\boldsymbol{\mathds{1}}[s\le_{D}u]\boldsymbol{\mathds{1}}[v>_{D}t] \\
		& \phantom{\quad\leq\quad}+\left(3\Delta_{su}+d_{G}(u,v)+d_{D_{1}}(v,t)\right)\boldsymbol{\mathds{1}}[s>_{D}u]\boldsymbol{\mathds{1}}[v\le_{D}t] \\
		& \phantom{\quad\leq\quad}+\left(3\Delta_{su}+d_{G}(u,v)+3\Delta_{vt}\right)\boldsymbol{\mathds{1}}[s>_{D}u]\boldsymbol{\mathds{1}}[v>_{D}t] \\
		& \quad=\quad d_{D_{1}}(s,u)\cdot\boldsymbol{\mathds{1}}[s\le_{D}u]+3\Delta_{su}\cdot\boldsymbol{\mathds{1}}[s>_{D}u]~~+~~d_{G}(u,v)\\&\phantom{\quad\leq\quad}+d_{D_{1}}(v,t)\cdot\boldsymbol{\mathds{1}}[v\le_{D}t]+3\Delta_{vt}\cdot\boldsymbol{\mathds{1}}[v>_{D}t] \\
		& \quad=\quad\alpha(s,u)+d_{G}(u,v)+\alpha(v,t)~.
	\end{align*}
	
\end{proof}

In fact, the triangle inequality of \Cref{lem:alphaAppliedRecursively} can also be applied w.r.t. the temporary active sets created during \Cref{alg:DigraphPartition}.

\begin{corollary}\label{cor:alphaAppliedRecursivelyMiddle}
	Suppose that during the execution of \Cref{alg:DigraphPartition} all the vertices in the shortest $s$–$t$ path $\pi$ are active, i.e., $\pi\subseteq A$.
	Assume a suffix $\pi[v,t]$ of $\pi$ joins the carved ball $B_{\ell,j}$ while the prefix $\pi[s,u]$ remains active, where $(u,v)\in\pi$ and $(u,v)\notin S$. Then 
	$$\alpha(s,t)\le\alpha(s,u)+d_G(u,v)+\alpha(v,t)~.$$
	Similarly, if the prefix $\pi[s,u]$ joins $B_{\ell,j}$, the suffix $\pi[v,t]\subseteq A$ remains active, and $(u,v)\in\pi$ with $(u,v)\notin S$, then $\alpha(s,t)\le\alpha(s,u)+d_G(u,v)+\alpha(v,t)$.
\end{corollary}
\begin{proof}
	We will prove the first proposition of the Corollary, as the second one is symmetric.
	After carving the ball $B_{\ell,j}$, \Cref{alg:DigraphPartition} continues to carve balls and eventually creates the cut set $S$.
	Note that after we carved $B_{\ell,j}$, its vertices cease to be active. In particular, all the edges in the suffix $(u,v)\circ\pi[v,t]$ do not belong to $S$, and thus
	$u\rightsquigarrow_{G\setminus S}v\rightsquigarrow_{G\setminus S}t$.
	As every SCC of $G\setminus S$ is either disjoint from $B_{\ell,j}$ (\Cref{clm:PartitionIteration}) or contained in $B_{\ell,j}$, it follows that $u<_Dv$ and $u<_Dt$.
	
	After carving $B_{\ell,j}$, the algorithm continues carving balls. $s$ and $u$ can end up in the same SCC or not. 
	We continue with a case analysis:
	\begin{itemize}
		\item $s$ and $u$ belong to the same SCC.  
		Note that since $u<_Dt$, by the continuity property of \Cref{def:Laminar}, $s<_Dt$, and thus $\alpha(s,t)=d_{D_1}(s,t)$. We divide into cases:
		\begin{itemize}
			\item $v$ and $t$ belong to the same SCC.  
			This is exactly the case described in \Cref{lem:alphaAppliedRecursively}. It follows that $\alpha(s,t)\le\alpha(s,u)+d_G(u,v)+\alpha(v,t)$.
			
			\item $v$ and $t$ belong to different SCCs. Here $v<_Dt$ (as $\pi[v,t]\subseteq G\setminus S$), and thus $\alpha(v,t)=d_{D_1}(v,t)$.  
			We apply \Cref{lem:alphaAppliedRecursively} on the four vertices $(s,u,v,v)$ (that is, $v$ also takes the role of $t$) to obtain $d_{D_{1}}(s,v)=\alpha(s,v)\le\alpha(s,u)+d_{G}(u,v)+\alpha(v,v)=\alpha(s,u)+d_{G}(u,v)$. It follows:
			\[
			\alpha(s,t)=d_{D_{1}}(s,t)\le d_{D_{1}}(s,v)+d_{D_{1}}(v,t)\le\alpha(s,u)+d_{G}(u,v)+\alpha(v,t)~.
			\]
		\end{itemize}
		
		\item $s$ and $u$ belong to different SCCs and $s\le_Du$.  
		By transitivity, it also holds that $s<_Dv$ and $s<_Dt$. Thus $\alpha(s,t)=d_{D_1}(s,t)$.  
		We divide into cases:
		\begin{itemize}
			\item $v$ and $t$ belong to different SCCs. Here $v<_Dt$ (as $\pi[v,t]\subseteq G\setminus S$), and thus $\alpha(v,t)=d_{D_1}(v,t)$. Using the triangle inequality, we have 
			\[
			\alpha(s,t)=d_{D_{1}}(s,t)\le d_{D_{1}}(s,u)+d_{G}(u,v)+d_{D_{1}}(v,t)=\alpha(s,u)+d_{G}(u,v)+\alpha(v,t)~.
			\]
			
			\item $v$ and $t$ belong to the same SCC.  
			We apply \Cref{lem:alphaAppliedRecursively} on the four vertices $(u,u,v,t)$ (that is, $u$ also takes the role of $s$) to obtain $d_{D_{1}}(u,t)=\alpha(u,t)\le\alpha(u,u)+d_{G}(u,v)+\alpha(v,t)= d_{G}(u,v)+\alpha(v,t)$. It follows:
			\[
			\alpha(s,t)=d_{D_{1}}(s,t)\le d_{D_{1}}(s,u)+d_{D_{1}}(u,t)\le\alpha(s,u)+d_{G}(u,v)+\alpha(v,t)~.
			\]
		\end{itemize}
		
		\item $s$ and $u$ belong to different SCCs and $u\le_Ds$.  
		Here $\alpha(s,u)=3\Delta_{su}$.  
		Let $\Delta$ be the diameter of the graph $G$ we are currently using in \Cref{alg:DigraphPartition}.
		Since $s,t$, and $u$ all belong to different SCCs, while still belonging to the graph $G$ representing some cluster of $\cP$, we have $\Delta_{st}=\Delta_{su}=\Delta$.
		We divide into cases:
		\begin{itemize}
			\item $s<_Dt$.  
			By \Cref{clm:DistanceUsing2HopSpanner}, $\alpha(s,t)=d_{D_1}(s,t)\le 2\Delta_{st}\le 3\Delta_{su}=\alpha(s,u)\le \alpha(s,u)+d_{G}(u,v)+\alpha(v,t)$.
			
			\item $s>_Dt$. Here
			$\alpha(s,t)=3\Delta_{st}=3\Delta_{su}=\alpha(s,u)\le \alpha(s,u)+d_{G}(u,v)+\alpha(v,t)$.
		\end{itemize}
	\end{itemize}
\end{proof}

\subsection{Proof of \Cref{thm:mainExpectedDistortionAspectRatio}}\label{subsec:ProofOfMainThm}
In \Cref{subsec:ProofOfMainThm,subsec:InductiveProof} we will bound $\E[\alpha(s,t)]$, which will conclude the proof of \Cref{thm:mainExpectedDistortionAspectRatio}. 
In essence, the proof is done using a potential function argument. Fix $\brick{c}=384$\submit{\todo{$c$ def}}.
We will bound the expected value of $\alpha$ using the following potential function $f:\N\times\R_{\ge0}\rightarrow\R_{\ge0}$:
\[
\brick{f(m,\Delta)}:=c\cdot L\cdot\ln m\cdot e^{\delta\cdot L\cdot\log(2m\Delta)}~.
\]
In addition, if either $m=0$ or $\Delta=0$, set $f(m,\Delta)=1$.
Note that as we assume $\Delta=\poly(n)$,  $f(m,\Delta)=O(\log n\cdot \log \log n)$.
The following is our inductive hypothesis.
\begin{inductiveHypothesis}\label{ih:main}
	Consider a digraph $G=(V,E,w)$, and a cluster $U\subseteq V$ such that (1) $G[U]$ is strongly connected, (2) $\diam(G,U)=\Delta$, and (3) $|E(G[U])|=m$.
	Let $(<_D,\cP)$ be the laminar topological order returned by \Cref{alg:LaminarPartition}, and let $D_1,D_2$ be the corresponding DAGs. Then for every $s,t\in U$, 
	$$\E[\alpha(s,t)]\le f(m,\Delta)\cdot d_{G[U]}(s,t)~.$$
\end{inductiveHypothesis}

We will next prove our main \Cref{thm:mainExpectedDistortionAspectRatio} using \Cref{ih:main}.
Note that the main difference between the two is that \Cref{thm:mainExpectedDistortionAspectRatio} also holds for digraphs that are not strongly connected.

Fix $s,t\in V$. If $s\not\rightsquigarrow_Gt$, then $d_G(s,t)=\infty$, and the theorem holds trivially.
We will thus assume that $s\rightsquigarrow_Gt$.
By \Cref{obs:alphaDominating}, in order to prove \Cref{thm:mainExpectedDistortionAspectRatio} it is enough to bound $\E[\alpha(s,t)]$.

Let $\pi$ be the shortest $s$–$t$ path in $G$.
Suppose first that all the vertices in $\pi$ belong to the same SCC $C$ in $G$. 
As the algorithm is executed independently on each SCC, by \Cref{ih:main} we have
\[
\E\left[\alpha(s,t)\right]\le f(m,\Delta)\cdot d_{G[C]}(s,t)
=f(m,\Delta)\cdot d_{G}(s,t)=O(\log n\cdot\log\log n)\cdot d_{G}(s,t)~.
\]
Next, we assume that the vertices of $\pi$ belong to several SCCs. The following claim from \cite{AHW25} will be useful. 
See the illustration in \Cref{fig:Generalcase}.
We provide a proof sketch.
\begin{claim}[\cite{AHW25}]\label{clm:path_struct}
	Let $\pi=(z_1,\dots,z_q)$ be an $s$–$t$ path in $G$ for some $s, t \in V$. There exists a sequence of SCCs
	$C_1, \dots, C_k$ of $G$ with the following properties:
	\begin{enumerate}
		\item Path $\pi$ contains exactly one edge $e_j = (u_j, v_{j+1})$ of the form $e_j \in C_{j} \times C_{j+1}$ for all $j \in [1, k-1]$. Moreover,
		\begin{itemize}
			\item  $\pi[s, u_1] \subseteq C_1$,
			\item $\pi[v_j, u_{j}] \subseteq C_{j}$ for all $j \in [2, k-1]$, and 
			\item $\pi[v_{k}, t] \subseteq C_{k}$. 
		\end{itemize}
		\item  $C_{j_1} \neq C_{j_2}$ for all $j_1, j_2 \in [1, k]$ such that $j_1 \neq j_2$.
	\end{enumerate}
\end{claim}
\begin{proof}[Proof sketch]
	The proof is by induction on $|\pi|=q$.
	Let $C_1$ be the SCC of $z_1$. 
	If $\pi\subseteq C_1$, then the claim trivially holds, and there is nothing to prove.
	Otherwise, let $v_2$ be the first vertex in $\pi$ not in $C_1$.
	Let $u_1$ be the vertex preceding $v_2$ in $\pi$.
	Note that $\pi[v_2,t]$ is disjoint from $C_1$. Indeed, if there exists some vertex $z\in \pi[v_2,t]\cap C_1$, then $\pi[s,v_2]$ is a path from $C_1$ to $v_2$, and $\pi[v_2,z]$ is a path from $v_2$ to $C_1$, implying that $v_2$ belongs to the SCC $C_1$, a contradiction.
	The proof is concluded by applying the induction hypothesis on $\pi[v_2,t]$ to obtain a sequence $C_2,\dots,C_k$ of SCCs with the required properties. Now $C_1,\dots,C_k$ is the sequence we need. 
\end{proof}

\begin{figure}[t]
	\begin{center}
		\includegraphics[width=0.9\textwidth]{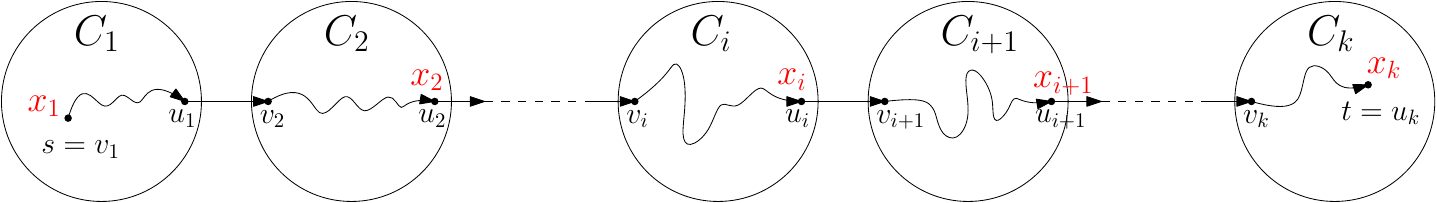}\caption{\emph{\small The shortest $s$–$t$ path is broken into the SCCs $C_1,\dots,C_k$. The intersection of $\pi$ with each SCC $C_i$ is the contiguous sub-path $\pi[u_i,v_i]$.}
		}\label{fig:Generalcase}
	\end{center}
\end{figure}

Using \Cref{clm:path_struct}, let $C_1,\dots,C_k$ be the SCCs into which the shortest $s$–$t$ path $\pi$ is divided, with the respective vertices $v_1=s,u_1,v_2,u_2,\dots,v_k,u_k=t$.
See the illustration in \Cref{fig:Generalcase}.
Set $x_1=s$, $x_2=u_2$, $x_3=u_3$, $\dots$, $x_{k-1}=u_{k-1}$, $x_k=t$.
Note that since $C_1,\dots,C_k$ are SCCs, and the edges $(u_1,v_2),(u_2,v_3),\dots,(u_{k-1},v_k)$ all belong to $D$ (as the algorithm only adds to $\widetilde{S}$ edges inside SCCs), it necessarily holds that $x_1<_D x_2<_D\dots<_Dx_k$.
We apply \Cref{lem:alphaAppliedRecursively} on the tuples $\left\{(x_i,u_i,v_{i+1},x_{i+1})\right\}_{i=1}^{k-1}$ to obtain:
\begin{equation}
	\alpha(s,t)=d_{D_{1}}(s,t)=d_{D_{1}}(x_{1},x_{k})=\sum_{i=1}^{k-1}d_{D_{1}}(x_{i},x_{i+1})\le\sum_{i=1}^{k-1}d_{G}(u_{i},v_{i+1})+\sum_{i=1}^{k}\alpha(u_{i},v_{i})~.\label{eq:alphaSTforPath}
\end{equation}
Using \Cref{ih:main}, we obtain 
\begin{align}
	\E[\alpha(s,t)] & \le\sum_{i=1}^{k-1}d_{G}(u_{i},v_{i+1})+\sum_{i=1}^{k}\E[\alpha(u_{i},v_{i})]\nonumber \\
	& \le\sum_{i=1}^{k-1}d_{G}(u_{i},v_{i+1})+ O(\log n\cdot\log\log n)\cdot\sum_{i=1}^{k}d_{G}(u_{i},v_{i})\nonumber \\
	& =O(\log n\cdot\log\log n)\cdot d_{G}(s,t)~.\label{eq:PathBreakThm1}
\end{align}

\subsection{Proof of \Cref{ih:main}}\label{subsec:InductiveProof}
The algorithm \texttt{Laminar-Topological-Order} is executed on $(G,U)$, where $G=(V,E,w)$ is a digraph and $U\subseteq V$ is a strongly connected cluster. 
The algorithm computes the diameter $\Delta=\diam(G,U)$ and the number of edges $m=|E(G[U])|$. It then makes a call to \texttt{Digraph-Partition}$(G[U],\Delta,m)$.
As a result, we obtain a subset of edges $S\subset E(G[U])$. The algorithm is then run recursively on all the SCCs of $G[U]\setminus S$. 
The algorithm \texttt{Digraph-Partition} has $L$ phases, where in each phase it carves balls using different parameters. After the $L$ phases end, the set of vertices that remain active is denoted by $\cR$.
The algorithm adds to $S$ some additional “long” edges. It is then guaranteed that every SCC in $G[\cR]\setminus S$ has diameter at most $\frac\Delta2$.

\begin{algorithm}[t]
	\caption{$S=\texttt{Digraph-Partition-recursive}(G[A],\Delta,m,S,\ell,j)$}\label{alg:DigraphPartitionAlt}
	\DontPrintSemicolon
	\If{There is a vertex $x_j\in A$ and $*\in\{+,-\}$ such that
		\hspace{200pt}\phantom{.}\qquad $\frac{m}{\mu_{\ell-1}}\le|B_{G[A]}^{*}(x_{j},r_{\ell-1})|\le|B_{G[A]}^{*}(x_{j},r_{\ell})|\leq\frac{m}{\mu_{\ell}}$\qquad\phantom{.}\label{line:ifConditionAlt}}{
		Sample $R \sim \Texp_{[r_{\ell-1},r_\ell]}(\lambda_\ell)$\;
		Set $B_{\ell,j}\leftarrow B_{G[A]}^{*}(x_j, R)$\;
		\Return{$\texttt{Digraph-Partition-recursive}(G[A\setminus B_{\ell,j}],\Delta,m, S \cup \delta^*(B_{\ell,j}),\ell,j+1)$}
	}
	\ElseIf{$\ell\ge2$}{
		\Return{$\texttt{Digraph-Partition-recursive}(G[A],\Delta,m,S,\ell-1,1)$}
	}
	\Else{
		Add to $S$ all the edges in $G[A]$ of weight greater than $r_0$\;
		\Return{$S$}	
	}
\end{algorithm}  

Our proof will be inductive; furthermore, the inductive argument will be applied internally in \texttt{Digraph-Partition}, per iteration of the algorithm (inside the while loop, \cref{line:WhileCurve}). For this, we will need a slightly stronger inductive hypothesis.
For the sake of analysis, we will formulate the algorithm \texttt{Digraph-Partition} differently. Specifically, we will replace the for and while loops by recursive calls.
See \Cref{alg:DigraphPartitionAlt} for the exact formulation.
Here, in addition to the parameters $G,\Delta,m$, the algorithm also receives: $A$, the set of active vertices; $S$, the set of edges deleted so far; $\ell$, the current phase of the algorithm; and $j$, the current iteration of the algorithm in the $\ell$th phase.%
\footnote{Actually, the algorithm can ignore $j$. However, it is useful to keep track of $j$ during the analysis.}
The call from \texttt{Laminar-Topological-Order} will be $\texttt{Digraph-Partition-recursive}(G[U],\Delta,m,\emptyset,L,1)$. 
\Cref{alg:DigraphPartitionAlt} checks whether there is a vertex $x_j$ fulfilling the condition $\frac{m}{\mu_{\ell-1}}\le|B_{G[A]}^{*}(x_{j},r_{\ell-1})|\le|B_{G[A]}^{*}(x_{j},r_{\ell})|\leq\frac{m}{\mu_{\ell}}$. If so, it carves the ball $B_{\ell,j}$ in exactly the same manner as in \Cref{alg:DigraphPartition}, updates the active set $A$ and the cut set $S$ accordingly, and makes a recursive call to itself, while updating the iteration to $j+1$.
Otherwise (no such vertex $x_j$), if $\ell>1$ (that is, we are not in the last phase), the algorithm makes a recursive call to itself while moving to phase $\ell-1$ and resets the iteration counter $j$ to $1$.
Otherwise ($\ell=1$), the algorithm adds all the edges in $G[A]$ of weight greater than $r_0$ and returns $S$.

Note that the output of \Cref{alg:DigraphPartition} and \Cref{alg:DigraphPartitionAlt} is exactly the same. We will thus prove \Cref{ih:main} as if we were using \Cref{alg:DigraphPartitionAlt}.
We will use a more delicate inductive hypothesis that takes into account the exact phase and iteration we are currently in. For this, we will update $f$ to depend also on the phase and iteration counters:
\[
\brick{f(m,\Delta,\ell,j)}=c\cdot L\cdot\ln m\cdot e^{\delta\cdot\left(L\cdot\log(m\Delta)+(\ell-1)+\frac{\mu_{\ell}+1-j}{\mu_{\ell}}\right)}~.
\]
In addition, if either $m=0$ or $\Delta=0$, set $f(m,\Delta,\ell,j)=1$.
We will prove the following inductive hypothesis: 
\begin{inductiveHypothesis}\label{ih:mainAlt}
	Suppose that during the execution of our algorithm there is a call $\texttt{Digraph-Partition-recursive}(G[A],\Delta,m,S,\ell,j)$.
	Then for every $s,t\in A$
	$$\E[\alpha(s,t)]\le f(m,\Delta,\ell,j)\cdot d_{G[A]}(s,t)~.$$
\end{inductiveHypothesis}

\sloppy Note that \Cref{ih:mainAlt} implies \Cref{ih:main}.
Indeed, \texttt{Laminar-Topological-Order} makes the call $\texttt{Digraph-Partition-recursive}(G[U],\Delta,m,S,L,1)$.
Thus for $s,t\in U$, by \Cref{ih:mainAlt} it follows that \submit{\todo{$f(m,\Delta,L,1)\ge f(m,\Delta)$}}
\begin{align*}
	\E[\alpha(s,t)] & \le f(m,\Delta,L,1)\cdot d_{G[A]}(s,t)=c\cdot L\cdot\log m\cdot e^{\delta\cdot\left(L\cdot\log(m\Delta)+(L-1)+\frac{\mu_{L}+1-1}{\mu_{L}}\right)}\cdot d_{G[A]}(s,t)\\
	& =c\cdot L\cdot\log m\cdot e^{\delta\cdot\left(L\cdot\log(m\Delta)+L\right)}\cdot d_{G[A]}(s,t)=f(m,\Delta)\cdot d_{G[A]}(s,t)~,
\end{align*}
implying \Cref{ih:main}. 
The rest of this subsection is devoted to proving \Cref{ih:mainAlt}. 
The proof will be by induction on all four parameters $m,\Delta,\ell,j$ in that order. 
That is, when proving for a call with parameters $(m,\Delta,\ell,j)$, we will assume \Cref{ih:mainAlt} holds for calls on $(m',\Delta',\ell',j')$ such that one of the following holds: (1) $m'<m$, (2) $m'=m$ and $\Delta'<\Delta$, (3) $m'=m$, $\Delta'=\Delta$, and $\ell'<\ell$, or (4) $m'=m$, $\Delta'=\Delta$, $\ell'=\ell$, and $j'<j$.

If $m=0$, then the cluster is a singleton, and clearly the distortion is $f(m,\Delta,\ell,j)=1$.
Otherwise, $m\ge1$, and $\Delta\ge1$ (as the minimum edge weight is $1$).
Fix $\brick{s,t}\in A$, and let $\brick{\pi}=(s=v_{0},v_{1},\dots,v_{q}=t)$ be the shortest $s$–$t$ path in $G[A]$. The base of the induction is when $d_{G[A]}(s,t)\ge \min\{r_0,\frac{1}{\lambda_\ell}\}$. In this case the induction holds trivially.
\begin{observation}\label{obs:stLong}
	If $d_{G[A]}(s,t)\ge \min\{r_0,\frac{1}{\lambda_\ell}\}$ then $\alpha(s,t)\le f(m,\Delta,\ell,j)\cdot d_{G[A]}(s,t)$.
	\submit{\todo{Constant defining $c$. Note that we also need to change from $\log$ to $\ln$}}
\end{observation}
\begin{proof}
	Note first that $r_{0}=\frac{\Delta}{2^{L+4}}=\frac{\Delta}{2^{\lceil\log\log(m+1)\rceil+4}}\ge\frac{\Delta}{2^{\log\log m+6}}=\frac{\Delta}{64\cdot\log m}$.
	Next, by \Cref{obs:lambdaUB} we have $\frac{1}{\lambda_{\ell}}\ge\frac{\Delta}{64\cdot L\cdot\log\mu_{\ell}}\ge\frac{\Delta}{128\cdot L\cdot\log m}$, where the last inequality holds as $\log\mu_{\ell}=2^{L-\ell}\le2^{\lceil\log\log(m+1)\rceil-1}\le\log(m+1)\le2\cdot\log m$.
	We conclude that if $d_{G[A]}(s,t)\ge\min\{r_{0},\frac{1}{\lambda_{\ell}}\}\ge\frac{\Delta}{128\cdot L\cdot\log m}$, then by \Cref{obs:alphaDominating}
	\[
	\alpha(s,t)\le3\Delta_{st}\le3\Delta\le384\cdot L\cdot\log m\cdot d_{G[A]}(s,t)\le f(m,\Delta,\ell,j)\cdot d_{G[A]}(s,t)~.\qedhere
	\]
\end{proof}

We will thus assume that $d_{G[A]}(s,t)< \min\{r_0,\frac{1}{\lambda_\ell}\}$.
Next, suppose that when executing $\texttt{Digraph-Partition-recursive}(G[A],\Delta,m,S,\ell,j)$, there is no center $x_j$ fulfilling the condition in \cref{line:ifConditionAlt} (of \Cref{alg:DigraphPartitionAlt}). There are two possible cases:
\begin{itemize}
	\item First suppose that $\ell>1$. Here the algorithm makes the call $\texttt{Digraph-Partition-recursive}(G[A],\Delta,m,S,\ell-1,1)$. 
	Note that by \Cref{obs:NumberOfBalls}, during the $\ell$th phase, the algorithm carved at most $\mu_\ell$ balls. In particular, $j\le\mu_\ell+1$ (as $j$ is the counter for the current iteration).
	It follows that 
	\submit{\todo{$\frac{f(m,\Delta,\ell,j)}{f(m,\Delta,\ell-1,1)}\ge1$}}
	\begin{align*}
		\frac{1}{\delta}\cdot\ln\left(\frac{f(m,\Delta,\ell,j)}{f(m,\Delta,\ell-1,1)}\right) & =\left((\ell-1)+\frac{\mu_{\ell}+1-j}{\mu_{\ell}}\right)-\left((\ell-2)+\frac{\mu_{\ell-1}+1-1}{\mu_{\ell-1}}\right)\\
		& =\frac{\mu_{\ell}+1-j}{\mu_{\ell}}\ge0~.
	\end{align*}
	Using \cref{ih:mainAlt} we conclude 
	\[
	\E[\alpha(s,t)]\le f(m,\Delta,\ell-1,1)\cdot d_{G[A]}(s,t)\le f(m,\Delta,\ell,j)\cdot d_{G[A]}(s,t)~.
	\]
	
	\item Otherwise, $\ell=1$. Again, by \Cref{obs:NumberOfBalls}, $j\le\mu_1+1$.
	At this stage, the algorithm adds to $S$ all edges of weight at least $r_0$. However, as we are assuming $d_{G[A]}(s,t)<r_0$, no edge on the shortest $s$–$t$ path $\pi$ is added to $S$.
	Next, the algorithm returns $S$ to \texttt{Laminar-Partition}. Denote $\cR= A$. Note that $\pi\subseteq G[\cR]\setminus S$.
	The algorithm \texttt{Laminar-Partition} has a graph $G$ and removes all $S$ edges. All vertices of $\pi$ belong to SCCs that are contained in $\cR$.
	
	Consider a pair of vertices $u,v\in \pi$ that belong to the same SCC $C$ of $G\setminus S$.
	Let $\Delta_C=\diam(G,C)$, and let $m_C$ be the respective number of edges in the cluster $C$.
	By \Cref{clm:PartitionIterationRDelta}, $\Delta_C\le\frac\Delta2$, while clearly $m_C\le m$. Denote $L_C=\lceil\log\log (m_C+1)\rceil$.
	The algorithm \texttt{Laminar-Partition} makes a call to $\texttt{Digraph-Partition-recursive}$ for every SCC, and in particular the call $\texttt{Digraph-Partition-recursive}(G[C]\setminus S,\Delta_C,m_C,\emptyset,L_{C},1)$.
	It thus follows by \Cref{ih:mainAlt} that \submit{\todo{$\frac{f(m,\Delta,1,\mu_{1}+1)}{f(m,\frac{\Delta}{2},L,1)}\ge1$}}
	\begin{equation}
		\E[\alpha(u,v)]\le f(m_{C},\Delta_{C},L_{C},1)\cdot d_{G[C]\setminus S}(s,t)\le f(m,\Delta,1,j)\cdot d_{G[A]}(u,v)~,\label{eq:ell1case}
	\end{equation}
	where the last inequality holds as $f(m_{C},\Delta_{C},L_{C},1)\le f(m,\Delta_{C},L_{C},1)\le f(m,\frac{\Delta}{2},L,1)\overset{(*)}{=}f(m,\Delta,1,\mu_{1}+1)\le f(m,\Delta,1,j)$.
	Here the equality $^{(*)}$ follows since 
	\[
	L\cdot\log(m\Delta)+(1-1)+\frac{\mu_{1}+1-(\mu_{1}+1)}{\mu_{1}}=L\cdot\log\!\left(m\frac{\Delta}{2}\right)+(L-1)+\frac{\mu_{L}+1-1}{\mu_{L}}~.
	\]
	As $\pi\subseteq G\setminus S$, and \cref{eq:ell1case} holds for every two vertices $u,v$ in $\pi$ that belong to the same SCC of $G\setminus S$, 
	we apply the exact same arguments as during \Cref{subsec:ProofOfMainThm}, leading to \cref{eq:alphaSTforPath,eq:PathBreakThm1}, where we replace $O(\log n\cdot\log\log n)$ by \cref{eq:ell1case}. Overall:
	$$\E[\alpha(s,t)]\le  f(m,\Delta,\ell,j)\cdot d_{G[A]}(s,t)~.$$
\end{itemize}

Finally, we arrive at the case where the algorithm found a center $x_j$ fulfilling the condition in \cref{line:ifConditionAlt} (of \Cref{alg:DigraphPartitionAlt}). 
We will assume w.l.o.g.\ that $*=+$. That is, $\frac{m}{\mu_{\ell-1}}\le|B_{G[A]}^{+}(x_{j},r_{\ell-1})|\le|B_{G[A]}^{+}(x_{j},r_{\ell})|\leq\frac{m}{\mu_{\ell}}$. 
The case where $*=-$ is symmetric.
The algorithm then samples a radius $R\sim\Texp_{[r_{\ell-1},r_\ell]}(\lambda_\ell)$ and carves the ball $B_{\ell,j}=B^+_{G[A]}(x_j,R)$.
For simplicity of notation, in the rest of the proof we will write $\brick{G}$ instead of $G[A]$. 

For every value $R\ge r_{\ell-1}$, let $\brick{\pi(R)}=\left\{v\in \pi\mid d_G(x_j,v)\le R\right\}$ be the subset of vertices of $\pi$ that join $B_{\ell,j}$ if we sample radius $R$.
Let $\brick{z_{1}}\in\pi$ be the closest vertex to the center $x_j$ in $G$ (i.e., the vertex $v$ minimizing $d_{G}(x_{j},v)$).
Let $\brick{\rho_1}=\min\{d_{G}(x_{j},v),r_{\ell-1}\}$ be the minimum value such that if $R\ge\rho_1$, then $z_1\in B_{\ell,j}$.
In particular, $z_1\in\pi(\rho_1)$, while for every $r_{\ell-1}\le R<\rho_1$, $\pi(R)=\emptyset$.
Let $\brick{\eta_1}\ge\rho_1$ be the minimum value such that $\pi(\eta_1)$ is a suffix of $\pi$. 
Note that $\eta_1$ indeed exists and is finite (this is because every vertex that is reachable from $x_j$ in $G$ has a finite distance from $x_j$, and for every reachable vertex $y\in \pi$, all the vertices in $\pi[y,t]$ are reachable).
We allow $\eta_1$ to have a value above $r_\ell$.
Let $\brick{s_1}$ be the leftmost vertex (w.r.t. $\pi$) in $\pi(\eta_1)$. That is, $\pi(\eta_1)=\pi[s_1,t]$. Note that it might be the case that $s_1\ne z_1$. We will also denote $\brick{t_1}=t$ and $\brick{Q_1}= \pi[s_1,t_1]$. See \Cref{fig:PathBrake} for illustration.

Suppose that there are vertices in $\pi\setminus \pi(\eta_1)$ that are reachable from $x_j$.
Denote by $\brick{t_2}$ the predecessor of $s_1$ in $\pi$.
Let $\brick{z_2}$ be the closest vertex to $x_j$ in $\pi[s,t_2]= \pi\setminus \pi(\eta_1)$, and denote by $\brick{\rho_2}=d_G(x_j,z_2)$. Note that $\rho_2$ is the minimum value such that if $R\ge\rho_2$, then $z_2\in B_{\ell,j}$.
That is, for values $R\in[\eta_1,\rho_2)$, it holds that $\pi(R)=Q_1$, while for $R\ge\rho_2$, we have $z_2\in\pi(R)$.
Let $\brick{\eta_2}\ge\rho_2$ be the minimum value such that $\pi(\eta_2)$ is a suffix of $\pi$.
Let $\brick{s_2}$ be the leftmost vertex (w.r.t. $\pi$) in $\pi(\eta_2)$. That is, $\pi(\eta_2)=\pi[s_2,t]$. We will also denote $\brick{Q_2}= \pi[s_2,t_2]$. 

\begin{figure}[t]
	\begin{center}
		\includegraphics[width=0.956\textwidth]{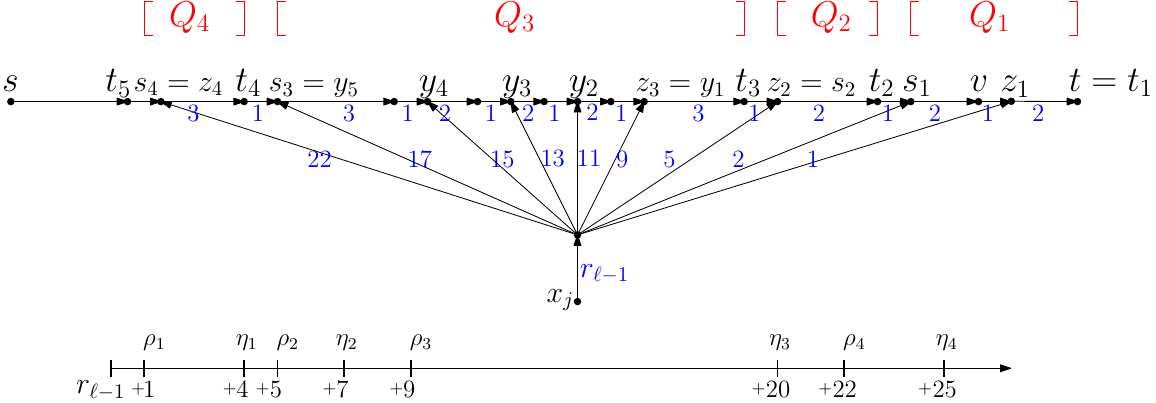}
		\caption{\emph{\small The shortest $s$–$t$ path $\pi$ is illustrated at the top. 
				The ball $B^+_G(x_j,R)=B_{\ell,j}$ is carved by sampling a radius $R\in[r_{\ell-1},r_\ell]$.
				The edge weights are noted in blue.
				The closest vertex on $\pi$ to $x_j$ is $z_1$ at distance $\rho_1=d_G(x_j,z_1)=r_{\ell-1}+1$. 
				Note that for $R=r_{\ell-1}+3$ the subset of carved vertices is $\pi(r_{\ell-1}+3)=\{s_1,z_1,t\}$, which is not a suffix of $\pi$.
				The first time we get a suffix is for the value $\eta_1=r_{\ell-1}+4$, where $Q_1=\pi(\eta_1)=\{s_1,v,z_1,t\}$.
				We continue in this manner to define the intervals $Q_2,Q_3,Q_4$.
				The prefix $\{s,t_5\}$ is not reachable from $x_j$.
				At the bottom we illustrate an axis containing the values of $\rho_i,\eta_i$ for $i\in\{1,\dots,4\}$. Here $+x$ denotes $R_{\ell-1}+x$. }
		}\label{fig:PathBrake}
	\end{center}
\end{figure}

We continue in this manner, defining $(z_3,\rho_3,\eta_3,s_3,t_3,Q_3)$, $(z_4,\rho_4,\eta_4,s_4,t_4,Q_4)$, $\dots$, and so on.
In general, after creating $\brick{(z_i,\rho_i,\eta_i,s_i,t_i,Q_i)}$, suppose that there are vertices in $\pi\setminus \pi(\eta_i)$ that are reachable from $x_j$.
Denote by $\brick{t_{i+1}}$ the predecessor of $s_i$ in $\pi$.
Let $\brick{z_{i+1}}$ be the closest vertex to $x_j$ in $\pi[s,t_{i+1}]= \pi\setminus \pi(\eta_i)$, and denote by $\brick{\rho_{i+1}}=d_G(x_j,z_{i+1})$ the minimum value such that if $R\ge\rho_{i+1}$, then $z_{i+1}\in B_{\ell,j}$.
That is, for values $R\in[\eta_i,\rho_{i+1})$, it holds that $\pi(R)=Q_1\cup\dots\cup Q_i=\pi[s_i,t]$, while for $R\ge\rho_{i+1}$, we have $z_{i+1}\in\pi(R)$.
Let $\brick{\eta_{i+1}}\ge\rho_{i+1}$ be the minimum value such that $\pi(\eta_{i+1})$ is a suffix of $\pi$.
Let $\brick{s_{i+1}}$ be the leftmost vertex (w.r.t. $\pi$) in $\pi(\eta_{i+1})$. That is, $\pi(\eta_{i+1})=\pi[s_{i+1},t]$. We will also denote $\brick{Q_{i+1}}= \pi[s_{i+1},t_{i+1}]$.

In this manner, we create $\left\{(z_i,\rho_i,\eta_i,s_i,t_i,Q_i)\right\}_{i=1}^{\brick{k}}$.
At the end there are two options:
\begin{itemize}
	\item $s_k=s$. In particular, $\pi(\eta_k)=\pi$.
	\item There are vertices in $\pi$ not in $\pi(\eta_k)$. These vertices are not reachable from $x_j$ and form a prefix of $\pi$.
	We denote $\brick{t_{k+1}}$ to be the predecessor of $s_k$ in $\pi$. 
\end{itemize}

With respect to the order of $\pi$, it holds that $z_{q}\le_{\pi}z_{q-1}\le_{\pi}\dots\le_{\pi}z_{1}$, where $v\le_\pi u$ denotes that $v$ appears before $u$ on $\pi$.
It holds that $\rho_{1}\le\eta_{1}<\rho_{2}\le\eta_{2}<\dots<\rho_{k}\le\eta_{k}$.
In the following lemma we bound the distance between $\eta_i$ and $\rho_i$. This lemma will be crucial to bound the probability that some edge along $\pi$ is added to the cut set $S$.

\begin{lemma}\label{lem:etaIrhoI}
	For $i\in[1,k]$, $\eta_i-\rho_i\le d_G(s_i,t_i)$.
\end{lemma}
\begin{proof}
	Denote $\pi[s_i,t_i]=(u_1,\dots,u_\beta)$.
	Let $\cA=\{u_q\in \pi[s_i,t_i]\mid \forall q'<q,~d_G(x_j,u_q)<d_G(x_j,u_{q'})\}$ be all vertices $u_q$ of $\pi[s_i,t_i]$ that are closer to $x_j$ than all vertices preceding them in $\pi[s_i,t_i]$ (that is, the prefix $\pi[s_i,u_q]$).
	For illustration, in \Cref{fig:PathBrake}, for $i=3$, $\cA=\{y_1=z_3,y_2,y_3,y_4,y_5=s_3\}$.
	Clearly $z_i\in\cA$ as $z_i$ is the closest vertex in $\pi[s,t_i]$ to $x_j$.
	Further, $s_i\in \cA$, as it is the only vertex in the prefix $\pi[s_i,s_i]$.
	Let $y_1,y_2,\dots,y_l$ be the vertices of $\cA$ sorted from right to left (w.r.t. $\pi$).
	Note that $d_G(x_j,y_1)<d_G(x_j,y_2)<\dots<d_G(x_j,y_l)$.
	In particular, $y_1=z_i$ and $y_l=s_i$. For simplicity of notation, denote $y_0=t_i$. 
	
	\begin{claim}\label{clm:PicksInsideIntervalBehaviour}
		For every $q\in[2,l]$,
		$d_{G}(x_{j},y_{q})\le\max\left\{ d_{G}(x_{j},y_{q'})+d_{G}(y_{q'},y_{q'-1})\right\} _{q'=1}^{q-1}$.
		Further, it also holds that $\eta_{i}\le\max\left\{ d_{G}(x_{j},y_{q})+d_{G}(y_{q},y_{q-1})\right\} _{q=1}^{l}$.
	\end{claim}
	\begin{proof}
		Denote $\tilde{\eta}=\max\left\{ d_{G}(x_{j},y_{q'})+d_{G}(y_{q'},y_{q'-1})\right\} _{q'=1}^{q-1}$, and suppose for the sake of contradiction that $d_{G}(x_{j},y_{q})>\tilde{\eta}$.
		Denote by $\pi_{i}(R)=\pi(R)\cap Q_{i}$ all the vertices in the interval
		$Q_{i}=\pi[s_{i},t_{i}]$ that belong to $\pi(R)$. We argue that
		$\pi_{i}(\tilde{\eta})=\pi[y_{q-1},t_{i}]$. Thus $\pi(\tilde{\eta})$ is a suffix of $\pi$, while $\rho_i\le\tilde{\eta}<\eta_i$, a contradiction to the definition of $\eta_i$ as the minimum value above $\rho_i$ such that $\pi(\eta_i)$ is a suffix.
		We now turn to proving $\pi_{i}(\tilde{\eta})=\pi[y_{q-1},t_{i}]$. Consider first a vertex $u\in\pi[y_{q-1},t_{i}]$. Then
		there is some $q'\in[1,q-1]$ such that $u\in\pi[y_{q'},y_{q'-1}]$.
		It holds that $d_{G}(x_{j},u)\le d_{G}(x_{j},y_{q'})+d_{G}(y_{q'},u)\le d_{G}(x_{j},y_{q'})+d_{G}(y_{q'},y_{q'-1})\le\tilde{\eta}$,
		thus indeed $u\in\pi_{i}(\tilde{\eta})$, and hence $\pi[y_{q-1},t_{i}]\subseteq\pi_{i}(\tilde{\eta})$.
		
		Next we argue that for every vertex $u\in Q_i\setminus \pi[y_{q-1},t_{i}]$, $u\notin \pi_i(\tilde{\eta})$.
		Suppose first that $u\in \pi[s_{i},y_{q}]$. Then, as $y_{q}\in\cA$, we have $d_{G}(x_{j},u)>d_{G}(x_{j},y_{q})>\tilde{\eta}$, implying $u\notin\pi(\tilde{\eta})$.
		Next suppose that $u\in\pi(y_{q},y_{q-1})$ (if any), and let $u'\in\pi(y_{q},y_{q-1})$ be the closest vertex to
		$x_{j}$. As $u'\notin\cA$, $u'$ is the closest vertex to $x_{j}$
		in $\pi(y_{q},y_{q-1})$ but not in $\pi[s_{i},y_{q}]$. It follows
		that $d_{G}(x_{j},u)\ge d_{G}(x_{j},u')\ge d_{G}(x_{j},y_{q})>\tilde{\eta}$,
		implying $u\notin\pi_{i}(\tilde{\eta})$. 
		
		Finally, using the exact same arguments as above, one can show that for ${\eta'=\max\left\{ d_{G}(x_{j},y_{q})+d_{G}(y_{q},y_{q-1})\right\} _{q=1}^{l}}$, it holds that $Q_i\subseteq \pi(\eta')$.
		As $\eta_i$ is the minimum value such that $Q_i\subseteq\pi(\eta_i)$, it follows that $\eta_i\le\eta'$, as required.
	\end{proof}
	
	Next, we argue by induction on $q\in[1,l]$ that $d_{G}(x_{j},y_{q})\le d_{G}(x_{j},y_{1})+d_{G}(y_{q-1},t_{i})$.
	The base case $q=1$ is trivial as $d_{G}(x_{j},y_{1})=d_{G}(x_{j},y_{1})+d_{G}(y_{0},t_{i})$.
	Next, assume that for every $q'<q$, $d_{G}(x_{j},y_{q'})\le d_{G}(x_{j},y_{1})+d_{G}(y_{q'-1},t_{i})$.
	Using \Cref{clm:PicksInsideIntervalBehaviour} and the inductive hypothesis, we conclude 
	\begin{align*}
		d_{G}(x_{j},y_{q}) & \le\max\left\{ d_{G}(x_{j},y_{q'})+d_{G}(y_{q'},y_{q'-1})\right\} _{q'=1}^{q-1}\\
		& \le\max\left\{ d_{G}(x_{j},y_{1})+d_{G}(y_{q'-1},t_{i})+d_{G}(y_{q'},y_{q'-1})\right\} _{q'=1}^{q-1}\\
		& =d_{G}(x_{j},y_{1})+\max\left\{ d_{G}(y_{q'},t_{i})\right\} _{q'=1}^{q-1}\\
		& =d_{G}(x_{j},y_{1})+d_{G}(y_{q-1},t_{i})~.
	\end{align*}
	Finally, using \Cref{clm:PicksInsideIntervalBehaviour} again, we conclude
	\begin{align*}
		\eta_{i} & \le\max\left\{ d_{G}(x_{j},y_{q})+d_{G}(y_{q},y_{q-1})\right\} _{q=1}^{l}\\
		& \le\max\left\{ d_{G}(x_{j},y_{1})+d_{G}(y_{q-1},t_{i})+d_{G}(y_{q},y_{q-1})\right\} _{q=1}^{l}\\
		& =d_{G}(x_{j},y_{1})+\max\left\{ d_{G}(y_{q},t_{i})\right\} _{q=1}^{l}\\
		& =d_{G}(x_{j},y_{1})+d_{G}(y_{l},t_{i})=d_{G}(x_{j},y_{1})+d_{G}(s_{i},t_{i})=\rho_{i}+d_{G}(s_{i},t_{i})~.\qedhere
	\end{align*}
\end{proof}

\paragraph*{Using real exponential distribution.} The radius $R$ is sampled using the truncated exponential distribution $\Texp_{[r_{\ell-1},r_\ell]}(\lambda_\ell)$. This is an exponential distribution with parameter $\lambda_\ell$ conditioned on the sampled value being in the interval $[r_{\ell-1},r_\ell]$.
This distribution is not memoryless, and hence considerably more complicated to analyze compared with the standard exponential distribution.
We will analyze the expected value of $\alpha(s,t)$ in two steps. 
First, we will ignore truncation and analyze the process as if we were to sample $R\ge r_{\ell-1}$ using the standard exponential distribution with parameter $\lambda_\ell$.
Afterward, we will incorporate the truncation into our analysis.

\begin{remark}\label{rem:ExponentialSizeArgument}
	To compute $\mathbb{E}\left[\alpha(s,t)\right]$ we will integrate over $\mathbb{E}\left[\alpha(s,t)\mid R=x\right]$ for all possible values $x$ of $R$ (see \cref{eq:ExpandingToAllValuesOfR}).
	In \Cref{lem:NoCut,lem:InductionOnPath}, we compute $\mathbb{E}\left[\alpha(s,t)\right]$ while using the real exponential distribution.	
	In particular, we will consider values $R> r_\ell$, which cannot occur. As such, for $x>r_\ell$, we can assign $\mathbb{E}\left[\alpha(s,t)\mid R=x\right]$ an arbitrary positive value (since \cref{eq:ExpandingToAllValuesOfR} will hold for any such value).
	Our usage of this fact will be by assuming that \Cref{clm:PartitionIteration} (bounding the number of edges in an SCC contained in $B_{\ell,j}$) holds even when using $R>r_j$. 
	This is acceptable because the computed value $\mathbb{E}\left[\alpha(s,t)\mid R=x\right]$ will still be positive.
\end{remark}

Next, we analyze $\E[\alpha(s,t)]$ for different values of $R$.
Note that by \Cref{obs:alphaDominating}, $\alpha(s,t)$ is bounded by $3\Delta$ (regardless of the sampled radius $R$).
For simplicity of notation, denote $\brick{\eta_0}=r_{\ell-1}$, $\brick{\rho_{k+1}}=\infty$, and $\brick{s_0}=t$.

\begin{lemma}\label{lem:NoCut}
	Suppose that $R\in [\eta_q,\rho_{q+1})$ for some $q\in[0,k]$. Then 
	\begin{equation*}
		\E\left[\alpha(s,t)\mid R\in[\eta_{q},\rho_{q+1})\right]
		\le f(m,\Delta,\ell,j+1)\cdot d_{G}(s,t)-c\cdot L\cdot\log\mu_{\ell}\cdot d_{G}(s_{q},t)~.
	\end{equation*}
\end{lemma}
\begin{proof}
	Consider first the case $q=0$, that is, $R\in [\eta_0,\rho_{1})$. 
	In this case none of the vertices of $\pi$ join $B_{\ell,j}$.%
	\footnote{If $\rho_1=\eta_0$, then $[\eta_0,\rho_{1})=\emptyset$. In this case the lemma holds trivially (as one cannot suppose that $R\in\emptyset$).} In particular, $\pi\subseteq G[A\setminus B_{\ell,j}]$. The algorithm will make the call $\texttt{Digraph-Partition-recursive}(G[A\setminus B_{\ell,j}],\Delta,m,S,\ell,j+1)$. By \Cref{ih:mainAlt} it holds that 
	\[
	\E[\alpha(s,t)\mid R<\rho_{1}]\le f(m,\Delta,\ell,j+1)\cdot d_{G[A\setminus B_{\ell,j}]}(s,t)
	= f(m,\Delta,\ell,j+1)\cdot d_{G}(s,t)\,-\,c\cdot L\cdot\log\mu_{\ell}\cdot d_{G}(s_{0},t)~.
	\]
	where in the last inequality we used $s_0=t$.
	
	Next, consider the case $R\in [\eta_q,\rho_{q+1})$ for some $q\in[1,k]$.
	In this case, all the vertices in the suffix $\pi(\eta_q)=\pi[s_q,t]$ join $B_{\ell,j}$, while the vertices in the prefix $\pi[s,t_{q+1}]$ remain active in $A\setminus B_{\ell,j}$.
	All edges in $\delta^+(B_{\ell,j})$ are added to $S$.
	However, the edge $(t_{q+1},s_q)$ is not added to $S$ and will be in $D$. 	
	\begin{claim}\label{clm:DistNoCutsqTot}
		$\E\left[\alpha(s_{q},t)\mid R\in[\eta_{q},\rho_{q+1})\right]\le \left(f(m,\Delta,\ell,j+1)-c\cdot L\cdot\ln\mu_{\ell}\right)\cdot d_{G}(s_q,t)$.
	\end{claim}
	\begin{proof}
		Recall that by \Cref{clm:PartitionIteration}, every SCC of $G\setminus S$ containing vertices from $B_{\ell,j}$ is fully contained in $B_{\ell,j}$ and has at most $\frac{m}{\mu_\ell}$ edges (see \Cref{rem:ExponentialSizeArgument}).
		Consider a pair of vertices $u,v$ in $\pi[s_q,t]$ that belong to the same SCC $C$ of $G\setminus S$. Denote by $m_C=|E(C)|\le \frac{m}{\mu_\ell}$ its number of edges, by $\Delta_C\le\Delta$ its diameter, and also $L_C=\lceil\log\log (m_C+1)\rceil$.
		The algorithm \texttt{Laminar-Partition} will execute $\texttt{Digraph-Partition-recursive}(G[C],\Delta_C,m_C,S,L_C,1)$. 
		Using \Cref{ih:mainAlt}, it holds that
		\[
		\E[\alpha(u,v)]\le f(m_C,\Delta_C,L_C,1)\cdot d_{G[C]}(u,v)\le f\!\left(\tfrac{m}{\mu_\ell},\Delta,L,1\right)\cdot d_{G[C]}(u,v)~.
		\]
		Next, we apply the exact same arguments as during \Cref{subsec:ProofOfMainThm}, leading to \cref{eq:alphaSTforPath,eq:PathBreakThm1}, where we replace $O(\log n\cdot\log\log n)$ by $f(\tfrac{m}{\mu_\ell},\Delta,L,1)$. 
		It follows that
		\[
		\E\left[\alpha(s_{q},t)\mid R\in[\eta_{q},\rho_{q+1})\right]\le f\!\left(\tfrac{m}{\mu_{\ell}},\Delta,L,1\right)\cdot d_{G}(s_q,t)~.
		\]
		To get the claim, note that
		\[
		\begin{aligned}
			f\!\left(\tfrac{m}{\mu_{\ell}},\Delta,L,1\right) 
			&= c\cdot L\cdot\ln\!\tfrac{m}{\mu_{\ell}}\cdot e^{\delta\cdot\left(L\cdot\log(\tfrac{m}{\mu_{\ell}}\Delta)+(L-1)+\tfrac{\mu_{L}+1-1}{\mu_{L}}\right)} \\
			&\le c\cdot L\cdot\ln m\cdot e^{\delta\cdot\left(L\cdot\log(m\Delta)+L\cdot(1-\log\mu_{\ell})\right)}-c\cdot L\cdot\ln\mu_{\ell}\\
			&\le f(m,\Delta,\ell,j+1)-c\cdot L\cdot\ln\mu_{\ell}~,
		\end{aligned}
		\]
		where the last inequality holds since (1) $j\le\mu_\ell$ (by \Cref{clm:PartitionIteration}, as we found a center $x_j$), and (2) $\mu_\ell\ge\mu_1=2$, hence $\log\mu_\ell\ge1$. 
	\end{proof}
	
	If $q=k$ and $s_q=s_k=s$, then \Cref{clm:DistNoCutsqTot} is exactly what we have to prove, and thus we are done. We will thus assume $s_q\ne s$ (regardless of the value of $q$).
	The algorithm continues to run on the active vertices in the prefix subpath $\pi[s,t_{q+1}]$. 
	By \Cref{ih:mainAlt} it holds that 
	\[
	\E[\alpha(s,t_{q+1})\mid R\in[\eta_{q},\rho_{q+1})]\le f(m,\Delta,\ell,j+1)\cdot d_{G}(s,t_{q+1})~.
	\]
	Using \Cref{cor:alphaAppliedRecursivelyMiddle}, we conclude
	\begin{align*}
		\E\left[\alpha(s,t)\mid R\in[\eta_{q},\rho_{q+1})\right] 
		& \le\E\left[\alpha(s,t_{q+1})\mid R\in\dots\right]+d_{G}(t_{q+1},s_{q})+\E\left[\alpha(s_{q},t)\mid R\in\dots\right]\\
		& \le f(m,\Delta,\ell,j+1)\cdot d_{G}(s,t_{q+1})+d_{G}(t_{q+1},s_{q})\\
		& \qquad+\left(f(m,\Delta,\ell,j+1)-c\cdot L\cdot\ln\mu_{\ell}\right)\cdot d_{G}(s_{q},t)\\
		& \le f(m,\Delta,\ell,j+1)\cdot d_{G}(s,t)-c\cdot L\cdot\ln\mu_{\ell}\cdot d_{G}(s_{q},t)~.\qedhere
	\end{align*}
\end{proof}

Next, we also consider the cases where some edges of $\pi$ join the cut set $S$.

\begin{lemma}\label{lem:InductionOnPath}
	Recall $\eta_{0}=r_{\ell-1}$ and $s_{0}=t$. For every $q\in[0,k]$ it holds that \\
	\phantom{.}\hfill $\mathbb{E}\left[\alpha(s,t)\mid R\ge\eta_{q}\right]\le f(m,\Delta,\ell,j+1)\cdot d_{G}(s,t)-c\cdot L\cdot\ln \mu_{\ell}\cdot d_{G}(s_{q},t)$.
\end{lemma}

\begin{proof}
	The proof is by induction on $q$. The base case is $q=k$. Given
	that $R\ge\eta_{k}$, by \Cref{lem:NoCut},
	\[
	\E\left[\alpha(s,t)\mid R\ge\eta_{k}\right]\le f(m,\Delta,\ell,j+1)\cdot d_{G}(s,t)-c\cdot L\cdot\ln \mu_\ell\cdot d_{G}(s_{k},t)~.
	\]	
	For the inductive step, assume $\mathbb{E}\left[\alpha(s,t)\mid R\ge\eta_{q+1}\right]\le f(m,\Delta,\ell,j+1)\cdot d_{G}(s,t)-c\cdot L\cdot\ln\mu_{\ell}\cdot d_{G}(s_{q+1},t)$.
	Suppose that $R\ge\eta_{q}$. 
	If $R\in[\eta_q,\rho_{q+1})$, then by \Cref{lem:NoCut} we have $\E\left[\alpha(s,t)\mid R\in[\eta_{q},\rho_{q+1})\right]\le f(m,\Delta,\ell,j+1)\cdot d_{G}(s,t)-c\cdot L\cdot\ln\mu_{\ell}\cdot d_{G}(s_{q},t)$, as required.
	Thus we assume $R\ge\rho_{q+1}$. 	
	Using the law of total expectation, 
	\begin{align*}
		\E\left[\alpha(s,t)\mid R\ge\rho_{q+1}\right] 
		& \le\Pr\!\left[R\in[\rho_{q+1},\eta_{q+1})\mid R\ge\eta_{q}\right]\cdot\E\left[\alpha(s,t)\mid R\in[\rho_{q+1},\eta_{q+1})\right]\\
		& \quad+\Pr\!\left[R\ge\eta_{q+1}\mid R\ge\eta_{q}\right]\cdot\E\left[\alpha(s,t)\mid R\ge\eta_{q+1}\right]~.
	\end{align*}
	To bound the first summand, note that by \Cref{lem:etaIrhoI}
	$\eta_{q+1}-\rho_{q+1}\le d_G(s_{q+1},t_{q+1})$. In addition, we are assuming $d_G(s_{q+1},t_{q+1})\le d_G(s,t)\le \frac{1}{\lambda_\ell}$, and thus $\lambda_{\ell}\cdot(\eta_{q+1}-\rho_{q+1})\le1$.
	Recall that by \Cref{obs:alphaDominating}, $\E\left[\alpha(s,t)\mid R\in[\rho_{q+1},\eta_{q+1})\right]\le3\Delta$.
	As we currently assume that $R$ is sampled using the ordinary (untruncated) exponential distribution with parameter $\lambda_\ell$, it holds that%
	\submit{\todo{$d_G(s,t)\le\frac{1}{\lambda_\ell}$}}
	\begin{align*}
		& \Pr\left[R\in[\rho_{q+1},\eta_{q+1})\mid R\ge\rho_{q+1}\right]\cdot\E\left[\alpha(s,t)\mid R\in[\rho_{q+1},\eta_{q+1})\right]\\
		& \qquad\quad\le\quad\left(1-e^{-\lambda_{\ell}\cdot(\eta_{q+1}-\rho_{q+1})}\right)\cdot3\Delta\quad\le\quad\lambda_{\ell}\cdot d_{G}(s_{q+1},t_{q+1})\cdot3\Delta~.
	\end{align*}
	Using the inductive assumption (and the naive bound $\Pr\!\left[ R\ge\eta_{q+1}\mid R\ge\eta_{q}\right]\le 1$), it follows that
	\begin{align*}
		\E\left[\alpha(s,t)\mid R\ge\rho_{q+1}\right] 
		& \le \lambda_{\ell}\cdot d_{G}(s_{q+1},t_{q+1})\cdot3\Delta + f(m,\Delta,\ell,j+1)\cdot d_{G}(s,t) - c\cdot L\cdot\ln\mu_{\ell}\cdot d_{G}(s_{q+1},t)\\
		& \overset{(*)}{\le} f(m,\Delta,\ell,j+1)\cdot d_{G}(s,t) - c\cdot L\cdot\ln\mu_{\ell}\cdot d_{G}(s_{q},t)~.
	\end{align*}
	To see why inequality $^{(*)}$ holds, recall from \Cref{obs:lambdaUB} that 
	$\lambda_\ell\le\frac{64}{\Delta}\cdot L\cdot\log\mu_{\ell}$,
	or equivalently $\lambda_{\ell}\cdot3\Delta\le192\cdot L\cdot\log\mu_{\ell}=\frac{192}{\ln2}\cdot L\cdot\ln\mu_{\ell}<c\cdot L\cdot\log\mu_{\ell}$ (recall $c=384>\frac{192}{\ln2}$).
	Thus 
	\begin{align*}
		c\cdot L\cdot\ln\mu_{\ell}\cdot d_{G}(s_{q+1},t) & =c\cdot L\cdot\ln\mu_{\ell}\cdot d_{G}(s_{q+1},s_{q})+c\cdot L\cdot\ln\mu_{\ell}\cdot d_{G}(s_{q},t)\\
		& \ge\lambda_{\ell}\cdot3\Delta\cdot d_{G}(s_{q+1},t_{q+1})+c\cdot L\cdot\ln\mu_{\ell}\cdot d_{G}(s_{q},t)~.
	\end{align*}
	and the lemma follows.	
\end{proof}

Applying \Cref{lem:InductionOnPath} to $q=0$ (recall $\eta_{0}=r_{\ell-1}$ and $s_{0}=t$),
\[
\mathbb{E}\left[\alpha(s,t)\mid R\ge r_{\ell-1}\right]\le f(m,\Delta,\ell,j+1)\cdot d_{G}(s,t)-c\cdot L\cdot\log\mu_{\ell}\cdot d_{G}(s_{0},t)=f(m,\Delta,\ell,j+1)\cdot d_{G}(s,t)~.
\]

Finally, it is time to incorporate truncation. Let $\Psi$ be the
event that we sampled a value $ R\ge r_{\ell}$. The probability of
this event is
\[
\Pr\left[\Psi\right]=\Pr\left[ R>r_{\ell}\mid R\ge r_{\ell-1}\right]=e^{-\lambda_{\ell}\cdot(r_{\ell}-r_{\ell-1})}=e^{-\ln(\tfrac{2\mu_{\ell}}{\delta})}=\tfrac{\delta}{2\mu_{\ell}}~.
\]
In particular,
\[
\frac{1}{\Pr\left[\overline{\Psi}\right]}=\frac{1}{1-\frac{\delta}{2\mu_{\ell}}}<1+\frac{\delta}{\mu_{\ell}}\le e^{\frac{\delta}{\mu_{\ell}}}~.
\]
Let $g$ be the density function of the ordinary exponential distribution
we are using (formally, $g(x)=\lambda_\ell\cdot e^{-\lambda_\ell\cdot (x-r_{\ell-1})}$ for $x\ge r_{\ell-1}$).
Note that the density function of $\Texp_{[r_{\ell-1},r_\ell ]}(\lambda_\ell)$ is $\frac{g(x)}{\Pr\left[\overline{\Psi}\right]}$. 
We conclude
\begin{align}
	\mathbb{E}\left[\alpha(s,t)\mid\overline{\Psi}\right] 
	&= \int_{r_{\ell-1}}^{r_{\ell}}\frac{g(x)}{\Pr\left[\overline{\Psi}\right]}\cdot\mathbb{E}\left[\alpha(s,t)\mid R=x\right]\,dx \nonumber\\
	& \le \frac{1}{\Pr\left[\overline{\Psi}\right]}\cdot\int_{r_{\ell-1}}^{\infty}g(x)\cdot\mathbb{E}\left[\alpha(s,t)\mid R=x\right]\,dx \label{eq:ExpandingToAllValuesOfR}\\
	& \le e^{\frac{\delta}{\mu_{\ell}}}\cdot f(m,\Delta,\ell,j+1)\cdot d_{G}(s,t)\nonumber\\
	& = f(m,\Delta,\ell,j)\cdot d_{G}(s,t)~,\nonumber
\end{align}
as required.

\section{Removing the Polynomial Aspect Ratio Assumption}\label{sec:RemoveAspectRatio}
During the proof of \Cref{thm:mainExpectedDistortionAspectRatio} we assumed that all the edge weights are polynomial (and at least $1$).
In this section we remove this assumption.
We will slightly change \Cref{alg:DigraphPartition}. 
Specifically, the only change is that we will use an alternative weight function on the edges: 
$$\brick{w'(e)}:=\begin{cases}
	w(e) & w(e)\ge\frac{\Delta}{n^{7}}\\
	0 & w(e)<\frac{\Delta}{n^{7}}
\end{cases}~.$$
That is, for every edge with weight less than $\frac{\Delta}{n^{7}}$, its weight is changed to $0$.
In more detail: the parameter $\Delta$ is provided to \Cref{alg:DigraphPartition} by \Cref{alg:LaminarPartition} w.r.t.\ the real weight function $w$.
Then, \Cref{alg:DigraphPartition} internally uses a different weight function (to look for potential centers in \cref{line:WhileCurve}, and to carve balls in \cref{line:CurveBall}).
This is the only change.
Note that as a result of this change, in \Cref{alg:DigraphPartition}, an edge $e=(u,v)$ of original weight at most $\frac{\Delta}{n^7}$ will never enter the cut set $S$.
Indeed, every outgoing ball $B^+_{G[A]}(x_j,R)$ containing $u$ will also contain $v$ (resp.\ every ingoing ball $B^-_{G[A]}(x_j,R)$ containing $v$ will also contain $u$).

All our claims, observations, and lemmas in \Cref{sec:algortihm} go through, as they do not depend on the random choices made by the algorithm.
The only claim that requires attention is \Cref{clm:PartitionIterationRDelta}.
Specifically, in \Cref{clm:PartitionIterationRDelta} we proved that for every SCC $C$ fully contained in $\cR$, it holds that $\diam(C,G)\le2r_L\le\frac\Delta2$.
However, this will now hold only w.r.t.\ the weight function $w'$ we are using during the proof.
Note that w.r.t.\ the actual weight function $w$, the weight of a shortest path might increase by up to $n\cdot \frac{\Delta}{n^{7}}=\frac{\Delta}{n^{6}}$.
It follows that
$$\diam(C,G)\le2r_L+\frac{\Delta}{n^6}\le\frac\Delta2~,$$
where in the last inequality we used \Cref{obs:rellSize} ($r_{L}=\frac{\Delta}{4}-\frac{\Delta}{2^{L+4}}\le\frac{\Delta}{4}-\frac{\Delta}{32\log(m+1)}$).
Thus the part we need from \Cref{clm:PartitionIterationRDelta} still holds. We conclude that all of \Cref{sec:algortihm} goes through.
In particular, the number of edges in $D_1$ is bounded by $O(n\log n+m\log^2(n\Delta))$ (see \Cref{clm:SparsityOfDAGs}).

Before we turn to \Cref{sec:ExpectedDistortion}, we need an additional definition. 
We will think of the DAGs $D_1,D_2$ as being constructed gradually.
Recall that the DAGs are constructed w.r.t.\ the laminar topological order $(\cP,<_D)$.
The DAG $D_2$ will remain as is. For $D_1$ we will use a different digraph (for the sake of analysis only).
Let \EMPH{$G_\xi$} be the digraph $D_1$, where for every cluster $C\in\cP$ of diameter $\Delta_C<\xi$, we keep in $G_\xi$ all the internal edges between vertices of $C$. 
Formally, 
Let $\brick{\cP^\xi_{\min}}$
be the maximal clusters in $\cP$ of diameter less than $\xi$. In other words, these are clusters $C\in\cP$ such that $\Delta_C<\xi$, and for every $C'\in\cP$ with $C'\supsetneq C$, $\Delta_{C'}\ge\xi$.
Note that $\cP^\xi_{\min}$ is a partition of $V$, since every vertex is contained in some cluster of diameter less than $\xi$ (by the singleton property of \Cref{def:Laminar}), and the set of maximal clusters is disjoint.
Let $\brick{\cP^\xi}=\left\{C\in\cP\mid\exists C'\in \cP^\xi_{\min} \mbox{ s.t. } C'\subseteq C\right\}$ be all clusters of $\cP$ that are either in $\cP^\xi_{\min}$, or have diameter at least $\xi$.
The digraph $G_\xi$ consists of the following edges:
\begin{itemize}
	\item The $2$-hop spanner $E_{H_1}$ (same as in $D_1$).
	\item For every edge $e=(u,v)\in D=E\setminus\widetilde{S}$, and every two clusters $C_u,C_v\in\cP^\xi$ such that $C_u$ contains $u$ but not $v$, and $C_v$ contains $v$ but not $u$, we add to $G_\xi$ the four edges ${e=(u,v)},(u,C_v^\first),(C_u^\last,v),(C_u^\last,C_v^\first)$.
	\item Internal edges in small clusters: $E_{\xi}=\left\{ (u,v)\in E\mid\exists C\in\cP^{\xi}_{\min}\text{ s.t. }u,v\in C\right\}$.
\end{itemize}

One can think of the algorithm constructing $D_1$ as follows:
Initially $\cP$ consists of all the SCCs in $G$.
Then for every SCC $C\in\cP$ we run \Cref{alg:DigraphPartition} to find a set of edges $S_C$. We remove the edges $\cup_{C\in\cP} S_C$ and obtain additional SCCs which we add to $\cP$. For every edge $e\notin \widetilde{S}$ between two different SCCs we add some additional edges to $D_1$. We then continue this process on all the new SCCs. Eventually, once all the created SCCs are treated, we obtain $D_1$. If we stop this process on each cluster $C\in\cP$ that reaches diameter below $\xi$, we obtain $G_\xi$.
In this sequential process, all the edges added to $D_1$ are only between vertices already reachable. We observe:
\begin{observation}\label{obs:reachabiltyDimpliesGxi}
	$u\rightsquigarrow_{D_1}v$ $\Rightarrow$ $u\rightsquigarrow_{G_\xi}v$.
\end{observation}
In addition, we observe from the sequential construction of $G_\xi$:
\begin{observation}\label{obs:GxiSCC}
	The set of SCCs in $G_\xi$ is $\cP^\xi_{\min}$.
\end{observation}

We proceed to modify \Cref{sec:ExpectedDistortion}, where the main idea is to work with $G_\xi$ instead of $D_1$ (for an appropriate value of $\xi$). 
\Cref{clm:DistanceUsing2HopSpanner} ($\min\{d_{D_1}(s,t),d_{D_2}(s,t)\}\le 2\Delta_{st}$) goes through (as we added the $2$-hop spanner).
We define $\alpha$ w.r.t.\ $G_\xi$:
$$\brick{\alpha_\xi(s,t)}:=d_{G_{\xi}}(s,t)\cdot\mathds{1}[s\rightsquigarrow_{G_\xi}t]+3\Delta_{st}\cdot\mathds{1}[s\not\rightsquigarrow_{G_\xi}t]~.$$
That is, we replace the order $<_D$ with reachability in $G_\xi$, and replace distance in $D_1$ by distance in $G_\xi$.
Consider \Cref{obs:alphaDominating}; clearly $\alpha_\xi(s,t)\le3\Delta$ (as we added the $2$-hop spanner edges).
The main part of \Cref{obs:alphaDominating} requires an additional argument. We state it below in \Cref{clm:alphaXiDominating}.

\Cref{lem:alphaAppliedRecursively} follows by the exact same lines (replacing $\le_D$ by $\rightsquigarrow_{G_\xi}$). 
This assumes $\xi$ is smaller than the current diameter.
\begin{lemma}\label{lem:alphaXiAppliedRecursively}
	Consider a cluster $U\in\cP$, on which we execute \Cref{alg:LaminarPartition}, where $\Delta_U\ge\xi$.	
	Consider two different SCCs $C_s,C_t$ of $G[U]\setminus S$ such that there is an edge $(u,v)\in \pi$ with $(u,v)\notin S$, $\pi[s,u]\subseteq C_s$ and $\pi[v,t]\subseteq C_t$. Then 
	$$\alpha_\xi(s,t)\le\alpha_\xi(s,u)+d_G(u,v)+\alpha_\xi(v,t)~.$$
\end{lemma}
\Cref{cor:alphaAppliedRecursivelyMiddle} holds in the same manner (we will not state it explicitly).
We are now ready to prove the appropriate version of \Cref{obs:alphaDominating}.

\begin{claim}\label{clm:alphaXiDominating}
	For every $s,t\in V$, $d_{D_{1}}(s,t)\cdot\mathds{1}[s\le_{D}t]+d_{D_{2}}(s,t)\cdot\mathds{1}[s\ge_{D}t]\le\alpha_\xi(s,t)+3n\xi$.
\end{claim}
\begin{proof}
	Fix $s,t\in V$. If $s\not\rightsquigarrow_Gt$, then $\Delta_{st}=\infty$, and $s\not\rightsquigarrow_{G_\xi}t$. In particular $\alpha_\xi(s,t)=\infty$, and there is nothing to prove.
	We will thus assume $s\rightsquigarrow_Gt$.
	
	If $s\not\rightsquigarrow_{G_\xi}t$, then $s\not\rightsquigarrow_{D_1}t$ (\Cref{obs:reachabiltyDimpliesGxi}). 
	It necessarily holds that $s>_Dt$.
	Indeed, if $s$ and $t$ belong to the same SCC $C\in G$, then $C\in\cP$, and if $s<_Dt$ then $E_{H_1}$ edges would ensure $s\rightsquigarrow_{D_1}t$.
	Otherwise, $s$ and $t$ belong to different SCCs of $G$. In this case the edges of $E_{H_1}$ and the additional edges we added to $D_1$ ensure $s\rightsquigarrow_{D_1}t$ (in a similar fashion to the arguments leading to \cref{eq:alphaSTforPath}). 
	As $s>_Dt$, it follows that 
	$d_{D_{1}}(s,t)\cdot\mathds{1}[s\le_{D}t]+d_{D_{2}}(s,t)\cdot\mathds{1}[s\ge_{D}t]=d_{D_{2}}(s,t)\le2\Delta_{st}\le\alpha_\xi(s,t)$, and the claim follows.
	We will thus assume $s\rightsquigarrow_{G_\xi}t$, and hence $\alpha_\xi(s,t)=d_{G_{\xi}}(s,t)$.
	
	Let $\pi_\xi$ be the shortest $s$–$t$ path in $G_\xi$.
	By \Cref{obs:GxiSCC}, the set of SCCs of $G_\xi$ is $\cP_{\min}^\xi$.
	We apply \Cref{clm:path_struct} on $G_\xi$ and $\pi_\xi$ to obtain SCCs $C_1,\dots,C_k\in \cP_{\min}^\xi$, and edges $\{e_j=(u_j,v_{j+1})\}_{j=1}^{k-1}$ such that $e_j\in C_j\times C_{j+1}$, $\pi_\xi[s=v_1,u_1]\subseteq C_1$, $\pi_\xi[v_k,t=u_k]\subseteq C_k$, and for every $j\in[2,k-1]$, $\pi_\xi[v_j,u_j]\subseteq C_j$.
	Note that by \Cref{obs:alphaDominating} (using the original $\alpha$), it holds that $\alpha(u_j,v_j)\le3\Delta_{C_j}$.
	Following \cref{eq:alphaSTforPath}, and as the diameters of all clusters in $\cP_{\min}^\xi$ are $<\xi$, it holds that 
	\begin{align*}
		d_{D_{1}}(s,t) & \overset{(\ref{eq:alphaSTforPath})}{\le}\sum_{i=1}^{k-1}d_{G}(u_{i},v_{i+1})+\sum_{i=1}^{k}\alpha(u_{i},v_{i})\\
		& \le d_{G_{\xi}}(s,t)+\sum_{i=1}^{k}3\Delta_{C_{i}}<\alpha_{\xi}(s,t)+3n\cdot\xi~.\qedhere
	\end{align*}
\end{proof}

Next, we turn to bounding the expected value of $\alpha_\xi(s,t)$.
We will use a slightly different potential function. Specifically, we normalize $\Delta$ by $\frac\xi2$. Set $\Delta^\xi=\max\{1,\frac{\Delta}{\xi/2}\}$. Then:
\[
\brick{f_\xi(m,\Delta,\ell,j)}=c\cdot L\cdot\ln m\cdot e^{\delta\cdot\left(L\cdot\log(m\Delta^\xi)+(\ell-1)+\frac{\mu_{\ell}+1-j}{\mu_{\ell}}\right)}~.
\]

The inductive hypothesis is similar, replacing $\alpha$ and $f$ by $\alpha_\xi$ and $f_\xi$. However, we need the inductive hypothesis only for clusters in $\cP^\xi$ (that is, for clusters with diameter at least $\xi$, or maximal clusters with smaller diameter).
\begin{inductiveHypothesis}\label{ih:mainAltAspectRatio}
	Suppose that during the execution of \Cref{alg:LaminarPartition} for a cluster $U\in\cP_\xi$, we call \texttt{Digraph-Partition-recursive}, and in one of the recursive calls, there is a call $\texttt{Digraph-Partition-recursive}(G[A],\Delta,m,S,\ell,j)$.
	Then for every $s,t\in A$
	$$\E[\alpha_\xi(s,t)]\le f_\xi(m,\Delta,\ell,j)\cdot d_{G[A]}(s,t)~.$$
\end{inductiveHypothesis}

Fix $s,t\in A$.
If $s\not\rightsquigarrow_{G[A]}t$, then $d_{G[A]}(s,t)=\infty$, and \Cref{ih:mainAltAspectRatio} holds trivially. 
We will thus assume $s\rightsquigarrow_{G[A]}t$.
The base case of the induction is when $\Delta<\xi$ (this is possible if $U\in\cP^\xi_{\min}$).
Note that as $\Delta^\xi\ge1$, $f_\xi(m,\Delta,\ell,j)\ge 1$ (note that $s\rightsquigarrow_{G[A]}t$ implies $m\ge2$).
In this case $G[A]\subseteq G_\xi$ (as we take all the internal edges of $G[U]$).
It follows that $\E[\alpha_{\xi}(s,t)]=\E[d_{G_{\xi}}(s,t)]\le d_{G[A]}(s,t)\le f_{\xi}(m,\Delta,\ell,j)\cdot d_{G[A]}(s,t)$, and thus the induction follows.

Next, suppose that $\Delta\ge\xi$.
\Cref{obs:stLong} holds.
The case where no center $x_j$ is found and $\ell>1$ holds as $f_\xi(m,\Delta,\ell,j)\ge f_\xi(m,\Delta,\ell-1,1)$.
The case where no center $x_j$ is found and $\ell=1$ holds as $f_\xi(m,\frac{\Delta}{2},L,1)\le f_\xi(m,\Delta,1,\mu_1+1)$.
Indeed, as $\Delta\ge\xi$, $(\frac{\Delta}{2})^{\xi}=\max\{1,\frac{\Delta/2}{\xi/2}\}=\frac\Delta\xi=\frac{\Delta^\xi}{2}$, and thus the exact same argument as in the original proof goes through.

We thus assume that a center $x_j$ was found.
The exact same analysis (now using \Cref{lem:alphaXiAppliedRecursively}) goes through. 
The only delicate point is to notice that during the proof of \Cref{clm:DistNoCutsqTot} it holds that 
$f_\xi(\frac{m}{\mu_{\ell}},\Delta,L,1)\le f_\xi(m,\Delta,\ell,j+1)-c\cdot L\cdot\ln\mu_{\ell}$.
All other arguments work in exactly the same way. We conclude that \Cref{ih:mainAltAspectRatio} holds.

We are now ready to prove \Cref{thm:mainExpectedDistortion}.
Fix $s,t\in V$ such that $s\rightsquigarrow_Gt$ (the theorem holds trivially if $s\not\rightsquigarrow_Gt$).
Let $\pi$ be the shortest $s$–$t$ path in $G$.
Set $\brick{\Upsilon}=d_G(s,t)\cdot n^{7}$ and $\brick{\xi}=\frac{1}{n}\cdot d_G(s,t)$.
Consider $G_\Upsilon$, which is the current graph we are holding when we have stopped \Cref{alg:LaminarPartition} on each cluster of diameter less than $\Upsilon$.
Observe that $\pi\subseteq G_\Upsilon$.
Indeed, in every execution of \Cref{alg:DigraphPartition} until this point, the diameter was $\Delta\ge\Upsilon$. As each edge $e\in\pi$ has weight at most $w(e)\le d_G(s,t)=\frac{\Upsilon}{n^{7}}\le\frac{\Delta}{n^{7}}$, in each execution of \Cref{alg:DigraphPartition} the weight of $e$ was changed to $0$, and hence it was not added to $S$. It follows that $\pi\subseteq G_\Upsilon$. In particular $d_{G_\Upsilon}(s,t)=d_G(s,t)$.

By \Cref{obs:GxiSCC}, the set of SCCs of $G_\Upsilon$ is $\cP_{\min}^\Upsilon$.
We apply \Cref{clm:path_struct} on $G_\Upsilon$ and $\pi$ to obtain SCCs $C_1,\dots,C_k\in \cP_{\min}^\Upsilon$, and edges $\{e_j=(u_j,v_{j+1})\}_{j=1}^{k-1}$ such that $e_j\in C_j\times C_{j+1}$, $\pi[s=v_1,u_1]\subseteq C_1$, $\pi[v_k,t=u_k]\subseteq C_k$, and for every $j\in[2,k-1]$, $\pi[v_j,u_j]\subseteq C_j$.
For every $j$, let $\Delta_j$ be the diameter of $C_j$, and $m_j=|E(G[C_j])|$. By the definition of $G_\Upsilon$, $\Delta_j\le\Upsilon$. 
Using \Cref{ih:mainAltAspectRatio}, it holds that
\[
\E[\alpha_\xi(v_j,u_j)]\le f_\xi(m_j,\Delta_j,L,1)\cdot d_{G_\Upsilon}(v_j,u_j)\le f_\xi(m,\Upsilon,L,1)\cdot d_{G_\Upsilon}(v_j,u_j)~.
\]
Note that
\begin{align*}
	f_{\xi}(m,\Upsilon,L,1) & =c\cdot L\cdot\ln m\cdot e^{\delta\cdot\left(L\cdot\log(m\Upsilon^{\xi})+(L-1)+\frac{\mu_{L}+1-1}{\mu_{L}}\right)}\\
	& =c\cdot L\cdot\ln m\cdot e^{\delta\cdot\left(L\cdot\log\!\left(4m\frac{\Upsilon}{\xi}\right)\right)}\\
	& =c\cdot L\cdot\ln m\cdot e^{\delta\cdot\left(L\cdot\log(4n^{2+7+1})\right)}=O(\log n\cdot\log\log n)~. 
\end{align*}
Using the exact same arguments leading to \cref{eq:alphaSTforPath} (replacing \Cref{lem:alphaAppliedRecursively} with \Cref{lem:alphaXiAppliedRecursively}) we get
\[
\alpha_{\xi}(s,t)=\sum_{i=1}^{k-1}d_{G}(u_{i},v_{i+1})+\sum_{i=1}^{k}\alpha_{\xi}(u_{i},v_{i})~.
\]
Applying \Cref{ih:mainAltAspectRatio}, it follows that 
\begin{align*}
	\E[\alpha_{\xi}(s,t)] & =\sum_{i=1}^{k-1}d_{G}(u_{i},v_{i+1})+\sum_{i=1}^{k}\E\left[\alpha_{\xi}(u_{i},v_{i})\right]\\
	& \le\sum_{i=1}^{k-1}d_{G}(u_{i},v_{i+1})+O(\log n\cdot\log\log n)\cdot\sum_{i=1}^{k}d_{G_{\xi}}(u_{i},v_{i})\\
	& =O(\log n\cdot\log\log n)\cdot d_{G}(s,t)~.
\end{align*}

Using now \Cref{clm:alphaXiDominating}, we conclude
\begin{align*}
	\E\left[d_{D_{1}}(s,t)\cdot\mathds{1}[s\le_{D}t]+d_{D_{2}}(s,t)\cdot\mathds{1}[s\ge_{D}t]\right] & \le\E\left[\alpha_{\xi}(s,t)\right]+3n\xi\\
	& =O(\log n\cdot\log\log n)\cdot d_{G}(s,t)+3n\cdot\frac{1}{n}\cdot d_{G}(s,t)\\
	& =O(\log n\cdot\log\log n)\cdot d_{G}(s,t)~.
\end{align*}

\section{Efficient Implementation}\label{sec:runtime}
In this section we explain how to sample from the stochastic embedding in $\tilde{O}(m)$ time.
Specifically, we assume here that the aspect ratio is polynomial, and introduce several modifications to \Cref{alg:LaminarPartition,alg:DigraphPartition} and to the construction of the DAGs $D_1,D_2$.
We then argue that the runtime is reduced to $\tilde{O}(m)$, while w.h.p. we sample DAGs from the desired distribution. However, with polynomially small probability we may make a mistake and sample DAGs with no guarantees.
The algorithm has several steps that are computationally expensive. Specifically:
\begin{enumerate}
	\item Diameter computation in \Cref{line:diamCompute} of \Cref{alg:LaminarPartition}.
	\item Computing the number of edges in the balls $B_{G[A]}^{*}(x_{j},r_{\ell})$ in \Cref{line:WhileCurve} of \Cref{alg:DigraphPartition}.
	\item Computing the edge weight $d_G(u,v)$ of the edges $(u,v)$ we are adding to either $D_1$ or $D_2$.
\end{enumerate}

Ignoring the time required for these three components, we can implement the algorithm in $\tilde{O}(m)$ time.
Indeed, \Cref{alg:LaminarPartition} is a recursive algorithm of logarithmic depth. It computes the number of edges (linear time), calls \Cref{alg:DigraphPartition}, and then finds the SCCs in $G[U]\setminus S$ and topologically orders them ($\tilde{O}(m)$ time using DFS).
Executing \Cref{alg:DigraphPartition} also takes $\tilde{O}(m)$ time (other than computing the density of the balls), since each edge can be carved only once.
Finally, when we have the laminar topological order $(\cP,<_D)$ we can find in $\tilde{O}(m)$ time all the edges in $D_1,D_2$.
We continue by describing the required modifications.
Our modifications for issues (1) and (3) follow \cite{AHW25}, while the modifications for issue (2) follow \cite{BFHL25}.

\subsubsection*{Avoiding Diameter Computations — modifications to \Cref{alg:LaminarPartition}.}
We assume that $G$ is strongly connected (otherwise start the execution of \Cref{alg:LaminarPartition} in \Cref{line:SCCpartition}).
\Cref{alg:LaminarPartition} will now receive the diameter $\Delta$ as part of its input.
Let $W$ be the maximal edge weight. In the first call, we will use the diameter $\Delta=n\cdot W$, which clearly bounds the diameter of every SCC.
\Cref{alg:LaminarPartition} will pass this $\Delta$ to \Cref{alg:DigraphPartition}. 
\Cref{alg:DigraphPartition} will return to \Cref{alg:LaminarPartition} the cut set $S$ (as previously), and in addition the set $\cR$ of remaining active vertices.
For every SCC $C$ of $G[U]\setminus S$, if $C\not\subseteq\cR$, we will make the call $\texttt{Laminar-Topological-Order}(G,C,\Delta)$, while for every $C\subseteq\cR$, we will make the call $\texttt{Laminar-Topological-Order}(G,C,\frac\Delta2)$.
For every cluster $C\in\cP$, let \EMPH{$\Delta_C$} be the diameter used when we made the (last) call $\texttt{Laminar-Topological-Order}(G,C,\Delta_C)$,
unless $C$ is a singleton, in which case $\Delta_C=0$.
By \Cref{clm:PartitionIterationRDelta}, the diameter parameters we are using will be dominating. That is, for every cluster $C$, $\diam(G,C)\le\Delta_C$.
This is enough to go over the entire analysis and see that all our arguments go through. 
Note that it might be that we run \Cref{alg:DigraphPartition} on an SCC $U$ with diameter $\ll\Delta$. 
In this case, the algorithm carves no balls and returns $\cR=U$, and in the next call we use $\Delta/2$. 

\subsubsection*{Avoiding Distance Computations — modifications to the construction of $D_1,D_2$.}
We describe the modifications required for the construction of $D_1$. The modifications to $D_2$ follow along the same lines.
When constructing $E_{H_1}$, for every edge $(u,v)$ computing $d_G(u,v)$ is too costly.
We therefore construct $E_{H_1}$ differently.
Specifically, we do not construct a global $2$-hop spanner.
Instead, for every SCC $C\in\cP$, we construct a $2$-hop spanner $E^1_C$ over the vertices of $C$, w.r.t. the order $<_D$.
The weight of all the added edges will be $\Delta_C$.
Note that these distances dominate the real distances, while \Cref{clm:DistanceUsing2HopSpanner} still holds.
The total number of edges added in this way is slightly higher. 
By \Cref{obs:RecursionDepth} every vertex belongs to $O(\log n)$ such $2$-hop spanners, adding at most $O(\log n)$ edges in each spanner, so overall we add $O(\log^2 n)$ edges in this way. 
\Cref{clm:SparsityOfDAGs} still holds.

For every edge $e=(u,v)\in D=E\setminus\widetilde{S}$, and every two disjoint clusters $C_u,C_v\in\cP$ such that $u\in C_u$ and $v\in C_v$, we added to $D_1$ the edge $(C_u^\last,C_v^\first)$ using the original distances as weights.
Instead, we give it the weight $\Delta_{C_u}+w(e)+\Delta_{C_v}$.
This can be computed efficiently, and it is the only bound used during the proof of \Cref{lem:alphaAppliedRecursively}.
Note that some edges might have several different weights; we keep the smallest one.
Overall, as \Cref{clm:DistanceUsing2HopSpanner} and \Cref{lem:alphaAppliedRecursively} still hold, the analysis goes through.

\subsubsection*{Estimating Densities — modification to \Cref{alg:DigraphPartition}}
During the execution of \Cref{alg:DigraphPartition} we are required to find a center $x_j$ satisfying the condition in \Cref{line:WhileCurve} of \Cref{alg:DigraphPartition} for the induced graph $G[A]$ w.r.t. the ever-changing set of active vertices $A$.
This is a very costly operation.
This exact issue appears in the low-diameter decomposition algorithm of \cite{BFHL25}. 
A key observation is that it is enough to satisfy the inequalities in \Cref{line:WhileCurve} of \Cref{alg:DigraphPartition} only approximately. 
Specifically, we call the condition in \Cref{line:WhileCurve} strict:
$$\mbox{\brick{Strict inequality:}}\qquad\frac{m}{\mu_{\ell-1}}\le|B_{G[A]}^{*}(x_{j},r_{\ell-1})|\le|B_{G[A]}^{*}(x_{j},r_{\ell})|\le\frac{m}{\mu_{\ell}}~.$$
The following inequality will be called soft:
$$\mbox{\brick{Soft inequality:}}\qquad\frac12\cdot\frac{m}{\mu_{\ell-1}}\le|B_{G[A]}^{*}(x_{j},r_{\ell-1})|\le|B_{G[A]}^{*}(x_{j},r_{\ell})|\le\frac54\cdot\frac{m}{\mu_{\ell}}~.$$
Suppose that during phase $\ell$ of \Cref{alg:DigraphPartition}, each time we choose to carve a ball around center $x_j$, it satisfies only the soft inequality; however, at the same time \Cref{clm:PartitionIterationRDelta} still miraculously holds.
We argue that in this case, the analysis goes through almost unchanged.

We now go over the analysis (other than \Cref{clm:PartitionIterationRDelta}) and see that everything works.
We will then show how to implement an (efficient) algorithm that chooses centers satisfying the soft inequality in such a way that \Cref{clm:PartitionIterationRDelta} will hold (see \Cref{obs:softSatisfy} and \Cref{clm:PartitionIterationRDeltaRestated}).
First note that \Cref{clm:PartitionIteration} holds, but now every cluster that is fully contained in a ball carved during the $\ell$'th phase satisfies $|E(C)|\le\frac54\cdot\frac{m}{\mu_\ell}$.
\Cref{obs:NumberOfBalls} still holds, but now the number of iterations is at most $2\cdot\mu_{\ell-1}$.
\Cref{cor:progress} holds as well (as we assume for now \Cref{clm:PartitionIterationRDelta} holds), but with constant $\frac54\cdot\frac{m}{2}$. Still, the depth of the recursion is logarithmic.
\Cref{clm:zD1D2Dags},
\Cref{clm:zDominating},
\Cref{clm:SparsityOfDAGs},
\Cref{clm:DistanceUsing2HopSpanner},
\Cref{obs:alphaDominating},
\Cref{lem:alphaAppliedRecursively}, 
and \Cref{cor:alphaAppliedRecursivelyMiddle}
all hold and are not affected by the changes.

We now move to the expected distortion analysis. We need to slightly update the potential function so that the analysis will go through.
This will take care of the fact that now there might be $2\mu_\ell$ iterations in phase $\ell$, and of the fact that balls might be slightly larger. Set  
\[
\brick{f(m,\Delta,\ell,j)}:=c'\cdot L\cdot\ln m\cdot e^{\delta\cdot\left(2\cdot L\cdot\left(\log_{\frac{8}{5}}m+\log\Delta\right)+2\cdot(\ell-1)+\frac{2\mu_{\ell}+1-j}{\mu_{\ell}}\right)}~,
\]
where $c'$ will be a large enough constant to be determined later.
Note that this potential will also imply distortion $O(\log n\cdot\log\log n)$.
We next verify that this potential function satisfies all the properties we need.
The new potential function is good enough for the base of \Cref{ih:mainAlt}; that is, \Cref{obs:stLong} holds (with the old constant $c$, and we can pick $c'\ge c$).
We specify the other uses:
\begin{itemize}
	\item In the induction step, in the case where there is no center $x_j$ and $\ell>1$, we move to the next phase without doing anything. The induction holds here as well because
	$f(m,\Delta,\ell,2\mu_{\ell}+1)\ge f(m,\Delta,\ell-1,1)$. 
	
	\item In the induction step, in the case where there is no center $x_j$ and $\ell=1$. Here we use the induction on smaller SCCs, where by \Cref{clm:PartitionIterationRDelta}, the diameter of the relevant SCCs is bounded by $\frac\Delta2$.  
	The induction holds here as well because
	$f(m,\Delta,1,2\mu_{1}+1)\ge f(m,\frac{\Delta}{2},L,1)$.
	
	\item The proofs of \Cref{clm:DistNoCutsqTot} and \Cref{lem:InductionOnPath} are the most delicate changes. Recall that in \Cref{clm:DistNoCutsqTot} we carve a ball and argue that the part of $\pi$ in the carved ball has now smaller expected $\alpha(s_q,t)$, as $\pi[s_q,t]$ is contained in a ball with smaller cardinality. This analysis goes through in a similar manner. Specifically, now the suffix $\pi[s_q,t]$ is contained in a cluster of cardinality at most $\frac{5}{4}\cdot\frac{m}{\mu_{\ell}}$. Accordingly, we can bound $\E\!\left[\frac{\alpha(s_q,t)}{d_{G[A]}(s_q,t)}\right]$ by:
	\begin{align*}
		f\!\left(\tfrac{5}{4}\cdot\tfrac{m}{\mu_{\ell}},\Delta,L,1\right)
		&=c'\cdot L\cdot\ln\!\left(\tfrac{5}{4}\cdot\tfrac{m}{\mu_{\ell}}\right)\cdot e^{\delta\cdot\left(2L\cdot\left(\log_{\frac{8}{5}}\!\left(\tfrac{5}{4}\cdot\tfrac{m}{\mu_{\ell}}\right)+\log\Delta\right)+2(L-1)+\tfrac{2\mu_{\ell}+1-1}{\mu_{\ell}}\right)}\\
		&\le c'\cdot L\cdot\ln m\cdot e^{\delta\cdot\left(2L\cdot\left(\log_{\frac{8}{5}}m+\log\Delta\right)+2L\cdot\!\left(1-\log_{\frac{8}{5}}\!\tfrac{4}{5}\mu_{\ell}\right)\right)}-c'\cdot L\cdot\ln\!\left(\tfrac{4}{5}\mu_{\ell}\right)\\
		&\overset{(*)}{\le}c'\cdot L\cdot\ln m\cdot e^{\delta\cdot\left(2L\cdot\left(\log_{\frac{8}{5}}m+\log\Delta\right)+2(\ell-1)+\tfrac{2\mu_{\ell}+1-(j+1)}{\mu_{\ell}}\right)}-c'\cdot L\cdot\ln\!\left(\tfrac{4}{5}\mu_{\ell}\right)\\
		&=f(m,\Delta,\ell,j+1)-c'\cdot L\cdot\ln\!\left(\tfrac{4}{5}\mu_{\ell}\right)~,
	\end{align*}
	where inequality $^{(*)}$ holds as $\mu_\ell\ge\mu_L=2$, and thus $\log_{\frac{8}{5}}\!\left(\tfrac{4}{5}\mu_{\ell}\right)\ge\log_{\frac{8}{5}}\!\left(\tfrac{8}{5}\right)=1$.
	
	It follows that, as in \Cref{clm:DistNoCutsqTot}, for every $q$,
	\[
	\E\!\left[\alpha(s_{q},t)\mid R\in[\eta_{q},\rho_{q+1})\right]\le\left(f(m,\Delta,\ell,j+1)-c'\cdot L\cdot\ln\!\left(\tfrac{4}{5}\mu_{\ell}\right)\right)\cdot d_{G}(s_{q},t)\,.
	\]
	We continue to \Cref{lem:InductionOnPath}, accordingly proving that $\mathbb{E}\!\left[\alpha(s,t)\mid R\ge\eta_{q}\right]\le f(m,\Delta,\ell,j+1)\cdot d_{G}(s,t)-c'\cdot L\cdot\ln\!\left(\tfrac{4}{5}\mu_{\ell}\right)\cdot d_{G}(s_{q},t)$.
	It still holds that 
	\[
	\Pr\!\left[R\in[\rho_{q+1},\eta_{q+1})\mid R\ge\rho_{q+1}\right]\cdot\E\!\left[\alpha(s,t)\mid R\in[\rho_{q+1},\eta_{q+1})\right]\le\lambda_{\ell}\cdot d_{G}(s_{q+1},t_{q+1})\cdot3\Delta~.
	\]
	Using the inductive assumption, as in the original proof, we conclude:
	\begin{align*}
		\E\!\left[\alpha(s,t)\mid R\ge\rho_{q}\right]
		&\le\lambda_{\ell}\cdot d_{G}(s_{q+1},t_{q+1})\cdot3\Delta+f(m,\Delta,\ell,j+1)\cdot d_{G}(s,t)\\[-2pt]
		&\qquad-c'\cdot L\cdot\ln\!\left(\tfrac{4}{5}\mu_{\ell}\right)\cdot d_{G}(s_{q+1},t)\\
		&\overset{(**)}{\le}f(m,\Delta,\ell,j+1)\cdot d_{G}(s,t)-c'\cdot L\cdot\ln\!\left(\tfrac{4}{5}\mu_{\ell}\right)\cdot d_{G}(s_{q},t)~.
	\end{align*}
	To see why inequality $^{(**)}$ holds, recall that during the original proof of \Cref{lem:InductionOnPath}, we showed that $\lambda_{\ell}\cdot3\Delta\le c\cdot L\cdot\ln\mu_{\ell}$. 	
	We now choose the constant $c'$ large enough so that $\lambda_{\ell}\cdot3\Delta\le c\cdot L\cdot\ln\mu_{\ell}\le c'\cdot L\cdot\ln\!\left(\tfrac{4}{5}\mu_{\ell}\right)$. It follows that 
	\begin{align*}
		c'\cdot L\cdot\ln\!\left(\tfrac{4}{5}\mu_{\ell}\right)\cdot d_{G}(s_{q+1},t)
		&=c'\cdot L\cdot\ln\!\left(\tfrac{4}{5}\mu_{\ell}\right)\cdot d_{G}(s_{q+1},s_{q})+c'\cdot L\cdot\ln\!\left(\tfrac{4}{5}\mu_{\ell}\right)\cdot d_{G}(s_{q},t)\\
		&\ge\lambda_{\ell}\cdot3\Delta\cdot d_{G}(s_{q+1},t_{q+1})+c'\cdot L\cdot\ln\!\left(\tfrac{4}{5}\mu_{\ell}\right)\cdot d_{G}(s_{q},t)\,,
	\end{align*}
	and the lemma follows.		
	\item At the end of the proof, to incorporate truncation (see \cref{eq:ExpandingToAllValuesOfR}), we used that $e^{\frac{\delta}{\mu\ell}}\cdot f(m,\Delta,\ell,j+1)=f(m,\Delta,\ell,j)$. This equality still holds.
\end{itemize}

\paragraph*{Implementing soft inequalities.}
We showed that if we carve balls around centers satisfying the soft inequality, while \Cref{clm:PartitionIterationRDelta} still holds, then the entire analysis works. We now turn to show how to do it efficiently.
We follow the steps of \cite{BFHL25}.
The key tool is Cohen's algorithm \cite{Cohen97} for efficiently approximating densities around all the vertices. 
\begin{lemma}[Cohen's Algorithm \cite{Cohen97}] \label{lem:cohen}
	Let $G$ be a directed weighted graph, let $r \geq 0$ and $\epsilon > 0$. There is an algorithm that runs in time $O(m \epsilon^{-2} \log^3 n)$ and computes approximations $\est(v,r)$ satisfying w.h.p. that
	\begin{equation*}
		(1-\eps)\cdot|B^{+}_G(v, r)| \leq \est(v,r) \leq (1 + \epsilon) |B^{+}_G(v, r)|~.
	\end{equation*}
\end{lemma}

One might hope that we can simply use \Cref{lem:cohen} to estimate densities and decide where to carve balls accordingly.
Unfortunately, this is not enough.
Indeed, one can estimate all the densities $B^{+}_{G[A]}(v, r_\ell)$ for all the active vertices $v\in A$ using \Cref{lem:cohen}. However, after we carve some ball $B_{\ell,j}$ these estimates are no longer valid, as the set of active vertices is reduced.
Instead, following \cite{BFHL25}, it is enough to use \Cref{lem:cohen} only $O(\log n)$ times.
We will use \Cref{lem:cohen} with some arbitrary fixed constant, say $\brick{\eps}=0.1$. In addition, we will assume that all the executions of \Cref{lem:cohen} are successful (and eventually argue that the algorithm succeeds w.h.p.).

In the original \Cref{alg:DigraphPartition} we had $L$ phases, and at most $\mu_\ell$ iterations at phase $\ell$.
Now, in addition, each phase $\ell$ will have two parts (outgoing and ingoing), and also \EMPH{rounds}, starting at $q=1$ and going up to $O(\log n)$.
Each round will have $\brick{\kappa}=\mu_\ell\cdot O(\log n)$ iterations (similarly to \Cref{alg:DigraphPartition}). 
We use $\ell$ for the phase count; there are two parts: ingoing and outgoing; $q$ is the round count ($O(\log n)$ per phase/part), and $j$ is the iteration count ($\kappa$ per round).
We describe now the outgoing part. The ingoing part is executed afterwards and is symmetric.

Let \EMPH{$A_\ell$} be the set of active vertices at the beginning of phase $\ell$.
We also denote by \EMPH{$A_{\ell,q}$} the set of active vertices at the beginning of round $q$, and by \EMPH{$A_{\ell,q,j}$} the set of active vertices at the beginning of round $q$, iteration $j$.

\paragraph*{Initialization of the $\ell$ Phase, outgoing Part.} 
Use \Cref{lem:cohen} to get estimates $\left\{\est(v,r_\ell)\right\}_{v\in A_{\ell}}$, where 
$(1-\eps)\cdot|B^{+}_{G[ A_{\ell}]}(v, r_\ell)|\le\est(v,r_\ell)\le(1+\eps)|B^{+}_{G[ A_{\ell}]}(v, r_\ell)|$.
Let 
$$\brick{\cQ}=\left\{v\in A_{\ell}\mid\est(v,r_\ell) \leq (1 + \eps)\cdot\frac{m}{\mu_\ell} \right\}$$ 
be the set of \EMPH{good vertices}.
That is, initially a vertex is called good if its out-ball of radius $r_\ell$ is not too dense. Note that for every good vertex it holds that
$|B^{+}_{G[ A_{\ell,q}]}(v, r_\ell)|\le (1+\eps)\cdot\est(v,r_\ell)\le(1+\eps)^2\cdot \frac{m}{\mu_\ell}\le\frac54\cdot \frac{m}{\mu_\ell}$.
On the other hand, for every vertex $v$ such that
$|B^{+}_{G[ A_{\ell,q}]}(v, r_\ell)|\le\frac{m}{\mu_\ell}$ it holds that $\est(v,r_\ell)\le(1+\eps)\cdot \frac{m}{\mu_\ell}$, and thus $v$ is good (at least initially).
Our set of good vertices is always decreasing; that is, some good vertices become \EMPH{bad} (i.e., not good) or inactive, but no bad vertex ever becomes good.
We denote by $\brick{Q_{\ell,q}}$ the set of good vertices at the beginning of round $q$, and by $\brick{Q_{\ell,q,j}}$ the set of good vertices after the $j$'th iteration of round $q$.
There will be $O(\log n)$ rounds.

\paragraph*{Description of Round $q$.}
Round $q$ consists of $\kappa$ iterations, followed by a finalization step. We then move to the next round.
At iteration $j$, we sample u.a.r. a good vertex $x_j\in \cQ_{\ell,q,j-1}$, and check (manually) whether $|B^{+}_{G[ A_{\ell,q,j-1}]}(x_j, r_{\ell-1})|\ge\frac12\cdot\frac{m}{\mu_{\ell-1}}$. 
We can run Dijkstra, and once we see $\frac12\cdot\frac{m}{\mu_{\ell-1}}$ edges we can stop. Thus the running time for this check is $\tilde{O}(\frac{m}{\mu_{\ell-1}})$.
If the ball $B^{+}_{G[ A_{\ell,q,j-1}]}(x_j, r_{\ell-1})$ is dense enough, we carve a ball (in the same way as in lines \ref{line:sampleR}–\ref{line:unmark} of \Cref{alg:DigraphPartition}). That is, sample $R$ (using the same distribution), and update $A$ and $S$ accordingly. We also remove all the vertices in the carved ball from the active and good sets. 
Otherwise ($|B^{+}_{G[ A_{\ell,q,j-1}]}(x_j, r_{\ell-1})|<\frac12\cdot\frac{m}{\mu_{\ell-1}}$), we do nothing and move to the next iteration.
We continue in this manner for $\kappa$ iterations (or until $\cQ_{\ell,q,j}=\emptyset$). We then move to the finalization step (in which no balls are carved). We observe:
\begin{observation}\label{obs:softSatisfy}
	Every center $x$ that was used as a center during the $\ell$'th phase satisfies the soft inequality. 
\end{observation}

In the finalization step, $A_{\ell,q+1}=A_{\ell,q,\kappa}$ is the set of active vertices.
We run \Cref{lem:cohen} again, but this time obtain estimates $\left\{\est(v,r_{\ell-1})\right\}_{v\in\cQ_{\ell,q,\kappa}}$ of the balls of radius $r_{\ell-1}$, w.r.t. the current active vertex set $A_{\ell,q+1}$.
Specifically, for every vertex $v\in\cQ_{\ell,q,\kappa}$ it holds that
$(1-\eps)\cdot|B_{G[A_{\ell,q+1}]}^{+}(v,r_{\ell-1})|\le\est(v,r_{\ell-1})\le(1+\eps)\cdot|B_{G[A_{\ell,q+1}]}^{+}(v,r_{\ell-1})|$. 
We then define 
\[
\brick{\cQ_{\ell,q+1}}=\left\{ v\in\cQ_{\ell,q,\kappa}\mid\est(v,r_{\ell-1})\ge(1-\eps)\cdot\frac{m}{\mu_{\ell-1}}\right\}~.
\]
That is, every vertex that is not dense enough (in the ball with radius $r_{\ell-1}$ now) becomes bad. 
Note that for every vertex $v\in \cQ_{\ell,q,\kappa}\setminus\cQ_{\ell,q+1}$ that became bad, it holds that $|B_{G[A_{\ell,q+1}]}^{+}(v,r_{\ell-1})|\le\frac{1}{1-\eps}\cdot\est(v,r_{\ell-1})<\frac{m}{\mu_{\ell-1}}$.
On the other hand, for every good vertex such that 
$|B_{G[A_{\ell,q+1}]}^{+}(v,r_{\ell-1})|\ge\frac{m}{\mu_{\ell-1}}$,
it holds that $\est(v,r_{\ell-1})\ge(1-\eps)\cdot|B_{G[A_{\ell,q+1}]}^{+}(v,r_{\ell-1})|\ge(1-\eps)\cdot\frac{m}{\mu_{\ell-1}}$, and thus $v$ remains good.
We conclude:

\begin{observation}\label{obs:StrictImpliesGood}
	Consider a vertex $v$ such that initially $|B_{G[A_{\ell}]}^{+}(v,r_{\ell})|\le\frac{m}{\mu_{\ell}}$, and at
	the end of round $q$ it holds that if $|B_{G[A_{\ell,q,\kappa}]}^{+}(v,r_{\ell-1})|\ge\frac{m}{\mu_{\ell-1}}$,
	then $v\in\cQ_{\ell,q+1}$.
\end{observation}

We next argue that in each round, w.h.p. the number of good vertices halves.
\begin{claim}\label{clm:GoodVerticesHalves}
	W.h.p. $|\cQ_{\ell,q+1}|\le\frac12\cdot|\cQ_{\ell,q}|$.
\end{claim}
\begin{proof}
	Consider iteration $j$, and let $\widetilde{\cQ}_{\ell,q,j-1}\subseteq \cQ_{\ell,q,j-1}$ be all the good vertices $v$ such that $|B_{G[A_{\ell,q,j-1}]}^{+}(v,r_{\ell-1})|\ge\frac12\cdot\frac{m}{\mu_{\ell-1}}$.
	If at iteration $j$ we sample a center $x_j$ from $\widetilde{\cQ}_{\ell,q,j-1}$, we call this \EMPH{progress}. In the event of progress, we carve a ball around $x_j$. In particular, at least $\frac12\cdot\frac{m}{\mu_{\ell-1}}$ edges cease to be active.
	As initially there are at most $m$ edges, in the entire round $q$ there can be at most $2\mu_{\ell-1}$ progress events.
	
	Suppose that $|\widetilde{\cQ}_{\ell,q,j-1}|\ge\frac12\cdot|\cQ_{\ell,q}|$. As we sample a center u.a.r., the probability of a progress event is at least $\frac12$.
	As there are $\mu_{\ell-1}\cdot O(\log n)$ iterations, for a large enough constant inside the $O(\cdot)$, the probability that $|\widetilde{\cQ}_{\ell,q,\kappa}|$ remains larger than $\frac12\cdot|\cQ_{\ell,q}|$ is less than $n^{-10}$ (as this would imply that in $\mu_{\ell-1}\cdot O(\log n)$ attempts, each with success probability at least $\frac12$, we had fewer than $2\mu_{\ell-1}$ progress events, which is highly unlikely; note that also $O(\mu_{\ell-1})$ rounds would suffice).
	
	It follows that $|\widetilde{\cQ}_{\ell,q,\kappa}|<\frac12\cdot|\cQ_{\ell,q}|$.
	In the finalization step, for every good vertex $v\notin \widetilde{\cQ}_{\ell,q,\kappa}$, it holds that $\est(v,r_{\ell-1})\le(1+\eps)\cdot|B_{G[A_{\ell,q,\kappa}]}^{+}(v,r_{\ell-1})|\le(1+\eps)\cdot\frac{1}{2}\cdot\frac{m}{\mu_{\ell-1}}<(1-\eps)\cdot\frac{m}{\mu_{\ell-1}}$.
	In particular, only vertices in $\widetilde{\cQ}_{\ell,q,\kappa}$ remain good. The claim follows.
\end{proof}

As we have $O(\log n)$ rounds, at the end of the phase, w.h.p., no good vertices remain. Using \Cref{obs:StrictImpliesGood}, we conclude:

\begin{corollary}\label{cor:sizeReductionEfficinet}
	At the end of the $\ell$ phase, outgoing part, for every active vertex $v\in A_{\ell-1}$, either $|B_{G[A_{\ell}]}^{+}(v,r_{\ell})|>\frac{m}{\mu_{\ell}}$, or $|B_{G[A_{\ell-1}]}^{+}(v,r_{\ell-1})|<\frac{m}{\mu_{\ell-1}}$.
\end{corollary}

Using \Cref{cor:sizeReductionEfficinet}, we are now ready to reprove \Cref{clm:PartitionIterationRDelta}.
\begin{claim}[\Cref{clm:PartitionIterationRDelta} restated]\label{clm:PartitionIterationRDeltaRestated}
	For every SCC $C$ in $G\setminus S$ which is contained in $\cR$, it holds that $\diam(C,G)\le2r_L\le\frac{\Delta}{2}$.
\end{claim}
\begin{proof}
	Denote $\cH=\left\{ v\in\cR\mid|B_{G[A_{L}]}^{+}(v,r_{L})|>\frac{m}{2}=\frac{m}{\mu_{L}}\text{ and }|B_{G[A_{L}]}^{-}(v,r_{L})|>\frac{m}{2}=\frac{m}{\mu_{L}}\right\}$.
	We argue that for every vertex $v\in \cR\setminus\cH$, the SCC of $v$ in $G\setminus S$ is $\{v\}$.
	Recall that for every $u,v\in \cH$, $d_{G[A_L]}(u,v)\le2r_L$. Thus the SCCs of $G\setminus S$ that are fully
	contained in $\cR$ are either singletons or fully contained in $\cH$. In both cases, the diameter is at most $2r_L$.
	
	Consider $v\in\cR\setminus\cH$, and suppose that $|B_{G}^{+}(v,r_L)|\leq\frac{m}{2}$ (the case $|B_{G}^{-}(v,r_L)|\leq\frac{m}{2}$ is symmetric using the ingoing part).
	By \Cref{cor:sizeReductionEfficinet}, as $|B_{G}^{+}(v,r_L)|\leq\frac{m}{\mu_L}$, necessarily
	$|B_{G[A_{L-1}]}^{+}(v,r_{L-1})|<\frac{m}{\mu_{L-1}}$.
	Note that this inequality also holds after we finish the ingoing part (as the number of vertices in the ball can only decrease).
	By induction, it follows that $|B_{G[A_0=\cR]}^{+}(v,r_{0})|=|B_{G[A_{0}]}^{+}(v,r_{0})|\leq\frac{m}{\mu_{0}}<1$.
	That is, all the outgoing edges from $v$ in $G[\cR]$ have weight greater than $r_0$.
	As all these edges are added to the cut set $S$ in \cref{line:RemoveHeavyEdges}, we conclude that in $G\setminus S$, $v$ indeed has no outgoing edges, and thus its SCC is $\{v\}$, as required.
\end{proof}

\paragraph*{Running time.}
Consider \Cref{alg:DigraphPartition} as described above.
Ignoring the process for choosing centers, the total time required to carve balls is $\tilde{O}(m)$ (as each edge can be carved only once). 
Next, we analyze the time required to choose centers. 
Focus on the $\ell$'th phase. 
There are $O(\log n)$ rounds, where in each round we execute \Cref{lem:cohen} ($\tilde{O}(m)$ time). In addition, in each round there are $O(\mu_\ell\cdot\log n)$ iterations, each taking $\tilde{O}(\frac{m}{\mu_\ell})$ time.
Overall each phase takes $\tilde{O}(m)$ time. 
There are $L$ phases, and thus the entire \Cref{alg:DigraphPartition} takes $\tilde{O}(m)$ time.
The internal computations of \Cref{alg:LaminarPartition}, as well as the construction of the DAGs $D_1,D_2$, take $\tilde{O}(m)$ time (as described above).
The depth of the recursion of the algorithm is $O(\log n)$, and thus each edge participates in at most $O(\log n)$ calls to \Cref{alg:DigraphPartition}. We conclude that overall, the running time of the algorithm is $\tilde{O}(m)$.

\section{Conclusion and Open Problems}
The main contribution of this paper is a stochastic embedding into DAGs with expected distortion $\tilde{O}(\log n)$ and sparsity $\tilde{O}(m)$.
The main underlying tool is our (implicit) construction of directed LDDs (low-diameter decompositions).
The best-known loss factor $\beta$ for directed LDDs is $O(\log n\log\log n)$ \cite{BFHL25,Li25}, whereas there is a lower bound of $\Omega(\log n)$ \cite{Bar96}.
Thus, by constructing a stochastic embedding into DAGs based on directed LDDs, one cannot hope to improve the result in this paper by more than an $O(\log\log n)$ factor.
However, this does not rule out other approaches.
Several open questions follow:
\begin{enumerate}
	\item Settling the best possible expected distortion. At present, we cannot even rule out the possibility of a stochastic embedding into DAGs with sparsity $\tilde{O}(m)$ and constant expected distortion. 
	Such a result, if possible, will not follow the steps of this paper (and \cite{AHW25}) using directed LDDs.
	On the other hand, it is not clear how to generalize the classic expander-based lower bounds for stochastic tree embeddings to the richer framework of DAGs.
	Settling the best possible expected distortion is the main open problem.
	
	\item General trade-off for DAG covers. We constructed a DAG cover with $O(\log n)$ DAGs and distortion $\tilde{O}(\log n)$.
	In contrast, for trees we have a full trade-off: one can get a tree cover with $O(n^{\nicefrac{1}{k}}\log n)$ trees and distortion $2k-1$ \cite{TZ05}.
	Can we get a similar trade-off for DAG covers?
	
	\item Ramsey-type DAG cover. Given a graph $G$, one can find a subset $M$ of $n^{1-\frac1k}$ vertices that embeds into a single tree with (worst-case) distortion $O(k)$ \cite{MN07}. 
	It would be nice to find something similar for DAGs.
	For starters, given a digraph, is it always possible to find a subset of $\frac n2$ vertices that embed into $2$ DAGs $(D_1,D_2)$ with sparsity $\tilde{O}(m)$ and worst-case distortion (naturally defined) $\tilde{O}(\log n)$?
	
	This is a very interesting question. A general trade-off for the Ramsey-type setting will also likely imply a general trade-off for DAG covers. However, it is not clear how to approach this problem using LDDs.
	Indeed, Ramsey-type embeddings of graphs into trees are based on LDDs, where one ensures that a certain set $M$ of vertices is “padded” at all scales; that is, they are never too close to the boundary of a cluster. This condition is sufficient to ensure low distortion for this set of vertices.
	We can hope to obtain a similar phenomenon for our laminar topological order. However, it would not imply low worst-case distortion for these vertices. Indeed, suppose that both $s$ and $t$ are “padded” at all scales.
	It might be that some vertices on the shortest $s$–$t$ path $\pi$ are badly cut (while being in an SCC not containing $s,t$). This could inflict significant distortion on the pair $(s,t)$.
	
	\item Remove aspect-ratio dependence. In our sparsity, we have a polylogarithmic dependence on the aspect ratio. It is likely redundant, and it would be nice to remove it. 
	
	\item Improve directed LDDs. The current upper bound for the loss factor is $O(\log n\log\log n)$ \cite{BFHL25,Li25}, while the lower bound is $\Omega(\log n)$ \cite{Bar96}. It would be nice to close this gap.
	Note that a similar gap also appears in stochastic and Ramsey-type embeddings into spanning trees \cite{AN19,ACEFN20}, and it seems to come from the same source.
	All these papers with $O(\log n \log\log n)$ upper bounds are based on the Seymour trick \cite{Seymour95}, 
	while papers with $O(\log n)$ upper bounds \cite{FRT04,Bartal04,MN07} rely on more sophisticated ideas, 
	which are not clearly applicable to (previously, spanning trees, and now) directed LDDs.
	
	\item Steiner points. This question was previously raised by \cite{AHW25}.
	In this paper (as well as in \cite{AHW25}) we constructed DAGs over the vertex set $V$.
	It is interesting to see whether allowing embedding into DAGs with Steiner points, that is, vertices not in $V$, could be helpful, leading to better upper bounds.
\end{enumerate}

{\small 
	\bibliographystyle{alphaurlinit}
	\bibliography{LSObib}}

\end{document}